\documentclass[a4paper,reqno]{amsart}

\usepackage[T1]{fontenc}

\usepackage[margin=2cm]{geometry}
\usepackage[foot]{amsaddr}
\usepackage[colorlinks, linkcolor = blue, citecolor = blue, urlcolor = blue]{hyperref}

\usepackage{amsmath,amsthm,amssymb,amsfonts}
\usepackage{ytableau}

\usepackage{mathtools}

\usepackage{thmtools}

\usepackage{mathrsfs}

\usepackage{cleveref}%[nameinlink]

\usepackage[shortlabels]{enumitem}

\usepackage{makecell}
\usepackage{multirow}
\usepackage{xcolor}
\usepackage{hhline}
\usepackage{diagbox}

\usepackage{caption}

\usepackage[group-separator={,},group-minimum-digits=4]{siunitx}

\usepackage{epigraph} 

\usepackage{graphicx}
\usepackage{tikz}
\usetikzlibrary{quantikz2}
\usetikzlibrary{backgrounds}
\usetikzlibrary{external}
\usetikzlibrary{calc}
\usetikzlibrary{shapes}
\usepackage{booktabs}

\usepackage{ytableau,varwidth}

\usepackage[backend=biber,style=alphabetic,maxnames=9,maxalphanames=5,isbn=false,backref]{biblatex}

\newtheorem*{theorem*}{Theorem}
\newtheorem{theorem}{Theorem}[section] % number within sections
\newtheorem{proposition}[theorem]{Proposition} % use the same numbering for all environments
\newtheorem{lemma}[theorem]{Lemma}

\newtheorem{corollary}[theorem]{Corollary}
\newtheorem{remark}[theorem]{Remark}

\crefname{lemma}{Lemma}{Lemmas}
\crefname{definition}{Definition}{Definitions}
\crefname{theorem}{Theorem}{Theorems}
\crefname{conjecture}{Conjecture}{Conjectures}
\crefname{section}{Section}{Sections}
\crefname{claim}{Claim}{Claims}
\crefname{appendix}{Appendix}{Appendices}
\crefname{figure}{Fig.}{Figs.}
\crefname{table}{Table}{Tables}
\crefname{proposition}{Proposition}{Propositions}
\crefname{corollary}{Corollary}{Corollaries}
\crefname{example}{Example}{Examples}
\crefname{remark}{Remark}{Remarks}

\providecommand\given{}
\newcommand\SetSymbol[1][]{%
    \nonscript\:#1\vert
    \allowbreak
    \nonscript\:
    \mathopen{}}
\DeclarePairedDelimiterX\Set[1]\{\}{%
    \renewcommand\given{\SetSymbol[\delimsize]}
    #1
}

\DeclarePairedDelimiter{\set}{\lbrace}{\rbrace}
\DeclarePairedDelimiter{\abs}{\lvert}{\rvert}
\DeclarePairedDelimiter{\norm}{\lVert}{\rVert}

\DeclarePairedDelimiter{\of}{\lparen}{\rparen}
\DeclarePairedDelimiter{\sof}{\lbrack}{\rbrack}

\newcommand{\defeq}{\vcentcolon=}
\newcommand{\eqdef}{=\vcentcolon}
\renewcommand{\leq}{\leqslant}
\renewcommand{\geq}{\geqslant}

\renewcommand{\bra}[1]{\langle{#1}\rvert}
\renewcommand{\ket}[1]{\lvert{#1}\rangle}
\renewcommand{\braket}[2]{\langle{#1}|{#2}\rangle}
\newcommand{\ketbra}[2]{\ket{#1}\bra{#2}}
\renewcommand{\proj}[1]{\ketbra{#1}{#1}}

\newcommand{\mx}[1]{\begin{pmatrix}#1\end{pmatrix}}
\newcommand{\bmx}[1]{\begin{bmatrix}#1\end{bmatrix}}
\newcommand{\smx}[1]{\begin{psmallmatrix}#1\end{psmallmatrix}} % small matrix

\newcommand{\ct}{^{\dagger}}
\newcommand{\tp}{^{\mathsf{T}}}
\newcommand{\ptp}{^{\mathsf{\Gamma}}}

\newcommand{\x}{\otimes}
\newcommand{\xp}[1]{^{\otimes #1}}
\newcommand{\ctxp}[1]{^{\dagger\otimes #1}}

\newcommand{\E}{\mathcal{E}} %matrix unit
\newcommand{\Epq}{\mathcal{E}} %set of all matrix units

\newcommand{\C}{\mathbb{C}} % complex numbers
\newcommand{\R}{\mathbb{R}} % real numbers
\newcommand{\cont}{\mathrm{cont}} % content of a diagram
\newcommand{\wcont}{\mathrm{wcont}} % walled content of a mixed diagram
\newcommand{\B}{\mathcal{B}} % Brauer algebra
\newcommand{\A}{\mathcal{A}} % algebra of partially transposed permutation matrices
\newcommand{\X}{\mathcal{X}} % Gelfand-Tsetlin subalgebra
\newcommand{\poly}{\mathrm{poly}}

\newcommand{\UschinvPQ}{U^\dagger_{\mathrm{Sch}}(p,q)} % inverse of mixed Schur transform
\newcommand{\UschPQ}{U_{\mathrm{Sch}}(p,q)} % mixed Schur transform
\newcommand{\Usch}{U_{\mathrm{Sch}}} % Schur transform

\newcommand{\CGqcTilde}[2]{\widetilde{\mathrm{CG}}^{#1}_{#2}} % Modified  (tilde) dual CG
\newcommand{\CGqc}[2]{\mathrm{CG}^{#1}_{#2}} % original and dual CG transform as quantum circuit operator

\newcommand{\rwpm}[4]{\of*{z^\pm}^{#1,#2}_{#3,#4}}
\newcommand{\rwsgn}[5]{\of*{z^{#1}}^{#2,#3}_{#4,#5}}

\newcommand{\0}{\varnothing}
\let\S\relax
\DeclareMathOperator{\S}{S} % symmetric group
\DeclareMathOperator{\CS}{\mathbb{C}S} % symmetric group algebra
\DeclareMathOperator{\End}{End} % endomorphisms
\DeclareMathOperator{\Tr}{Tr} % trace
\DeclareMathOperator{\spn}{span} % span
\DeclareMathOperator{\CG}{CG} % Clebsch--Gordan

\newcommand{\M}[1]{\mathcal{M}(#1)} % Set of Paths which differ in the middle and have a mobile element
\newcommand{\GT}{\mathrm{GT}} % Gelfand-Tsetlin patterns
\newcommand{\SYT}{\mathrm{SYT}} % standard Young tableaux
\newcommand{\SSYT}{\mathrm{SSYT}} % semi-standard Young tableaux
\newcommand{\AC}{\mathrm{AC}} % addable cells
\newcommand{\RC}{\mathrm{RC}} % removeable cells
\newcommand{\Irr}[1]{\widehat#1} % set of irreducible modules
\newcommand{\Paths}{\mathrm{Paths}} % set of paths
\newcommand{\Res}{\mathrm{Res}} % Restriction functor
\newcommand{\loops}{\mathrm{loops}} % number of loops
\newcommand{\SWAP}{\mathrm{SWAP}} % SWAP gate

\newcommand{\U}[1]{\mathrm{U}_{#1}} % unitary group
\newcommand{\IrrU}[1]{\widehat{\mathrm{U}}_{#1}} % unitary group irreps

\newcommand{\m}{\mathbf{m}} % row of a Gelfand-Tsetlin pattern
\newcommand{\n}{\mathbf{n}} % row of a Gelfand-Tsetlin pattern
\newcommand{\pt}{\mathbin{\vdash}} % partitions
\newcommand{\e}{\varepsilon}

\newcommand{\squb}{\sqsubseteq}

\newcommand{\w}{0.5cm}

\newcommand{\bx}[3]{
  \draw[fill = white] #3 (#1*\w-\w/2,-#2*\w-\w/2) rectangle (#1*\w+\w/2,-#2*\w+\w/2);
}
\newcommand{\br}[1]{\of*{#1}}

\newcommand{\yd}[2][0.4]{%
  \begin{tikzpicture}[scale = #1, baseline={([yshift=-0.6ex]current bounding box.center)}]
    \foreach \li [count = \y] in {#2} {
      \foreach \x in {1,...,\li} {
        \bx{\x}{\y}{}
      }
    }
  \end{tikzpicture}
}

\newcommand\restr[2]{{% we make the whole thing an ordinary symbol
  \left.\kern-\nulldelimiterspace % automatically resize the bar with \right
  #1 % the function
  \vphantom{\big|} % pretend it's a little taller at normal size
  \right|_{#2} % this is the delimiter
  }}

\newcommand{\HDots}{\ \ldots\ } % horizontal dots
\newcommand{\qb}{\setwiretype{b}} % turn the wire into a bundle
\newcommand{\qq}{\setwiretype{q}} % turn the wire into a bundle
\newcommand{\gateCyc}{\gate[wires=6,nwires=3]{\pi^\dagger}}
\newcommand{\gateSchurPone}{\gate[wires=8,nwires=4,label style={yshift=0.3cm}]{U_{\mathrm{Sch(p,1)}}}}
\newcommand{\gateSchurInvPone}{\gate[wires=8,nwires=4]{U_{\mathrm{Sch(p,1)}}^\dagger}}
\newcommand{\gateW}{\gate{W_{\lambda}}}
\newcommand{\gateWinv}{\gate{W_{\lambda}^\dagger}}
\newcommand{\gateR}[1]{\gate{\omega_{p+1}^{#1}}}
\newcommand{\gateQFT}{\gate{\mathrm{QFT}_{p+1}}}
\newcommand{\gateQFTinv}{\gate{\mathrm{QFT}_{p+1}^\dagger}}

\newcommand{\gateCGpm}[1]{\gate[wires=3]{\CGqc{\pm}{#1}}}
\newcommand{\gateCGtilde}[2]{\gate[wires=3]{\CGqcTilde{#1}{{#2}}}}
\newcommand{\gateCpm}[1]{\gate[wires=3]{C^{\pm}_{#1}}}

\newcommand{\BrauerTikZStyle}{\tikzset{
  every path/.style = {semithick}, % edge thickness
  looseness = 0.7, % bendy wires
  dot/.style = {shift only, radius = 1.3pt, fill = black}, % dot size
  l/.style = {out = -90, in =  90}, % vertical lines
  u/.style = {out = -90, in = -90}, % horizontal u-shaped lines
  n/.style = {out =  90, in =  90}  % horizontal n-shaped lines
}}

\title[Gelfand--Tsetlin basis for partially transposed permutations]{Gelfand--Tsetlin basis for partially transposed permutations,\\with applications to quantum information}

\author{Dmitry Grinko\textsuperscript{1}\hspace{-.25em}}
\email{dmitry.grinko@cwi.nl}

\author{Adam Burchardt\textsuperscript{1}\hspace{-.25em}}
\address{\textsuperscript{1}Institute for Logic, Language, and Computation, University of Amsterdam and QuSoft, Amsterdam, The Netherlands}
\email{adam.burchardt@cwi.nl}

\author{Maris Ozols\textsuperscript{1,2}\hspace{-.25em}}
\address{\textsuperscript{2}Korteweg-de Vries Institute for Mathematics and Institute for Theoretical Physics, University of Amsterdam, The Netherlands}
\email{marozols@gmail.com}

\begin{document}

\begin{abstract}
We study representation theory of the partially transposed permutation matrix algebra, a matrix representation of the diagrammatic walled Brauer algebra.
This algebra plays a prominent role in mixed Schur--Weyl duality that appears in various contexts in quantum information.
Our main technical result is an explicit formula for the action of the walled Brauer algebra generators in the Gelfand--Tsetlin basis.
It generalizes the well-known Gelfand--Tsetlin basis for the symmetric group (also known as Young's orthogonal form or Young--Yamanouchi basis).

We provide two applications of our result to quantum information.
First, we show how to simplify semidefinite optimization problems over unitary-equivariant quantum channels by performing a symmetry reduction.
Second, we derive an efficient quantum circuit for implementing the optimal port-based quantum teleportation protocol, exponentially improving the known trivial construction.
As a consequence, this also exponentially improves the known lower bound for the amount of entanglement needed to implement unitaries non-locally.

Both applications require a generalization of quantum Schur transform to tensors of mixed unitary symmetry.
We develop an efficient quantum circuit for this mixed quantum Schur transform and provide a matrix product state representation of its basis vectors.
For constant local dimension, this yields an efficient classical algorithm for computing any entry of the mixed quantum Schur transform unitary.
\end{abstract}

\maketitle

\tableofcontents

%%%%%%%%%%%%%%%%%%%%%%%%%%%%%%%%%%%%%%%%%%%%%%%%%%%%%%%%%%%%%%%%%%%%%%%%%%%%%%%%%%%%%%%%%%%%%%%%%%%%%%%%%%%%%%%%%%%%%%%%%%%%%%%%%%%%%%%%%%%%%%%%%%%%%%%%%%%%%%%
%%%%%%%%%%%%%%%%%%%%%%%%%%%%%%%%%%%%%%%%%%%%%%%%%%%%%%%%%%%%%%%%%%%%%%%%%%%%%%%%%%%%%%%%%%%%%%%%%%%%%%%%%%%%%%%%%%%%%%%%%%%%%%%%%%%%%%%%%%%%%%%%%%%%%%%%%%%%%%%

\section{Introduction}\label{sec:Introduction}

\subsection{Background}

Symmetry plays a fundamental role in physics and mathematics, and is a common problem-solving technique in both fields.
The unitary group of symmetries is particularly important in quantum mechanics and particle physics, since it captures the properties of the basic building blocks of our universe.
Unitary symmetry is closely intertwined with permutational symmetry.
Indeed, permuting identical particles or simultaneously applying the same unitary rotation on each of them are two commuting actions.
This observaiont is captured by \emph{Schur--Weyl duality}, which has become an important tool in quantum information, quantum algorithms, and quantum many-body physics.

The simplest instance of Schur--Weyl duality is for two qubits.
Let $\ket{\psi^-} \defeq (\ket{01} - \ket{10}) / \sqrt{2}$ denote the maximally entangled \emph{singlet state}, which is the unique anti-symmetric state on two qubits.
It has the property that
$(U \x U) \ket{\psi^-} = \det(U) \ket{\psi^-}$ for any $U \in \U{2}$.
In addition,
$\SWAP \ket{\psi^-} = - \ket{\psi^-}$ where $\SWAP \in \U{4}$ exchanges the two qubits.
The three-dimensional symmetric subspace orthogonal to $\ket{\psi^-}$ is similarly invariant under the actions of both $U \x U$ and $\SWAP$.
Hence, we can decompose $(\C^2)\xp{2}$ into mutually orthogonal subspaces invariant under the commuting unitary and permutation actions.
Schur--Weyl duality generalizes this observation to $(\C^d)\xp{p}$ for any local dimension $d$ and number of systems $p$.

Schur--Weyl duality is particularly useful in quantum information where one often needs to deal with many identical copies of a quantum state or to apply the same unitary to many systems in parallel.
Its algorithmic manifestation, quantum Schur transform, can be efficiently implemented
\cite{HarrowThesis,bch2006quantumschur,kirby2018practical,krovi2019quantumschur}
and has many applications \cite{WrightThesis,HarrowThesis},
such as quantum spectrum \cite{SpectrumEstimation} and entropy estimation \cite{EntropyEstimation}, quantum state tomography \cite{keyl2006quantumstateestimation,
haah2017optimaltomography,
o2016efficient,o2016efficient2}, and quantum majority vote \cite{QuantumMajority}.

The main focus of our paper is a variant of Schur--Weyl duality, known as \emph{mixed Schur--Weyl duality}, which is equally important in quantum information but has received less attention due to its more complicated nature \cite{grinko2022linear}.
The simplest instance of mixed Schur--Weyl duality is also for two qubits.
Here, instead of the singlet state $\ket{\psi^-}$, we single out the canonical maximally entangled state
$\ket{\phi^+} \defeq (\ket{00} + \ket{11}) / \sqrt{2}$.
This state is invariant under a slightly different unitary action, namely
$(U \x \bar{U}) \ket{\phi^+} = \ket{\phi^+}$
for any $U \in \U{2}$.
In addition,
$\SWAP\ptp \ket{\phi^+} = 2 \ket{\phi^+}$
where $\SWAP\ptp = 2 \proj{\phi^+}$ denotes the partial transpose of $\SWAP$.
Similarly, the three-dimensional orthogonal complement of $\ket{\phi^+}$ is also invariant under the action of both $U \x \bar{U}$ and $\SWAP\ptp$.
Mixed Schur--Weyl duality generalizes this observation by partitioning $(\C^d)^{p+q}$ into subspaces that are invariant under the unitary action $U\xp{p} \x \bar{U}\xp{q}$ and the matrix algebra $\A^d_{p,q}$ of partially transposed permutations that are transposed only on the last $q$ systems \cite{koike1989decomposition,bchlls,halverson1996characters,nikitin2007centralizer}.
In particular, $[\A^d_{p,q}, U\xp{p} \x \bar{U}\xp{q}] = 0$ for all $U \in \U{d}$.
The usual Schur--Weyl duality corresponds to the special case when either $p=0$ or $q=0$.

%%%%%%%%%%%%%%%%%%%%%%%%%%%%%%%%%
% Curly brace
\newcommand{\curlybrace}[3]{
  \def\r{0.2} % curly bracket radius
  \def\e{0.1} % overhang
  \draw[radius = \r] (#1-#3-\e,#2-\r)
    arc[start angle = 180, end angle =  90] -- (#1-\r,#2) arc[start angle = -90, end angle = 0]
    arc[start angle = 180, end angle = 270] -- (#1+#3-\r+\e,#2)
    arc[start angle =  90, end angle =   0];
}
%%%%%%%%%%%%%%%%%%%%%%%%%%%%%%%%%
% A random Brauer diagram
\newcommand{\diagRand}[2][0.5]{\,
  \tikz[scale = #1, baseline = (O)]{
    \BrauerTikZStyle
    \path (1.5,0.35) coordinate (O);
    % Nodes
    \foreach \i in {1,...,5} {
      \fill (\i,1) circle [dot] coordinate (A\i);
      \fill (\i,0) circle [dot] coordinate (B\i);
    }
    % Extra stuff
    #2
    % Edges
    \draw (A1) to [l] (B1);
    \draw (A2) to [l] (B3);
    \draw (A4) to [l] (B5);
    \draw (A3) to [u] (A5);
    \draw (B2) to [n] (B4);
    % Wall
    \draw [dashed] (3.5,1.2) -- (3.5,-0.2);
  }\,
}
% Diagram labels
\newcommand{\Labels}{
    \foreach \i in {1,...,5} {
        \path (A\i)+(0, 0.5) node {$x_\i$};
        \path (B\i)+(0,-0.5) node {$y_\i$};
    }
}
%%%%%%%%%%%%%%%%%%%%%%%%%%%%%%%%%
\newcommand{\diagSWAP}{\,
  \tikz[scale = 0.5, baseline = (O)]{
    \BrauerTikZStyle
    \tikzset{looseness = 1}
    \path (1.5,0.35) coordinate (O);
    \foreach \i in {1,2} {
      \fill (\i,1) circle [dot] coordinate (A\i);
      \fill (\i,0) circle [dot] coordinate (B\i);
    }
    \draw (A1) to [l] (B2);
    \draw (A2) to [l] (B1);
  }\,
}
%%%%%%%%%%%%%%%%%%%%%%%%%%%%%%%%%
\newcommand{\diagContract}{\,
  \tikz[scale = 0.5, baseline = (O)]{
    \BrauerTikZStyle
    \tikzset{looseness = 1}
    \path (1.5,0.35) coordinate (O);
    \foreach \i in {1,2} {
      \fill (\i,1) circle [dot] coordinate (A\i);
      \fill (\i,0) circle [dot] coordinate (B\i);
    }
    \draw (A1) to [u] (A2);
    \draw (B1) to [n] (B2);
    \draw [dashed] (1.5,1.2) -- (1.5,-0.2);
  }\,
}
%%%%%%%%%%%%%%%%%%%%%%%%%%%%%%%%%

Partially transposed permutations can be easily visualized as diagrams.
If $\sigma \in \S_{p+q}$ is a permutation on $p+q$ objects then its partial transpose $\sigma\ptp$ is obtained by exchanging the last $q$ inputs and outputs of $\sigma$:
\begin{equation}
  \Bigg(\;
  \begin{tikzpicture}[baseline = 0.4cm]
    \BrauerTikZStyle
    % p and q
    \curlybrace{2.0}{1.4}{1.0} \node at (2.0,1.9) {$p=3$};
    \curlybrace{4.5}{1.4}{0.5} \node at (4.5,1.9) {$q=2$};
    % Permutation diagram
    \begin{scope}
      % Nodes
      \foreach \i in {1,...,5} {
        \fill (\i,1) circle [dot] coordinate (A\i);
        \fill (\i,0) circle [dot] coordinate (B\i);
      }
      % Edges
      \draw (A1) to [l] (B1);
      \draw (A2) to [l] (B3);
      \draw (A5) to [l] (B4);
      \draw (A3) to [l] (B5);
      \draw (A4) to [l] (B2);
      % Wall
      \draw [dashed] (3.5,1.2) -- (3.5,-0.2);
    \end{scope}
  \end{tikzpicture}
  \;\Bigg)\ptp \quad = \quad
  \begin{tikzpicture}[baseline = 0.5cm]
    \BrauerTikZStyle
    \begin{scope}[xshift = 6.6cm]
      % Nodes
      \foreach \i in {1,...,5} {
        \fill (\i,1) circle [dot] coordinate (A\i);
        \fill (\i,0) circle [dot] coordinate (B\i);
      }
      % Edges
      \draw (A1) to [l] (B1);
      \draw (A2) to [l] (B3);
      \draw (A4) to [l] (B5);
      \draw (A3) to [u] (A5);
      \draw (B2) to [n] (B4);
      % Wall
      \draw [dashed] (3.5,1.2) -- (3.5,-0.2);
    \end{scope}
  \end{tikzpicture}
  \;\;\;.
  \label{eq:diagram transpose}
\end{equation}
While $\sigma\ptp$ is no longer a permutation, we can still multiply such diagrams by concatenating them in the same way as permutations (in case closed loops appear in this process, we remove them and multiply the diagram by $d^l$ where $l$ is the number of loops).
The set of all partially transposed permutation diagrams under this composition forms the \emph{walled Brauer algebra}
$\B^d_{p,q} \defeq \spn_\C \set{\sigma\ptp : \sigma \in \S_{p+q}}$,
a diagrammatic algebra of dimension $\dim(\B^d_{p,q}) = (p+q)!$.
In contrast, in quantum information we encounter only its matrix representation
\begin{equation}
    \A_{p,q}^d
    \defeq \psi^d_{p,q}(\B_{p,q}^d),
    \label{eq:A=psiB}
\end{equation}
where
$\psi^d_{p,q} \colon \B^d_{p,q} \to \End\of[\big]{(\C^d)^{p+q}}$
is a linear map from diagrams to matrices on $p+q$ qudits of dimension $d$.
For example, the matrix representation of the diagram from \cref{eq:diagram transpose} has the following standard basis entries:
\begin{align}
  \bra{x_1, \dotsc, x_{5}}
  \, \psi^d_{3,2} \of[\Big]{\diagRand{}} \,
  \ket{y_1, \dotsc, y_{5}}
  = \!\!\!
    \diagRand{\Labels}
    \!\!\!
  = \delta_{x_1,y_1}
    \delta_{x_2,y_3}
    \delta_{x_3,x_5}
    \delta_{x_4,y_5}
    \delta_{y_2,y_4},
  \label{eq:psi example}
\end{align}
for any $x_1, \dotsc, x_5 \in [d]$ and $y_1, \dotsc, y_5 \in [d]$.
In particular, note that
\begin{equation}
    \SWAP\ptp
    = \of*{\diagSWAP}\ptp
    = \diagContract
    = d \; \proj{\phi^+}
    \label{eq:contraction}
\end{equation}
is proportional to the canonical maximally entangled state on two qudits (see also \cref{tab:TranspositionVSCOntraction}).

This interplay between permutations and entanglement is why the matrix algebra $\A^d_{p,q}$ and mixed Schur--Weyl duality is so relevant to quantum information.
It appears in a variety of contexts,
particularly in scenarios with multiple input and output systems
such as
quantum state purification \cite{Purification}
and cloning
\cite{Cloning,Cloning2,AsymmetricCloning},
port-based
\cite{mozrzymas2018optimal,studzinski2017port,leditzky2020optimality,christandl2021asymptotic,studzinski2021degradation}
and multi-port-based teleportation
\cite{kopszak2020multiport,studzinski2020efficient,mozrzymas2021optimal},
and quantum algorithms \cite{Majority}.
This symmetry also occurs in situations that involve the partial transpose on several systems,
such as in entanglement detection
\cite{EntanglementDetection,balanzójuandó2021positive},
universality of qudit gate sets
\cite{sawicki2021check,dulian2022matrix,słowik2022calculable},
and $\U{d}$-equivariant quantum circuits \cite{Circuits,zheng2021speeding}.
It is also relevant in high-energy physics
\cite{Branes,Candu}.

%%%%%%%%%%%%%%%%%%%%%%%%%%%%%%%%%%%%%%%%%%%%%%%%%%%%%%%%%%%%%%%%%%%%%%%%%%%%%%%%%%%%%%%%%%%%%%%%%%%%%%%%%%%%%%%%%%%%%%%%%%%%%%%%%%%%%%%%%%%%%%%%%%%%%%%%%%%%%%%

\subsection{Historical context}

The walled Brauer algebra $\B^d_{p,q}$ is a restricted version of the full Brauer algebra \cite{Brauer}, and a prominent example of a diagram algebra, which has been widely studied \cite{turaev1989operator,koike1989decomposition,bchlls,Benkart,nikitin2007centralizer,Bulgakova} (see \cite{Panorama} for a survey on Brauer and other diagram algebras).
Mixed Schur--Weyl duality was established in \cite{koike1989decomposition,bchlls}, where the matrix algebra $\A^d_{p,q}$ of partially transposed permutations was first introduced.
Motivated by applications to quantum information, this algebra was subsequently studied in \cite{mozrzymas2021optimal,studzinski2020efficient}.
In particular, the $q = 1$ case was considered in
\cite{zhang2007permutation,studzinski2013commutant,mozrzymas2014structure,mozrzymas2018simplified}.
The characters of the walled Brauer algebra were first derived by Halverson \cite{halverson1996characters} (see also \cite{nikitin2007centralizer}).

The representation theory of walled Brauer algebras has been strongly influenced by representation theory of symmetric group algebras, which correspond to the special case when either $p = 0$ or $q = 0$.
The representation theory of the symmetric group is widely used and has a long history \cite{Sagan,rutherford2013substitutional,ceccherini2010representation,howe2022invitation}.
In particular, there are three commonly used forms for irreducible representations of the symmetric group
\cite{rutherford2013substitutional}
(see \cref{tab:irrep forms} for examples):
\begin{enumerate}
    \item \emph{Young's natural form} provides invertible matrices with integer entries (i.e., in $\mathbb{Z}$),
    \item \emph{Young's seminormal form} provides invertible matrices with rational entries (i.e., in $\mathbb{Q}$),
    \item \emph{Young's orthogonal form} (also known as \emph{Young--Yamanouchi basis}) provides real orthogonal matrices whose entries are square roots of rational numbers (i.e., in $\pm \sqrt{\mathbb{Q}_{\geq 0}}$).
\end{enumerate}
The basis change between Young's seminormal and orthogonal forms is diagonal and corresponds to normalization, while the basis change between seminormal and natural forms is triangular \cite{armon2021transition}.

%%%%%%%%%%%%%%%%%%%%%%%%%%%%%%%%%%%%%%%%%%%%%%%%%%%%%%%%%%%%%%%%%%%%%%%%

\begin{table}
%--------------------------------------
% Wrapper for any TikZ figure
\newcommand{\Gr}[2][-0.1cm]{
\begin{tikzpicture}[semithick, baseline = #1]
  % Dimensions
  \def\X{0.50}
  \def\Y{0.45}
  #2
\end{tikzpicture}}
\newcommand{\gr}[1]{\Gr{#1}}
%--------------------------------------
% Round wires
\newcommand{\round}{.. controls +(0.3,0) and +(-0.3,0) ..}
%--------------------------------------
% Permutations
\newcommand{\perma}{
  \draw (0, \Y)   --   (\X, \Y);
  \draw (0,  0)   --   (\X,  0);
  \draw (0,-\Y)   --   (\X,-\Y);
}
\newcommand{\permb}{
  \draw (0, \Y) \round (\X,  0);
  \draw (0,  0) \round (\X,-\Y);
  \draw (0,-\Y) \round (\X, \Y);
}
\newcommand{\permc}{
  \draw (0, \Y) \round (\X,-\Y);
  \draw (0,  0) \round (\X, \Y);
  \draw (0,-\Y) \round (\X,  0);
}
\newcommand{\permd}{
  \draw (0, \Y) \round (\X,  0);
  \draw (0,  0) \round (\X, \Y);
  \draw (0,-\Y)   --   (\X,-\Y);
}
\newcommand{\perme}{
  \draw (0, \Y) \round (\X,-\Y);
  \draw (0,  0)   --   (\X,  0);
  \draw (0,-\Y) \round (\X, \Y);
}
\newcommand{\permf}{
  \draw (0, \Y)   --   (\X, \Y);
  \draw (0,  0) \round (\X,-\Y);
  \draw (0,-\Y) \round (\X,  0);
}
%--------------------------------------
\begin{equation*}
  \newcommand{\vspc}[1]{\rule{0pt}{#1}}
  \begin{array}{c|cccccc}
    \vspc{10pt}
    \pi \in \S_3 & e & \sigma_2 \sigma_1 & \sigma_1 \sigma_2 & \sigma_1 & \sigma_1 \sigma_2 \sigma_1 & \sigma_2 \\[4pt] \hline
    \vspc{22pt}
    \text{Diagram of $\pi$} & \gr{\perma} & \gr{\permb} & \gr{\permc} & \gr{\permd} & \gr{\perme} & \gr{\permf} \\[12pt] \hline
    \vspc{20pt}
    R_{\text{nat}}(\pi) &
      \mx{1&0\\0&1} &
      \mx{-1&-1\\ 1& 0} &
      \mx{ 0& 1\\-1&-1} &
      \mx{ 1& 0\\-1&-1} &
      \mx{-1&-1\\ 0& 1} &
      \mx{ 0& 1\\ 1& 0} \\[10pt]
    R_{\text{semi}}(\pi) &
                  \mx{1&0\\0&1} &
     \dfrac{1}{2} \mx{-1&-3\\ 1&-1} &
     \dfrac{1}{2} \mx{-1& 3\\-1&-1} &
                  \mx{ 1& 0\\ 0&-1} &
     \dfrac{1}{2} \mx{-1&-3\\-1& 1} &
     \dfrac{1}{2} \mx{-1& 3\\ 1& 1} \\[10pt]
    R_{\text{orth}}(\pi) &
                  \mx{1&0\\0&1} &
     \dfrac{1}{2} \mx{-1& -\sqrt{3} \\  \sqrt{3}&-1} &
     \dfrac{1}{2} \mx{-1&  \sqrt{3} \\ -\sqrt{3}&-1} &
                  \mx{1&0\\0&-1} &
     \dfrac{1}{2} \mx{-1& -\sqrt{3} \\ -\sqrt{3}& 1} &
     \dfrac{1}{2} \mx{-1&  \sqrt{3} \\  \sqrt{3}& 1} \\[10pt]
  \end{array}
\end{equation*}
\caption{\label{tab:irrep forms}Young's natural, seminormal, and orthogonal form (denoted by $R_{\text{nat}}$, $R_{\text{semi}}$, and $R_{\text{orth}}$, respectively) for the two-dimensional irrep $\yd{2,1}$ of the symmetric group $\S_3$.}
\end{table}

%%%%%%%%%%%%%%%%%%%%%%%%%%%%%%%%%%%%%%%%%%%%%%%%%%%%%%%%%%%%%%%%%%%%%%%%

In the context of quantum information, Young's orthogonal form is by far the most useful since it provides unitary matrices that can readily be used as operations in a quantum computer.
In particular, irreps of this form are produced by the quantum Schur transform \cite{HarrowThesis,Berg}.
Our goal is to extend Young's orthogonal form from the symmetric group to the matrix algebra $\A^d_{p,q}$, and to derive the corresponding mixed Schur transform that decomposes $\A^d_{p,q}$ into irreps of this form.

Since our approach is based on a very general and well-established strategy that involves decomposing an algebra (or a group) into a chain of subalgebras (or subgroups), we will refer to the resulting basis as \emph{Gelfand--Tsetlin basis}, a term that is commonly used for these types of constructions.
Confusingly, this means that the basis we obtain can be referred to by three different names:
``Young's orthogonal form'',
``Young--Yamanouchi basis'', and
``Gelfand--Tsetlin basis''.
We will use only the latter term throughout the rest of this paper.

%%%%%%%%%%%%%%%%%%%%%%%%%%%%%%%%%%%%%%%%%%%%%%%%%%%%%%%%%%%%%%%%%%%%%%%%%%%%%%%%%%%%%%%%%%%%%%%%%%%%%%%%%%%%%%%%%%%%%%%%%%%%%%%%%%%%%%%%%%%%%%%%%%%%%%%%%%%%%%%

\subsection{Strategy}

While the diagrammatic walled Brauer algebra $\B^d_{p,q}$ is well studied, its matrix representation $\A^d_{p,q}$ has received much less attention.
A key difficulty in studying $\A^d_{p,q}$ is that
\begin{equation}
    \A_{p,q}^d
    \cong \B_{p,q}^d \,/\, \ker(\psi_{p,q}^d)
\end{equation}
according to \cref{eq:A=psiB}.
While the linear map $\psi_{p,q}^d \colon \B^d_{p,q} \to \End\of[\big]{(\C^d)^{p+q}}$ that turns diagrams into matrices has a simple description, see \cref{eq:psi example,eq:Brauer action}, its kernel $\ker(\psi_{p,q}^d)$ is non-trivial when $d < p + q$ and its basis has a rather complicated description.
In addition, since $\B_{p,q}^d$ is not semisimple when $d < p + q - 1$ \cite{cox2008blocks}, it can be difficult to obtain results for $\B_{p,q}^d$ and then transfer them to $\A_{p,q}^d$ (which is always semisimple).

Our strategy hinges on the close connection between the walled Brauer algebra $\B^d_{p,q}$ and the group algebra $\C(\S_p \times \S_q)$ corresponding to its two symmetric subgroups.
We will make use of the fact that $\B^d_{p,q}$ is generated by diagrams $\sigma_1, \dotsc, \sigma_{p+q-1}$ \cite{nikitin2007centralizer} shown in \cref{eq:generators}, where $\sigma_i$ ($i \neq p$) are \emph{transpositions} of consecutive systems $i$ and $i+1$ which generate $\S_p \times \S_q$, while the remaining generator $\sigma_p$ \emph{contracts} systems $p$ and $p+1$, e.g., see \cref{eq:contraction} (see \cref{def:Brauer} for more details).
The matrix algebra $\A_{p,q}^d$ is generated by $\psi^d_{p,q}(\sigma_1), \dotsc, \psi^d_{p,q}(\sigma_{p+q-1})$.

\subsection{Summary of our results}

Our main technical result is \cref{thm:main}, which provides an explicit construction of all irreducible representations of $\A^d_{p,q}$, the matrix algebra of partially transposed permutations on $p+q$ qudits.
We provide a formula that allows to evaluate the irrep matrix entries on each of the generators, which by homomorphism and linearity fully determines the irrep on the rest of the algebra.

An important feature of our construction is that it provides irrep matrix entries in the Gelfand--Tsetlin basis.
This basis has a recursive definition and hence is automatically adapted to a natural sequence of subalgebras obtained by including the generators $\sigma_i$ one by one.
This guarantees that the generators have a particularly sparse representation and gives a conceptually simple way to pinpoint their non-zero entries and describe their action.
In addition, our basis coincides with the basis produced by the mixed quantum Schur transform.
The conceptual simplicity of the Gelfand--Tsetlin basis together with its operational connection with the mixed quantum Schur transform is precisely what will allow us in \cref{sec:PBT} to derive an efficient quantum circuit for implementing the optimal port-based quantum teleportation protocol.

\begin{theorem*}[Qualitative statement of \cref{thm:main}]
When evaluated on one of the generators, any irreducible representation of the matrix algebra $\A^d_{p,q}$ has the following form:
\begin{itemize}
    \item Transpositions $\sigma_i$ ($i \neq p$) are represented by a direct sum of $1 \times 1$ and $2 \times 2$ blocks, where each $1 \times 1$ block is equal to $\pm 1$, while for each $2 \times 2$ block there is a constant $r \in \mathbb{Z}$ such that the block is equal to
    \begin{equation}
        \mx{
          \frac{1}{r} & \sqrt{1 - \frac{1}{r^2}} \\
          \sqrt{1 - \frac{1}{r^2}} & - \frac{1}{r}
        },
    \end{equation}
    which is an orthogonal reflection.
    The exact signs and values of $r$ can be inferred from Young's orthogonal form for symmetric groups.
    \item The contraction $\sigma_p$ is represented by a direct sum of rank-$1$ matrices with eigenvalue $d$.
\end{itemize}
When $p = 0$ or $q = 0$, our formula reduces to the well-known Young's orthogonal form for the symmetric group.
\end{theorem*}

A concrete example of how all irreps of $\A^3_{3,2}$ look like in the Gelfand--Tsetlin basis is provided in \cref{tab:A32}.
The main idea of the proof is that, thanks to this being the Gelfand--Tsetlin basis, we already know from Young's orthogonal form how all generators (except for the contraction $\sigma_p$) are supposed to act.
We made an educated guess for the matrix representation of $\sigma_p$ and verified that it indeed works.
This requires checking that the matrix representations of all generators satisfy the walled Brauer algebra relations stated in \cref{def:Brauer}.

%%%%%%%%%%%%%%%%%%%%%%%

As our second technical result, we investigate the mixed quantum Schur transform which can be used both to block diagonalize the matrix algebra $\A^d_{p,q}$ as well as to prepare the Gelfand--Tsetlin basis vectors.
More specifically, in \cref{thm:MPS} of \cref{sec:MPS} we show that the the rows of the mixed quantum Schur transform (or the Schur basis states) admit a matrix product state representation with bond dimension $(p+q)^{O(d^2)}$.
This means that, for a constant local dimension $d$, we can compute the matrix entries of mixed Schur transform in polynomial time.
In addition, \cref{thm:mixed_schur} of \cref{sec:SchTransQuantum} we provide an efficient quantum circuit with $\poly(p+q,d,\log 1/\e)$ gates for computing the mixed Schur transform on a quantum computer.
In particular, this transform achieves the Gelfand--Tsetlin basis from \cref{thm:main} on the walled Brauer algebra register.

%%%%%%%%%%%%%%%%%%%%%%%%%%%%%%%%%%%%%%%%%%%%%%%%%%%%%%%%%%%%%%%%%%%%%%%%

\subsubsection{Application to SDP symmetry reduction}

Semidefinite optimization is an important tool in quantum information \cite{SDPs,watrous}.
When solving semidefinite optimization problems (SDPs), it is useful to take their symmetries into account, both in theory \cite{invariantSDPs} and in practice \cite{RepLAB}.
Eliminating irrelevant degrees of freedom, also known as \emph{symmetry reduction}, can lead to significant computational savings and yield useful theoretical insights into the structure of solutions.

In quantum information, semidefinite optimization often intersects with Schur--Weyl duality due to the types of problems commonly investigated.
If the SDP matrix variable, which typically represents a quantum state or a measurement operator, commutes with $U\xp{p}$, the problem can be significantly simplified by working in Schur basis.
More generally, if the matrix variable, such as the Choi matrix of a quantum channel or superchannel, commutes with $U\xp{p} \x \bar{U}\xp{q}$, one should instead work in mixed Schur basis.
We refer to this type of symmetry as local \emph{unitary equivariance} since it most commonly occurs when considering the Choi matrix of a unitary-equivariant quantum channel (see \cref{sec:Unitary-equivariant quantum channels} for more details).
Recent examples where this symmetry occurs are quantum majority vote \cite{Majority},
black-box transformations of quantum gates
\cite{quintino2019reversing,quintino2019probabilistic,quintino2021deterministic,yoshida2021universal,yoshida2022reversing,ComplexConjugation},
asymmetric cloning
\cite{AsymmetricCloning,nechita2023asymmetric},
entanglement witnesses
\cite{huber2021dimensionfree},
and monogamy of entanglement
\cite{Monogamy}.

Our first application of the Gelfand--Tsetlin basis is for symmetry reduction of a general class of unitary-equivariant SDPs.
Our main result in this regard is \cref{thm:SDP},
which generalizes \cite{grinko2022linear} from linear to semidefinite optimization with unitary symmetries.

\begin{theorem*}[Qualitative statement of \cref{thm:SDP}]
    An SDP with a matrix variable $X \succeq 0$ subject to unitary equivariance $[X,U\xp{p} \x \bar{U}\xp{q}] = 0$ for all $U \in \U{d}$, can be reduced from $d^{2(p+q)}$ variables to $\dim(\A^d_{p,q})$ variables.
    This reduction can be computed in time polynomial in $\dim(\A^d_{p,q})$.
\end{theorem*}

For example, in the regime where $p$ and $q$ are small constants while $d$ is large (e.g., $p=2$, $q=3$, $d=1000$), solving the SDP is clearly impossible by conventional methods since it has $1000^{10} = 10^{30}$ variables.
In contrast, our method reduces it to a problem with only $(2+3)! = 5! = 120$ variables.
When $d$ is also small, even fewer variables suffice (e.g., $42$ variables when $d = 2$, see Appendix~E of \cite{grinko2022linear}).

\subsubsection{Application to port-based teleportation}

Our second application of the Gelfand--Tsetlin basis is to port-based teleportation.
In \cref{thm:pbt} we provide an efficient quantum circuit for port-based teleportation based on mixed quantum Schur transform.
This is a long-standing open problem that has been open since the invention of port-based teleportation
\cite{IshizakaHiroshima,ishizaka2009quantum}.

Together these two applications demonstrate the power of walled Brauer algebra techniques when applied to problems in quantum information, and suggests further possible applications of our Gelfand--Tsetlin basis and mixed quantum Schur transform.
We expect it to be particularly useful for deriving quantum circuits for general unitary-equivariant quantum channels.

\subsection{Related work}

Young's natural representation for the walled Brauer algebra has been constructed by Nikitin \cite{nikitin2007centralizer}, while Young's orthogonal and seminormal forms for the full Brauer algebra are described in \cite{nazarov1996young} and \cite{enyangseminormal}, respectively.
In addition, \cite{semikhatov2017quantum} has obtained a seminormal form for the $q$-deformed version of the walled Brauer algebra.
However, taking the ``classical'' $q \rightarrow 1$ limit of their construction and renormalizing the resulting basis vectors to obtain the corresponding orthogonal form is non-trivial.
Moreover, their construction only works for semisimple walled Brauer algebras $\B_{p,q}^d$, so the problem of adapting their construction to $\A_{p,q}^d$ is still open.
In the context of quantum information, the $q=1$ case of $\A^d_{p,q}$ before our work was studied in \cite{studzinski2013commutant,mozrzymas2014structure,mozrzymas2018simplified}, and some aspects of the general $q$ case were studied in \cite{studzinski2020efficient}.

While preparing our manuscript, we became aware of two simultaneous works \cite{quynh2023mixedschur,jiani023reducing}, which achieve similar results as ours.
In particular, \cite{quynh2023mixedschur} constructs mixed quantum Schur transform in the same way as our \cref{thm:mixed_schur}, based on the original construction from \cite{bch2006quantumschur,HarrowThesis}.
However, they have a number of different applications which do not intersect with our applications.
The second work \cite{jiani023reducing} finds an efficient implementation of the pretty good measurement for the optimal port-based teleportation protocol by decomposing induced representations. They achieve a similar result as our \cref{thm:pbt}, i.e., they construct a protocol with polynomial complexity. However, their approach is different from ours and is based on the other side of the Schur--Weyl duality: they induce from the symmetric group register, while we decompose a tensor product of unitary irreducible representations via the dual Clebsch--Gordan transform from our \cref{thm:mixed_schur}.
Thus our approaches are complimentary to each other in a representation-theoretic sense.

Various applications of the quantum Schur transform were recently studied in a number of papers \cite{havlivcek2018quantum,havlivcek2019classical,zheng2021speeding}. In \cite{havlivcek2018quantum,havlivcek2019classical}, authors computed the matrix entries of the quantum Schur transform, although only for qubits. Ref.~\cite{zheng2021speeding} studies how to define efficient ansatze for variational quantum algorithms for problems with $U\xp{n}$ symmetry. We expect our work to be useful in extending their results to the setting of unitary-equivariance symmetry $U\xp{p} \x \bar{U}\xp{q}$.

In our previous work~\cite{grinko2022linear}, we used the method developed in~\cite{doty2019canonical} to computed the primitive central idempotents of $\A^d_{p,q}$, which is the same as determining the diagonal matrix units (or diagonal entries of irreps).
In this paper, we determine all irrep entries, including the off-diagonal ones.
This immediately allows us to extend our previous application from linear to semidefinite programming with unitary symmetries.

%%%%%%%%%%%%%%%%%%%%%%%%%%%%%%%%%%%%%%%%%%%%%%%%%%%%%%%%%%%%%%%%%%%%%%%%%%%%%%%%%%%%%%%%%%%%%%%%%%%%%%%%%%%%%%%%%%%%%%%%%%%%%%%%%%%%%%%%%%%%%%%%%%%%%%%%%%%%%%%
%%%%%%%%%%%%%%%%%%%%%%%%%%%%%%%%%%%%%%%%%%%%%%%%%%%%%%%%%%%%%%%%%%%%%%%%%%%%%%%%%%%%%%%%%%%%%%%%%%%%%%%%%%%%%%%%%%%%%%%%%%%%%%%%%%%%%%%%%%%%%%%%%%%%%%%%%%%%%%%

\section{Preliminaries}\label{sec:Preliminaries}

\subsection{Young diagrams and tableaux}
\label{sec:tableaux}

A \emph{partition} $\lambda \vdash p$ of an integer $p \geq 0$ is a tuple of integers $\lambda = (\lambda_1, \dotsc, \lambda_k)$ such that $\lambda_1 \geq \cdots \geq \lambda_k \geq 0$
and $\lambda_1 + \cdots + \lambda_k = p$.
We denote by $\ell (\lambda) = \max \set{k : \lambda_k > 0}$ the \emph{length} of $\lambda$.
A partition $\lambda \vdash p$ can be graphically represented as a \emph{Young diagram}---a collection of $p$ \emph{cells} arranged in $k$ rows with $\lambda_i$ of them in the $i$-th row.
For example,
\begin{equation}
    \yd[1]{3,1}
    \label{eq:YD}
\end{equation}
represents the partition $\lambda=(3,1)$.
We call $\mu = (\mu_1,\dotsc,\mu_{k'})$ a \emph{subpartition} of $\lambda = (\lambda_1, \dotsc, \lambda_k)$, and write $\mu \subseteq \lambda$ if $k' \leq k$ and $\mu_i \leq \lambda_i$ for $i=1,\dotsc,k'$.
The \emph{size} $\abs{\lambda}$ of Young diagram denotes the number of boxes $p$.

Any cell $u \in \lambda$ of a Young diagram $\lambda$ can be specified by its row and column coordinates $i$ and $j$.
The \emph{content} of cell $u = (i,j)$ is
\begin{equation}
  \cont(u) \defeq j - i.
  \label{eq:content}
\end{equation}
For example, the cells of the Young diagram $(5,3,3)$ have the following content:
\begin{equation}
    \ytableausetup{centertableaux}
    \begin{ytableau}
        0 & 1 & 2 & 3 & 4 \\
       -1 & 0 & 1 \\
       -2 &-1 & 0
    \end{ytableau}
    \;.
\end{equation}
Note that content is constant on diagonals of $\lambda$ and indicates how far each diagonal is from the main one.
Content increases by one when going right or up, and decreases by one when going left or down.
The \emph{axial distance} (also known as \emph{hook} or \emph{Manhattan distance}) between two cells $u = (i,j)$ and $v = (i',j')$ in a Young diagram is
\begin{equation}
    r(u,v) \defeq \cont(v) - \cont(u).
\end{equation}
For example, the axial distance from cell $u$ to all other cells in the Young diagram $(5,3,3)$ is as follows:
\begin{equation}
    \ytableausetup{centertableaux}
    \begin{ytableau}
        -1 & u & 1 & 2 & 3 \\
        -2 &-1 & 0 \\
        -3 &-2 &-1
    \end{ytableau}
    \;.
\end{equation}

For a Young diagram $\lambda$, a cell $u \in \lambda$ is called \emph{removable} if the diagram $\lambda / u$ obtained by removing the cell $u$ from $\lambda$ is a valid Young diagram.
Similarly, a cell $u \notin \lambda$ is called \emph{addable} if the diagram $\lambda \cup u$ obtained by adding the cell $u$ to $\lambda$ is a valid Young diagram.
The set of all removable cells of $\lambda$ is denoted by $\RC(\lambda)$, while the set of all addable cells by $\AC(\lambda)$.
For example, the Young diagram $\lambda = (5,3,3)$ (shown in gray) has two removable cells (shown in white):
$r_1 = (1,5)$,
$r_2 = (3,3)$,
and three addable cells:
$a_1 = (1,6)$,
$a_2 = (2,4)$,
$a_3 = (4,1)$:
\begin{equation}
    \ytableausetup{centertableaux}
    \begin{ytableau}[*(lightgray)]
      & & & & r_1 & *(white) a_1 \\
      & & & *(white) a_2 \\
      & & r_3 \\
      *(white) a_4
    \end{ytableau}
    \;.
\end{equation}

A \emph{Young tableau} $T$ of shape $\lambda \vdash p$ is a Young diagram with cells filled with some natural numbers.
A \emph{standard Young tableau} $T$ of shape $\lambda \vdash p$ is obtained by filling the boxes of the Young diagram $\lambda$ with symbols from $[p] \defeq \set{1,\dotsc,p}$ strictly increasing across rows and down the columns.
For example,
\begin{equation}
    \label{eq:SYT}
    T_1 =
    \ytableausetup{centertableaux}
    \begin{ytableau}
        1 & 2 & 3 \\
        4
    \end{ytableau}\;,\qquad
    T_2 =
    \ytableausetup{centertableaux}
    \begin{ytableau}
        1 & 2 & 4 \\
        3
    \end{ytableau}\;,\qquad
    T_3 =
    \ytableausetup{centertableaux}
    \begin{ytableau}
        1 & 3 & 4 \\
        2
    \end{ytableau}
\end{equation}
are all standard Young tableaux of shape $\lambda=(3,1)$.
The set of all standard Young tableaux of a given shape $\lambda$ is denoted as $\SYT(\lambda)$.
According to the well-known \emph{hook length formula},
\begin{equation}
    \abs{\SYT(\lambda)}
    = \frac{(\lambda_1 + \dotsb + \lambda_d)!}{\prod_{u \in \lambda} h(u)}
    \eqdef d_\lambda
    \label{eq:hook length formula}
\end{equation}
where $h(u)$ denotes the \emph{hook length} of cell $u$ (the number of boxes to the right and below $u$, including $u$ itself).

Notice that any standard Young tableau $T \in \SYT(\lambda)$ can be represented as a sequence of $(p+1)$ Young diagrams
\begin{equation}
    T = \of{T^{0},\dotsc,T^p},
    \label{eq:SYTsequence}
\end{equation}
such that $T^i \pt i$, $T^i \subseteq T^{i+1}$, and $T^p =\lambda$.
Indeed, $T^i$ is obtained from the Young diagram $\lambda \vdash p$ by keeping only those boxes of $T$ whose symbols are in $[i]$.
For example, the Young tableau $T_1$ from \cref{eq:SYT} is represented by the following sequence of Young diagrams (see also \cref{fig:S4}):
\begin{equation}
    \label{eq:YTex}
    T_1 = \of*{
        \0 \,,\,
        \yd[1]{1} \,,\,
        \yd[1]{2} \,,\,
        \yd[1]{3} \,,\,
        \yd[1]{3,1}
    \,}.
\end{equation}

Consider a standard Young tableau $T \in \SYT(\lambda)$ of shape $\lambda \vdash p$ and an arbitrary permutation $\pi \in \S_p$.
We will denote by $\pi T$ the tableau obtained by permuting cell fillings of $T$ according to $\pi$.
For example, the Young tableaux $T_1,T_2,T_3$ presented in \cref{eq:SYT} are related in the following way: $T_2 = (34) T_1$, $T_3 = (23) T_2 = (234) T_1$, where $(34),(23),(234) \in \S_4$.
Note that $\pi T$ is not necessarily a standard tableau, e.g., consider $(14) T_1$.

Given a standard Young tableau $T$, we define
\begin{equation}
    \cont_i(T) \defeq \cont(T^i \backslash T^{i-1}).
    \label{eq:conti}
\end{equation}
This is simply the content of the cell of $T$ containing $i$. Moreover, we define $r_i(T)$ to be the hook distance between cells containing $i$ and $i+1$, i.e.,
\begin{equation}
    r_i(T) \defeq \cont_{i+1}(T) - \cont_i(T).
    \label{eq:distance}
\end{equation}

A \emph{semistandard Young tableau} $M$ of a shape $\lambda$ and entries in $[d]$ is obtained by filling the boxes of a Young diagram $\lambda$ with symbols from $[d]$ weakly increasing across the rows and strictly increasing down the columns. For example,
\begin{equation}
    \label{eq:SSYT}
    M_1 =
    \ytableausetup{centertableaux}
    \begin{ytableau}
        1 & 1 & 1 \\
        2
    \end{ytableau}\;,\qquad
    M_2 =
    \ytableausetup{centertableaux}
    \begin{ytableau}
        1 & 1 & 2 \\
        2
    \end{ytableau}\;,\qquad
    M_3 =
    \ytableausetup{centertableaux}
    \begin{ytableau}
        1 & 2 & 2 \\
        2
    \end{ytableau}
\end{equation}
are all semistandard Young tableau of shape $\lambda = (3,1)$ with entries in $[2]$.
We will denote the set of all semistandard Young tableaux of shape $\lambda$ and entries in $[d]$ by $\SSYT(\lambda,d)$.
According to the well-known \emph{Weyl dimension formula}
\cite[eq.~(11.46)]{louck2008unitary}
and \emph{hook-content formula},
\begin{equation}
    \abs{\SSYT(\lambda,d)}
    = \prod_{1 \leq i < j \leq d}
    \frac{\lambda_i - \lambda_j + j - i}{j-i}
    = \frac{\prod_{u \in \lambda} (d + \cont(u))}{\prod_{u \in \lambda} h(u)}
    \eqdef m_\lambda.
    \label{eq:hook content formula}
\end{equation}

Recording the number of times each number appears in tableau $M$ gives a sequence known as the \emph{weight} of $M$, denoted as $w(M)$:
\begin{equation}
    w(M)_i \defeq \text{``the number of $i$'s in tableau $M$''}.
    \label{eq:w(M)}
\end{equation}
For example, the tableaux presented in \cref{eq:SSYT} have weights $(3,1)$, $(2,2)$, and $(1,3)$, respectively.
We can extend the notion of weight also to sequences of natural numbers.
The \emph{weight} $w(x_1,\dotsc,x_k)$ of a sequence $x_1,\dotsc,x_k$ records the number of times each natural number appears in it.
For example,
\begin{equation}
1,2,2,2
,\qquad
2,1,2,2
,\qquad
2,2,1,2
,\qquad
2,2,2,1
\end{equation}
are all sequences of weight $w(M)=(1,3)$.

%%%%%%%%%%%%%%%%%%%%%%%%%%%%%%%%%%%%%%%%%%%%%%%%%%%%%%%%%%%%%%%%%%%%%%%%%%%%%%%%%%%%%%%%%%%%%%%%%%%%%%%%%%%%%%%%%%%%%%%%%%%%%%%%%%%%%%%%%%%%%%%%%%%%%%%%%%%%%%%

\subsection{Bratteli diagram and Gelfand--Tsetlin basis}\label{sec:BrattDiag}

An algebra over $\C$ is \emph{semisimple} if it is isomorphic to a direct sum $\bigoplus_{i=1}^k \End(\C^{n_i})$ of full matrix algebras, for some integers $k,n_1,\dotsc,n_k \geq 1$ \cite[Theorem~2.2.4]{Cox}.
Following \cite{doty2019canonical}, we call a sequence
\begin{equation}
    \A_0 \hookrightarrow \A_1 \hookrightarrow \dots \hookrightarrow \A_n
    \label{eq:As}
\end{equation}
of finite-dimensional semisimple algebras over $\C$ a \emph{multiplicity-free family} if
\begin{enumerate}[(a)]
    \item $\A_0 \cong \C$,
    \item each embedding $\A_i \hookrightarrow \A_{i+1}$ is unity-preserving\footnote{It maps the unit (i.e., the identity matrix) of one algebra to the unit of the other algebra.}, and
    \item the restriction from $\A_{i+1}$ to $\A_i$ is multiplicity-free\footnote{This means that the restriction of any simple $\A_{i+1}$-module to $\A_i$ is isomorphic to a direct sum of pairwise non-isomorphic simple $\A_i$-modules.}.
\end{enumerate}
A canonical example of such family is the sequence $\CS_0 \hookrightarrow \CS_1 \hookrightarrow \dots \hookrightarrow \CS_n$ of symmetric group algebras, where $\CS_i$ is embedded into $\CS_{i+1}$ by acting trivially on the element $i+1$ \cite{OV1996, VO2005}.
Two other examples of such families are walled Brauer algebras $\B^d_{p,q}$ and partially transposed permutation matrix algebras $\B^d_{p,q}$ that will be particularly important to us (see \cref{sec:Bpq,sec:Apq} for more details).

To any multiplicity-free family of algebras $\A_0 \hookrightarrow \A_1 \hookrightarrow \dots \hookrightarrow \A_n$ we can associate a \emph{Bratteli diagram}.
It is a directed acyclic simple graph $(V,E)$ whose vertices correspond to simple modules of $\A_i$ and edges show how they decompose when restricted to the previous subalgebra $\A_{i-1}$ \cite{Bratteli}.
More specifically, the vertices of the Bratteli diagram are divided into levels $i = 0, \dotsc, n$, i.e., $V = V_0 \sqcup V_1 \sqcup \dots \sqcup V_n$, where $V_i$ denotes the set of vertices at \emph{level}~$i$. Furthermore, at level $0$ the set $V_0 \defeq \set{\0}$ contains only the \emph{root} $\0$, while at levels $i = 1, \dotsc, n$ the set $V_i \defeq \Irr{\A_i}$ contains the labels of all non-isomorphic simple modules of $\A_i$.
We will call the vertices at level $n$ \emph{leaves}.
We draw an edge\footnote{The notion of Bratteli diagram can be extended to sequences of algebras which are not necessarily multiplicity-free. In that case the corresponding graph is not simple anymore, and the number of edges connecting vertices $\lambda$ and $\mu$ is equal to the number of direct summands in $\Res^{\A_{i+1}}_{\A_i} V^{\mu}$ isomorphic to $V^{\lambda}$, i.e., the multiplicity of $V^{\lambda}$ in the restcition of $V^{\mu}$.} $\lambda \rightarrow \mu$ from vertex $\lambda \in \Irr{\A_i}$ to vertex $\mu \in \Irr{\A_{i+1}}$ if and only if the simple module $V^{\lambda}$ corresponding to $\lambda$ is isomorphic to a direct summand in $\Res^{\A_{i+1}}_{\A_i}V^{\mu}$, the restriction of $V^{\mu}$ to the subalgebra $\A_i$, see \cref{fig:S4}.

\begin{figure}[!ht]
\begin{tikzpicture}[> = latex,
 MT/.style = {draw = blue!40, line width = 3pt},
  cut/.style = {thin, dashed, red, rounded corners = 12pt, fill = red, fill opacity = 0.12},
  every node/.style = {inner sep = 3pt, anchor = west}]
  \def\W{2.0cm}
  \def\H{1cm}
  % Labels
  \foreach \i in {0,1,2,3,4} {
    \node at (\i*\W,2.9*\H) {$\CS_\i$};
  }
  % Level 0
  \node (0)    at (0*\W, 0*\H) {$\0$};
  % Level 1
  \node (1)    at (1*\W, 0*\H) {\yd{1}};
  % Level 2
  \node (2)    at (2*\W, 1*\H) {\yd{2}};
  \node (11)   at (2*\W,-1*\H) {\yd{1,1}};
  % Level 3
  \node (3)    at (3*\W, 2*\H) {\yd{3}};
  \node (21)   at (3*\W, 0*\H) {\yd{2,1}};
  \node (111)  at (3*\W,-2*\H) {\yd{1,1,1}};
  % Level 4
  \node (4)    at (4*\W, 2.0*\H) {\yd{4}};
  \node (31)   at (4*\W, 1.0*\H) {\yd{3,1}};
  \node (22)   at (4*\W, 0.0*\H) {\yd{2,2}};
  \node (211)  at (4*\W,-1.0*\H) {\yd{2,1,1}};
  \node (1111) at (4*\W,-2.0*\H) {\yd{1,1,1,1}};
  % Edges 0 -> 1
  \draw[->] (0) -- (1);
  % Edges 1 -> 2
  \draw[->] (1) -- (2.190);
  \draw[->] (1) -- (11.170);
  % % Edges 2 -> 3
  \draw[->] (2) -- (3.190);
  \draw[->] (2) -- (21.165);
  \draw[->] (11) -- (21.195);
  \draw[->] (11) -- (111.170);
  % Edges 3 -> 4
  \draw[->] (3) -- (4);
  \draw[->] (3) -- (31);
  \draw[->] (21) -- (31);
  \draw[->] (21) -- (22);
  \draw[->] (21) -- (211);
  \draw[->] (111) -- (211);
  \draw[->] (111) -- (1111);
  \begin{scope}[on background layer]
    % Paths in \M{T}
    \draw[MT] (0) -- (1);
    \draw[MT]  (1) -- (2.190);
    \draw[MT] (2) -- (3.190);
    \draw[MT]  (3) -- (31);
  \end{scope}
\end{tikzpicture}
\caption{\label{fig:S4}The Bratteli diagram for the family of symmetric group algebras
$\CS_0 \hookrightarrow
 \CS_1 \hookrightarrow
 \CS_2 \hookrightarrow
 \CS_3 \hookrightarrow
 \CS_4$,
also known as Young's lattice. The vertices at level $k$ are labelled by Young diagrams $\lambda \vdash k$ corresponding to all non-isomorphic irreducible representations of $\CS_k$. We have highlighted the path corresponding to sequence~\eqref{eq:YTex} and terminating at leaf $\lambda=(3,1)$.}
\end{figure}
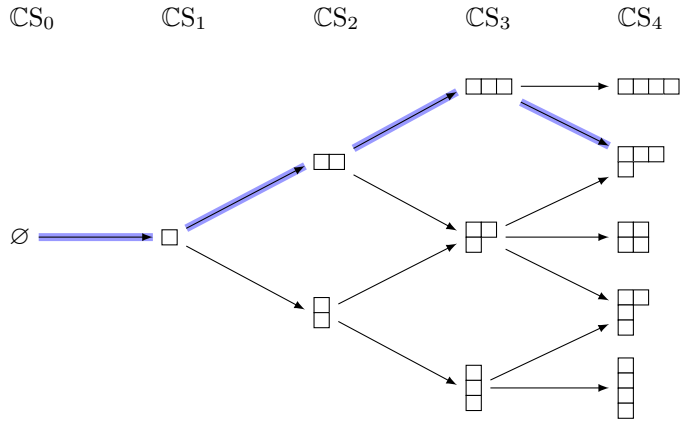

Paths\footnote{By a ``path'' we always mean a \emph{directed path}, i.e., a path that traverses edges only in the allowed direction (from lower to higher levels).} in the Bratteli diagram play an important role in the representation theory of the corresponding algebras, so we will introduce some notation for them.
If $\lambda$ and $\mu$ are two vertices at levels $i < j$ of the Bratteli diagram, we will denote the set of all paths from $\lambda$ to $\mu$ by
\begin{equation}
    \Paths_{i,j} (\lambda,\mu).
\end{equation}
We will denote the set of all paths starting at the root $\0$ and terminating at vertex $\lambda$ at level $j$ by
\begin{equation}
    \Paths_j (\lambda) \defeq \Paths_{0,j} \of[\big]{\0,\lambda}.
\end{equation}
When $j = n$, i.e., $\lambda$ is a leaf, we will abbreviate this to
\begin{equation}
    \Paths(\lambda) \defeq \Paths_n (\lambda).
\end{equation}
An arbitrary path $T = (T^0, \dotsc,T^n)$ in the Bratteli diagram can be \emph{decomposed} at level $i \in \set{0,\dotsc,n}$ as
$T = T_1 \circ T_2$ where
$T_1 = (T^0, \dotsc, T^i)$ and
$T_2 = (T^i, \dotsc, T^n)$
belong to
$\Paths_i(T^i)$ and
$\Paths_{i,n}(T^i,T^n)$, respectively.

An irreducible representation $\lambda \in \Irr{\A_n}$ of $\A_n$ corresponds to a leaf $\lambda$ of the Bratteli diagram.
The basis for the representation space $V^\lambda$ of $\lambda$ is indexed by all paths from the root $\0$ to $\lambda$. We call this the \emph{Gelfand--Tsetlin basis} for the multiplicity-free family $\A_0 \hookrightarrow \A_1 \hookrightarrow \dots \hookrightarrow \A_n$.
Such basis can be obtained by choosing any leaf $\lambda \in \Irr{\A_n}$ and considering the restriction $\smash{\Res^{\A_n}_{\A_{n-1}} V^\lambda}$ of the corresponding simple $\A_n$-module $V^\lambda$ to $\A_{n-1}$, which is multiplicity-free.
This restriction can then be iterated further along any path in the Bratteli diagram towards the root $\0$ that corresponds to the one-dimensional algebra $\A_0 \cong \C$.
Doing this along all $\Paths(\lambda)$ connecting $\0$ and $\lambda$ results in a decomposition of the chosen simple $\A_n$-module $V^\lambda$ into one-dimensional simple $\A_0$-modules.
Repeating this procedure for all leaves $\lambda \in \Irr{\A_n}$ produces the full \emph{Gelfand--Tsetlin basis} of $\bigoplus_{\lambda \in \Irr{\A_n}} V^\lambda$, which consists of all paths in the Bratteli diagram.

As an example, consider the sequence of symmetric group algebras
\begin{equation}
\label{eq:symG}
\CS_0 \hookrightarrow \CS_1 \hookrightarrow\cdots\hookrightarrow
\CS_p .
\end{equation}
It is well-known that irreducible representations $\Irr{\CS_i}$ of the symmetric group algebra $\CS_i$ (or equivalently the symmetric group $\S_i$ itself) are in one-to-one correspondence with Young diagrams $\lambda \vdash i$, and that the sequence~\eqref{eq:symG} of subalgebras is multiplicity-free.
Hence, the vertices of the corresponding Bratteli diagram are labelled by Young diagrams $\lambda \vdash i$.
Moreover, two vertices $\mu \vdash (i-1)$ and $\lambda \vdash i$ are connected if $\mu$ is a sub-diagram of $\lambda$ obtained by removing one box, see \cref{fig:S4}.
As a consequence, the sequence \eqref{eq:symG} determines the Gelfand--Tsetlin basis of the module corresponding to any irreducible representation $\lambda \in \Irr{\CS_p}$, also known as the \emph{Young--Yamanouchi basis} for symmetric group.
Vectors of this basis are labelled by paths $T \in \Paths(\lambda)$ ending at vertex $\lambda$ in the Bratteli diagram. Notice that such paths are in one-to-one correspondence with the elements of $\SYT(\lambda)$ since the $i$-th edge in the path indicates where the box containing $i$ appears in the standard tableau.
This implies that the dimension of irrep $\lambda \in \Irr{\CS_p}$ is equal to $\abs{\Paths(\lambda)} = \abs{\SYT(\lambda)}$, i.e., the number of standard Young tableaux of shape $\lambda$.
Recall that $\S_p = \langle \sigma_1, \dotsc, \sigma_{p-1} \rangle$, i.e., $\S_p$ is generated by transpositions $\sigma_i = (i,i+1)$ that exchange two consecutive elements $i$ and $i+1$.
The action of these generators has a remarkably simple form in the Young--Yamanouchi basis.
Indeed, consider a basis vector $\ket{T} \in V^\lambda$ where $T \in \SYT(\lambda)$ in the representation space $V^\lambda$ of irrep $\lambda$.
Then $\sigma_i$ acts on $\ket{T}$ in the following way:
\begin{equation}
    \label{eq:SymGroupYY}
    \sigma_i \, \ket{T} =
    \frac{1}{r_i(T)} \, \ket{T} + \sqrt{1-\frac{1}{r_i(T)^2}} \; \ket{\sigma_i T}
\end{equation}
where $r_i(T)$ is the axial distance between cells containing $i$ and $i+1$ in the Young tableau $T$, see \cref{eq:distance}, and $\sigma_i T \in \SYT(\lambda)$ is the tableau $T$ with numbers $i$ and $i+1$ exchanged, see \cref{sec:tableaux}.
Note that $r_i(T) = \pm 1$ whenever $\sigma_i T$ is not a standard Young tableau, hence the coefficient in front of $\ket{\sigma_i T}$ vanishes.
Our main technical result (see \cref{thm:main}), generalizes \cref{eq:SymGroupYY} to the matrix algebra $\A^d_{p,q}$ of partially transposed permutations.

%%%%%%%%%%%%%%%%%%%%%%%%%%%%%%%%%%%%%%%%%%%%%%%%%%%%%%%%%%%%%%%%%%%%%%%%%%%%%%%%%%%%%%%%%%%%%%%%%%%%%%%%%%%%%%%%%%%%%%%%%%%%%%%%%%%%%%%%%%%%%%%%%%%%%%%%%%%%%%%

%%%%%%%%%%%%%%%%%%%%%%%%%%%%%%%%%%%%%%%%%%%%%%%%%%%%%%%%%%%%%%%%%%%%%%%%%%%%%%%%%%%%%%%%%%%%%%%%%%%%%%%%%%%%%%%%%%%%%%%%%%%%%%%%%%%%%%%%%%%%%%%%%%%%%%%%%%%%%%%

\subsection{Gelfand--Tsetlin patterns}
\label{sec:GT}

A \emph{Gelfand--Tsetlin pattern} $M$ of shape $\lambda$ and length $d$ is represented by a triangular table with $d$ rows and $i$ integers in the $i$-th row (when counted from the bottom):
\begin{equation}
    \label{eq:GTpattern}
    M =
    \begin{bmatrix}
        m_{1,d} && m_{2,d} && \cdots && m_{d-1, d} && m_{d,d} \\
        &m_{1,d-1}&&&\cdots&&&m_{d-1,d-1}&\\
        &&\ddots&&\vdots&&\reflectbox{$\ddots$}&&\\
        &&& m_{1,2}&&m_{2,2}&&&\\
        &&&& m_{1,1}&&&&
    \end{bmatrix}
    =
    \begin{bmatrix}
        \m_d \\
        \m_{d-1} \\
        \vdots \\
        \m_2 \\
        \m_1
    \end{bmatrix},
\end{equation}
where $\m_j \defeq (m_{1,j}, \dotsc, m_{j,j})$ denotes the $j$-th row of $M$.
The entries $m_{i,j}$ of $M$ are subject to the following constraints: the top row of $M$ is equal to $\lambda$, i.e.,
$\m_d = (m_{1,d}, \dotsc, m_{d,d}) = (\lambda_1, \dotsc, \lambda_d)$, and all entries $m_{i,j}$ satisfy the so-called \emph{interlacing} or \emph{in-betweenness} condition:
\begin{equation}
    \label{eq:GTpatterns-relations}
    m_{i,j} \geq m_{i,j-1} \geq m_{i+1,j}
    \quad \text{for all} \quad
    1 \leq i < j \leq d.
\end{equation}
We will write $\m_{j-1} \squb \m_{j}$ as a shorthand for the interlacing relations \eqref{eq:GTpatterns-relations} between vectors $\m_{j-1}$ and $\m_{j}$.
We will denote the set of all Gelfand--Tsetlin patterns of shape $\lambda$ and length $d$ by $\GT(\lambda,d)$.

For any partition $\lambda$ of length $d$, the Gelfand--Tsetlin patterns $\GT(\lambda,d)$ are in one-to-one correspondence with semistandard Young tableaux $\SSYT(\lambda,d)$:
\begin{equation}
    \abs{\GT(\lambda,d)}
    = \abs{\SSYT(\lambda,d)}
    = m_\lambda.
\end{equation}
Indeed, for any tableau $T \in \SSYT(\lambda,d)$ let $m_{i,j}$ be the number of symbols from $[j]$ in the $i$-th row of $T$.
Equivalently, $\m_j \subseteq \lambda$ is the shape of the subtableau of $T$ formed by entries less or equal to $j$.
Then \cref{eq:GTpattern} constitutes the Gelfand--Tsetlin pattern of tableau $T$.
For example, the Gelfand--Tsetlin patterns corresponding to the semistandard Young tableaux in \cref{eq:SSYT} are
\begin{equation}
    \label{eq:GTpatterns}
    M_1 =
    \begin{bmatrix}
        3 &   & 1 \\
          & 3 &
    \end{bmatrix},\qquad
    M_2 =
    \begin{bmatrix}
        3 &   & 1 \\
          & 2 &
    \end{bmatrix},\qquad
    M_3 =
    \begin{bmatrix}
        3 &   & 1 \\
          & 1 &
    \end{bmatrix},
\end{equation}
which are in fact all Gelfand--Tsetlin patterns $\GT(\lambda,d)$ of shape $\lambda = (3,1)$ and length $d = 2$.
Conversely, given any Gelfand--Tsetlin pattern $M$, the corresponding semistandard tableau $T$ has shape $\lambda = \m_1$ given by the first row of $M$, while the filling of the $i$-th row of $T$ can be deduced from the $i$-th diagonal of $M$.
Indeed, $m_{i,j} - m_{i,j-1}$ is the number of $j$'s in the $i$-th row of $T$.

The \emph{weight} of Gelfand--Tsetlin pattern $M$ consists of differences of consecutive row sums:
\begin{equation}
    \label{eq:weightOfGTpattern}
    w(M) \defeq
    \of[\big]{
        w(\m_{1}),
        w(\m_{2}) - w(\m_{1}),
        \dotsc,
        w(\m_{d}) - w(\m_{d-1})
    }
    \quad \text{where} \quad
    w(\m_{j}) \defeq \sum_{i=1}^{j} m_{i,j}.
\end{equation}
Notice that the weight of a Gelfand--Tsetlin pattern coincides with the weight \eqref{eq:w(M)} of the corresponding semistandard Young tableaux.
For example, the patterns \eqref{eq:GTpatterns} have weights $(3,1)$, $(2,2)$, and $(1,3)$, respectively.

\subsection{Mixed Young diagrams and tableaux}\label{sec:mixed diagrams}

In this section, we define combinatorial notions analogous to Young diagrams and tableaux that will be used in \cref{sec:mixed Schur} to describe mixed Schur--Weyl duality and mixed quantum Schur transform.
First, let us first describe three equivalent ways of representing a pair of Young diagrams.

A \emph{mixed Young diagram} of length $d$ is a pair of Young diagrams $\lambda = (\lambda_l,\lambda_r)$ of total length at most $d$: $\ell(\lambda_l) + \ell(\lambda_r) \leq d$.
Equivalently, we can represent $\lambda$ by combining the diagrams $\lambda_l$ and $\lambda_r$ into a single \emph{staircase} $\widetilde{\lambda} \in \mathbb{Z}^d$ obtained by subtracting from
$\lambda_l = (\lambda_{l,1}, \dotsc, \lambda_{l,d})$
the reverse of
$\lambda_r = (\lambda_{r,1}, \dotsc, \lambda_{r,d})$
\cite{stembridge1987rational}:
\begin{equation}
    \label{eq:GTpatternOfAPair}
    \widetilde{\lambda} \defeq
    \of[\big]{
        \lambda_{l,1} - \lambda_{r,d},\,
        \lambda_{l,2} - \lambda_{r,d-1},\,
        \dotsc,\,
        \lambda_{l,d} - \lambda_{r,1}
    }.
\end{equation}
Intuitively, the staircase $\widetilde{\lambda}$ corresponds to rotating the diagram $\lambda_r$ by $180$ degrees and attaching it at the bottom of $\lambda_l$ (see \cref{fig:WalledConcatenation}).
Since $\ell(\lambda_l) + \ell(\lambda_r) \leq d$, this operation is reversible and one can easily recover $\lambda_l$ and $\lambda_r$ from the staircase $\widetilde{\lambda}$.
Finally, a third way of representing the same concept is via \emph{walled concatenation} $(\hat{\lambda},s)$ where $s \defeq \lambda_{r,1}$ and $\hat{\lambda}$ is a Young diagram of shape
\begin{equation}
    \hat{\lambda}
    \defeq \lambda_{r,1} + \widetilde{\lambda}
    =
    \of[\big]{
        \lambda_{r,1} + \lambda_{l,1} - \lambda_{r,d},\,
        \lambda_{r,1} + \lambda_{l,2} - \lambda_{r,d-1},\,
        \dotsc,\,
        \lambda_{r,1} + \lambda_{l,d-1} - \lambda_{r,2},\
        \lambda_{l,d}
    }.
\end{equation}
This diagram corresponds to shifting the staircase $\widetilde{\lambda}$ to the right by $\lambda_{r,1}$ boxes so that all its entries become non-negative.
Equivalently, we can obtain $\hat{\lambda}$ by adding $\lambda_{r,1}$ columns of $d$ boxes on the left of $\lambda_l$ and then removing the diagram $\lambda_r$ (rotated by $180$ degrees) from the bottom of these columns
(see \cref{fig:WalledConcatenation}).
This process is reversible, so we can easily convert $(\hat{\lambda},s)$ back into the pair of diagrams $(\lambda_l,\lambda_r)$ or the staircase $\widetilde{\lambda}$.
Mixed Young diagrams, walled concatenations, and staircases are three equivalent ways of thinking about the same combinatorial object, see \cref{fig:WalledConcatenation}.
Throughout the paper we will use the same symbol $\lambda$ to denote either of these three concepts, depending on the convenience in a given context.

%% https://tex.stackexchange.com/questions/55068/is-there-a-tikz-equivalent-to-the-pstricks-ncbar-command
\tikzset{
    ncbar angle/.initial=90,
    ncbar/.style={
        to path=(\tikztostart)
        -- ($(\tikztostart)!#1!\pgfkeysvalueof{/tikz/ncbar angle}:(\tikztotarget)$)
        -- ($(\tikztotarget)!($(\tikztostart)!#1!\pgfkeysvalueof{/tikz/ncbar angle}:(\tikztotarget)$)!\pgfkeysvalueof{/tikz/ncbar angle}:(\tikztostart)$)
        -- (\tikztotarget)
    },
    ncbar/.default=0.5cm,
}

\tikzset{square left brace/.style={ncbar=0.5cm}}
\tikzset{square right brace/.style={ncbar=-0.5cm}}

\tikzset{round left paren/.style={ncbar=0.5cm,out=120,in=-120}}
\tikzset{round right paren/.style={ncbar=0.5cm,out=60,in=-60}}

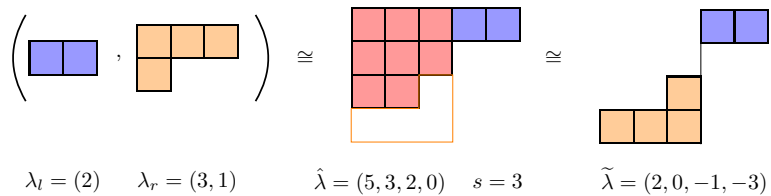
\begin{figure}[!ht]

\resizebox{0.6\textwidth}{!}{
\begin{tikzpicture}[inner sep = 0in, outer sep = 0in]
    \begin{scope}
        \draw [black, thick] (-0.6,-0.7) to [round left paren ] (-0.6,0.7);
        \draw [black, thick] ( 3.1,-0.7) to [round right paren] ( 3.1,0.7);
        \node (n) {\begin{varwidth}{5cm}{
        \begin{ytableau}
            *(blue!40) & *(blue!40)  \\
        \end{ytableau}}\end{varwidth}};
        \node at (0,-2) {$\lambda_l=(2)$};
        \begin{scope}[shift={(2,0)}]
            \node (n) {\begin{varwidth}{5cm}{
            \begin{ytableau}
                *(orange!40) & *(orange!40)& *(orange!40)  \\
                *(orange!40)
            \end{ytableau}}\end{varwidth}};
            \node at (0,-2) {$\lambda_r=(3,1)$};
        \end{scope}
        \node[very thick] (A) at (3.9,0) {$\cong$};
        \node[very thick] (A) at (0.9,0) {,};
    \end{scope}
    % $\quad$
    \begin{scope}[xshift = 6cm]%[inner sep=0in,outer sep=0in]
        \node (n) {\begin{varwidth}{5cm}{
        \begin{ytableau}
            *(red!40) & *(red!40) & *(red!40)  &*(blue!40)   &*(blue!40) \\
            *(red!40)&*(red!40)  & *(red!40)  \\
            *(red!40)& *(red!40)    \\
        \end{ytableau}}\end{varwidth}};
        \draw[orange] (n.south west)--++(0,-.5*1.09)--++(1.5*1.09,0)--++(0,1*1.09)--++(-.5*1.09,0)--++(0,-.5*1.09)--++(-1*1.09,0);
        \node at (-0.9,-2) {$\hat{\lambda}=(5,3,2,0)$};
        \node at ( 1.0,-2) {$s = 3$};
        \node[very thick] (A) at (1.9,0) {$\cong$};
    \end{scope}
    \begin{scope}[xshift = 10cm]%[inner sep=0in,outer sep=0in]
        \node (n) at (0,-0.3cm) {\begin{varwidth}{5cm}{
        \begin{ytableau}
            \none & \none & \none & *(blue!40) & *(blue!40) \\
            \none & \none & \none \\
            \none & \none & *(orange!40) \\
            *(orange!40) & *(orange!40) & *(orange!40)
        \end{ytableau}}\end{varwidth}};
        \draw (n.south east)++(-1.09,-0.01) -- ++(0,2*1.085);
        \node at (0,-2) {$\widetilde{\lambda}=(2,0,-1,-3)$};
    \end{scope}
\end{tikzpicture}
}
\caption{Three equivalent ways of representing a pair of Young diagrams:
as a mixed Young diagram
$\lambda = (\lambda_l,\lambda_r)$,
as a staircase
$\widetilde{\lambda}$,
or as a walled concatenation
$(\hat{\lambda},\lambda_{r,1})$.
Here the total length of all tableaux is $d=4$.}
\label{fig:WalledConcatenation}
\end{figure}

Gelfand--Tsetlin patterns \eqref{eq:GTpattern} generalize straightforwardly to mixed Young diagrams simply by allowing the entries $m_{i,j}$ of the pattern to be negative.
In particular, the first row $\m_{d} = (m_{1,d},\dotsc,m_{d,d})$ of a Gelfand--Tsetlin pattern does not have to be a partition but may also be a staircase.
Formally, for any mixed Young diagram $\lambda$ we let
$\GT(\lambda,d) \defeq \GT(\widetilde{\lambda},d)$,
i.e., Gelfand--Tsetlin patterns of shape $\lambda$ are defined as those of the corresponding staircase $\widetilde{\lambda}$.\footnote{This is a slight abuse of notation, as $\lambda$ in $\GT(\lambda,d)$ may be either a single Young diagram or a pair of Young diagrams. However, it will always be clear from context what kind of object $\lambda$ is. Furthermore, both notions agree when the right diagram in a pair of diagrams is empty: $\GT((\lambda_l,\0),d)=\GT(\lambda_l,d)$ for any Young diagram $\lambda_l$. Alternatively, one can always think of $\lambda$ as a staircase, in which case no ambiguity arises.}
Equivalently, one can replace the staircase $\widetilde{\lambda}$ by the walled concatenation $\hat{\lambda}$ and consider the set $\GT(\hat{\lambda},d)$ instead.
Indeed, there is a simple bijection between the sets $\GT(\widetilde{\lambda},d)$ and $\GT(\hat{\lambda},d)$:
for any $M \in \GT(\widetilde{\lambda},d)$ one can construct the corresponding $M' \in \GT(\hat{\lambda},d)$ by subtracting $\widetilde{\lambda}_d = m_{d,d}$ from each entry $m_{i,j}$ of the pattern $M$.
In particular, this implies that
\begin{equation}
    \abs{\GT(\widetilde{\lambda},d)}
    = \abs{\GT(\hat{\lambda},d)}
    = m_{\hat{\lambda}},
\end{equation}
where $m_{\hat{\lambda}}$ can be computed using \cref{eq:hook content formula}.

Recall from \cref{sec:GT} that for any Young diagram $\lambda$, the Gelfand--Tsetlin patterns $\GT(\lambda,d)$ are in one-to-one correspondence with semistandard Young tableaux $\SSYT(\lambda,d)$.
Similarly, for any staircase $\widetilde{\lambda}$, one can interpret $\GT(\widetilde{\lambda},d)$ as \emph{mixed semistandard Young tableaux}.

Recall from \cref{eq:SYTsequence} that a sequence of Young diagrams can be interpreted as a standard Young tableau.
Similarly, for any shape
$\lambda = (\lambda_l,\lambda_r)$, where $\lambda_l \pt p - k$ and $\lambda_r \pt q-k$
for some $p,q \geq 0$ and $k \geq 0$
such that $0 \leq k \leq \min(p,q)$, we define a \emph{mixed standard Young tableau} $T$ of length $d$ as a sequence $(T^0,T^1,T^2,\dotsc,T^{p+q})$ of mixed Young diagrams such that each $T^i$ has length $d$ and
\begin{enumerate}
    \item $T^0 \defeq (\0,\0)$ and $T^{p+q} \defeq \lambda = (\lambda_l, \lambda_r)$,
    \item if $p \geq i \geq 1$ then $T^{i-1}_l \subseteq T^{i}_l$ with $\abs{T^{i-1}_l} + 1 = \abs{T^{i}_l}$ and $T^{i-1}_r = T^{i}_r = \0$,
    \item if $p+q \geq i > p$ then either
    \begin{enumerate}
        \item $T^{i}_l \subseteq T^{i-1}_l$ with $\abs{T^{i}_l} + 1 = \abs{T^{i-1}_l}$ and $T^{i-1}_r = T^{i}_r$, or
        \item $T^{i-1}_r \subseteq T^{i}_r$ with $\abs{T^{i-1}_r} + 1 = \abs{T^{i}_r}$ and $T^{i-1}_l = T^{i}_l$.
    \end{enumerate}
\end{enumerate}
Such sequence corresponds to a path in a certain Bratteli diagram (see \cref{sec:rep_theory_A} for more details).

Similar to \cref{eq:SYT,eq:YTex}, we can translate a sequence $T$ of mixed Young diagrams into what is essentially a pair of standard tableaux, thus justifying calling $T$ a \emph{mixed standard Young tableau}.
The first $p$ steps of this translation build up the left tableau following the same procedure as described in \cref{sec:tableaux}.
The remaining steps either build up the right tableau in the same way or add secondary entries to the left tableau indicating at which steps the corresponding boxes are removed.
For example, when $p = 3$, $q = 2$, and $k = 1$, the sequence
\begin{equation}
    T = \of[\big]{\br{\0,\0},\br{\yd{1},\0},\br{\yd{2},\0},\br{\yd{3},\0},\br{\yd{2},\0},\br{\yd{2},\yd{1}}}
    \label{eq:T}
\end{equation}
corresponds to the following pair of tableau:
\begin{equation}
    T = \of*{\,
    \ytableausetup{centertableaux}
    \begin{ytableau}
        \scriptstyle 1 & \scriptstyle 2 & \scriptstyle 3,4
    \end{ytableau}
    \,,\,
    \ytableausetup{centertableaux}
    \begin{ytableau}
        \scriptstyle 5
    \end{ytableau}\,
    }.
    \label{example:mixed_tab}
\end{equation}
This is a mixed Young tableau of shape $\lambda = ((2),(1))$ since the left tableau has only two boxes remaining.
Similar to \cref{sec:tableaux}, for any permutation $\pi \in \S_p \times \S_q$ we will write $\pi T$ to denote the tableau obtained by permuting the cell fillings of $T$ according to $\pi$.

%%%%%%%%%%%%%%%%%%%%%%%%%%%%%%%%%%%%%%%%%%%%%%%%%%%%%%%%%%%%%%%%%%%%%%%%%%%%%%%%%%%%%%%%%%%%%%%%%%%%%%%%%%%%%%%%%%%%%%%%%%%%%%%%%%%%%%%%%%%%%%%%%%%%%%%%%%%%%%%
%%%%%%%%%%%%%%%%%%%%%%%%%%%%%%%%%%%%%%%%%%%%%%%%%%%%%%%%%%%%%%%%%%%%%%%%%%%%%%%%%%%%%%%%%%%%%%%%%%%%%%%%%%%%%%%%%%%%%%%%%%%%%%%%%%%%%%%%%%%%%%%%%%%%%%%%%%%%%%%

\section{Representation theory of partially transposed permutation matrix algebra}\label{sec:repThApqd}

\subsection{Walled Brauer algebra $\B^d_{p,q}$}\label{sec:Bpq}

Let $p,q \geq 0$ be integers and $d \in \C$ arbitrary\footnote{We will later require $d \geq 2$ to be an integer as well.}.
The \emph{walled Brauer algebra} $\B^d_{p,q}$ consists of formal complex linear combinations of diagrams, where each diagram has two rows of $p+q$ nodes each and a vertical \emph{wall} between the first $p$ and the last $q$ nodes \cite{turaev1989operator,koike1989decomposition,bchlls,Benkart,nikitin2007centralizer,Bulgakova}.
All nodes are connected in pairs, and any two connected nodes are either on the same side of the wall and in different rows, or the other way around.
For example, the following diagram
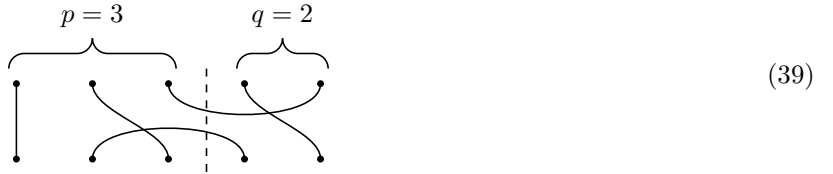
\begin{equation}
  \begin{tikzpicture}[baseline = 1cm]
    \BrauerTikZStyle
    % p and q
    \curlybrace{2.0}{1.4}{1.0} \node at (2.0,1.9) {$p=3$};
    \curlybrace{4.5}{1.4}{0.5} \node at (4.5,1.9) {$q=2$};
    % Nodes
    \foreach \i in {1,...,5} {
      \fill (\i,1) circle [dot] coordinate (A\i);
      \fill (\i,0) circle [dot] coordinate (B\i);
    }
    % Edges
    \draw (A1) to [l] (B1);
    \draw (A2) to [l] (B3);
    \draw (A4) to [l] (B5);
    \draw (A3) to [u] (A5);
    \draw (B2) to [n] (B4);
    % Wall
    \draw [dashed] (3.5,1.2) -- (3.5,-0.2);
  \end{tikzpicture}
\end{equation}
belongs to $\B^d_{3,2}$.
The addition in $\B^d_{p,q}$ is defined by simply adding the respective coefficients in the two formal linear combinations.
Multiplication of two diagrams corresponds to their \emph{concatenation} by identifying the bottom row of the first diagram with the top row of the second diagram.
Any loops that may have appeared in this process are erased and the resulting diagram is multiplied by the scalar $d^{\#\loops}$:
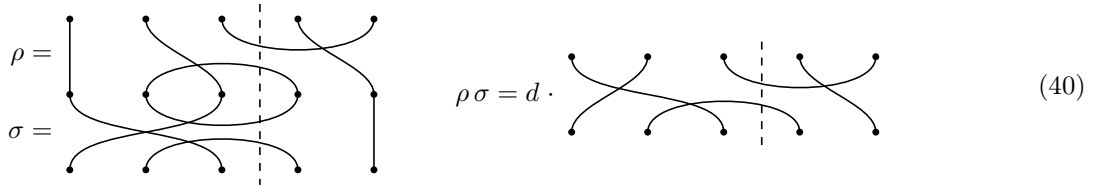
\begin{equation}
  \begin{tikzpicture}[baseline = 1cm]
    \BrauerTikZStyle

    \begin{scope}
      % Labels
      \node[anchor = east] at (0.9,1.5) {$\rho =$};
      \node[anchor = east] at (0.9,0.5) {$\sigma =$};
      % Nodes
      \foreach \i in {1,...,5} {
        \fill (\i,2) circle [dot] coordinate (A\i);
        \fill (\i,1) circle [dot] coordinate (B\i);
        \fill (\i,0) circle [dot] coordinate (C\i);
      }
      % Edges 1
      \draw (A1) to [l] (B1);
      \draw (A2) to [l] (B3);
      \draw (A4) to [l] (B5);
      \draw (A3) to [u] (A5);
      \draw (B2) to [n] (B4);
      % Edges 2
      \draw (B1) to [l] (C3);
      \draw (B3) to [l] (C1);
      \draw (B5) to [l] (C5);
      \draw (B2) to [u] (B4);
      \draw (C2) to [n] (C4);
      % Wall
      \draw [dashed] (3.5,2.2) -- (3.5,-0.2);
    \end{scope}

    \begin{scope}[xshift = 6.6cm, yshift = 0.5cm]
      % Labels
      \node[anchor = east] at (0.9,0.5) {$\rho \, \sigma = d \;\cdot$};
      % Nodes
      \foreach \i in {1,...,5} {
        \fill (\i,1) circle [dot] coordinate (D\i);
        \fill (\i,0) circle [dot] coordinate (E\i);
      }
      % Edges
      \draw (D1) to [l] (E3);
      \draw (D2) to [l] (E1);
      \draw (D4) to [l] (E5);
      \draw (D3) to [u] (D5);
      \draw (E2) to [n] (E4);
      % Wall
      \draw [dashed] (3.5,1.2) -- (3.5,-0.2);
    \end{scope}

  \end{tikzpicture}
\end{equation}
where parameter $d$ of a walled Brauer algebra $\B^d_{p,q}$ explicitly appear.
Multiplication of diagrams extends by linearity to multiplication in a walled Brauer algebra $\B^d_{p,q}$.

Note that the group algebra of the permutation group $\S_p \times \S_q$
forms a subalgebra of the walled Brauer algebra $\B^d_{p,q}$
consisting only of those diagrams where no edge goes across the wall.
In fact, the two algebras are isomorphic when $p = 0$ or $q = 0$, i.e.,
$\CS_p \cong \B^d_{p,0} \cong \B^d_{0,p}$ for any value of $d$.
The walled Brauer algebra $\B^d_{p,q}$ itself is a subalgebra of the so-called full \emph{Brauer algebra} $\B^d_{p+q}$ that is defined similarly but without the wall and with no restrictions on which pairs of nodes can be connected \cite{Brauer}.
This algebra was originally introduced by Brauer \cite{Brauer} for studying Schur--Weyl-like dualities of orthogonal and symplectic groups.

The walled Brauer algebra $\B^d_{p,q}$ is generated by \emph{transpositions} $\sigma_i$ that swap the $i$-th and $(i+1)$-th node of the two rows,
where $i \in \set{1, \dotsc, p+q-1} \setminus \set{p}$,
and a \emph{contraction} $\sigma_p$ between the $p$-th and $(p+1)$-th node within each row.
For example, $\B^d_{2,2}$ is generated by
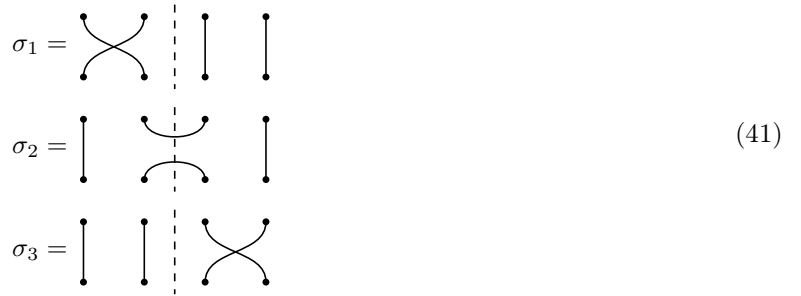
\begin{equation}
  \begin{tikzpicture}[baseline = 0.5cm, scale = 0.8]
    \BrauerTikZStyle
    \tikzset{looseness = 1}

    \begin{scope}[yshift = 1.7cm]
      \node[anchor = east] at (0.9,0.5) {$\sigma_1 =$};
      % Nodes
      \foreach \i in {1,...,4} {
        \fill (\i,1) circle [dot] coordinate (A\i);
        \fill (\i,0) circle [dot] coordinate (B\i);
      }
      % Edges
      \draw (A1) to [l] (B2);
      \draw (A2) to [l] (B1);
      \draw (A3) to [l] (B3);
      \draw (A4) to [l] (B4);
      % Wall
      \draw [dashed] (2.5,1.2) -- (2.5,-0.2);
    \end{scope}

    \begin{scope}
      \node[anchor = east] at (0.9,0.5) {$\sigma_2 =$};
      % Nodes
      \foreach \i in {1,...,4} {
        \fill (\i,1) circle [dot] coordinate (C\i);
        \fill (\i,0) circle [dot] coordinate (D\i);
      }
      % Edges
      \draw (C1) to [l] (D1);
      \draw (C2) to [u] (C3);
      \draw (D2) to [n] (D3);
      \draw (C4) to [l] (D4);
      % Wall
      \draw [dashed] (2.5,1.2) -- (2.5,-0.2);
    \end{scope}

    \begin{scope}[yshift = -1.7cm]
      \node[anchor = east] at (0.9,0.5) {$\sigma_3 =$};
      % Nodes
      \foreach \i in {1,...,4} {
        \fill (\i,1) circle [dot] coordinate (E\i);
        \fill (\i,0) circle [dot] coordinate (F\i);
      }
      % Edges
      \draw (E1) to [l] (F1);
      \draw (E2) to [l] (F2);
      \draw (E3) to [l] (F4);
      \draw (E4) to [l] (F3);
      % Wall
      \draw [dashed] (2.5,1.2) -- (2.5,-0.2);
    \end{scope}

  \end{tikzpicture}
  \label{eq:generators}
\end{equation}
One can also define $\B^d_{p,q}$ abstractly in terms of relations between its generators \cite{Bulgakova,brundan2012gradings,nikitin2007centralizer}.

\begin{restatable}{definition}{brauerdef}\label{def:Brauer}
Let $p,q \geq 0$ be integers and $d \in \C$.
The \emph{walled Brauer algebra} $\B^d_{p,q}$ is a finite-dimensional associative algebra over $\C$ generated by $\sigma_1, \dotsc, \sigma_{p+q-1}$ subject to the following relations:
\begin{align}
  & \mathrm{(a)} \enskip \sigma_i^2 = 1 \enskip (i \neq p), &
  & \mathrm{(b)} \enskip \sigma_i \sigma_{i+1} \sigma_i = \sigma_{i+1} \sigma_i \sigma_{i+1} \enskip (i \neq p-1,p), &
  & \mathrm{(c)} \enskip \sigma_i \sigma_j = \sigma_j \sigma_i \enskip (|i-j| > 1),
  \label{rel:transpositions} \\
  & \mathrm{(d)} \enskip \sigma_p^2 = d \sigma_p, &
  & \mathrm{(e)} \enskip \sigma_p \sigma_{p \pm 1} \sigma_p = \sigma_p, &
  & \mathrm{(f)} \enskip \sigma_p \sigma_i = \sigma_i \sigma_p \enskip (i \neq p \pm 1),
  \label{rel:contraction1} \\
  & \mathrm{(g)} \enskip \sigma_p \sigma_{p+1} \sigma_{p-1} \sigma_p \sigma_{p-1} = \sigma_p \sigma_{p+1} \sigma_{p-1} \sigma_p \sigma_{p+1},
  \hspace{-6cm}
  \label{rel:contraction2} \\
  & \mathrm{(h)} \enskip \sigma_{p-1} \sigma_p \sigma_{p+1} \sigma_{p-1} \sigma_p = \sigma_{p+1} \sigma_p \sigma_{p+1} \sigma_{p-1} \sigma_p.
  \hspace{-6cm}
  \label{rel:contraction3}
\end{align}
\end{restatable}

We can embed the walled Brauer algebra
$\B_{i,0}^d$ into
$\B_{i+1,0}^d$
by making adding a new pair of nodes on the left-hand side of the wall, which are connected by a vertical line.
Similarly, we can embed
$\B_{p,j}^d$ into
$\B_{p,j+1}^d$
by adding a new pair of nodes on the right-hand side of the wall.
This produces the following sequence of inclusions among walled Brauer algebras:
\begin{equation}
    \label{eq:B inclusions}
    \B_{0,0}^d \hookrightarrow
    \B_{1,0}^d \hookrightarrow \cdots \hookrightarrow \B_{p,0}^d \hookrightarrow
    \B_{p,1}^d \hookrightarrow \cdots \hookrightarrow
    \B_{p,q}^d.
\end{equation}
This is a multiplicity-free family when the underlying algebras are semisimple.

There exist special elements inside $\B_{p,q}^d$ called \emph{Jucys--Murphy elements}, which are given by
\begin{equation}
  J_k \defeq
  \begin{cases}
    0 & \text {if } k = 1, \\
    \sum_{i=1}^{k-1} (i,k) & \text{if } 2 \leq k \leq p, \\
    \sum_{i=p+1}^{k-1} (i,k) - \sum_{i=1}^p \overline{(i,k)} + d & \text{if } p+1 \leq k \leq p+q,
  \end{cases}
  \label{eq:JM for B}
\end{equation}
where $(i,k)$ is the transposition of elements $i$ and $k$, and $\overline{(i,k)} \defeq (i,p-1) (p,k) \sigma_p (i,p-1) (p,k)$ is the contraction between $i$ and $k$.
These elements are important in the representation theory of $\B_{p,q}^d$ \cite{brundan2012gradings,sartori2015walled,jung2020supersymmetric} and the related matrix algebra $\A^d_{p,q}$ \cite{grinko2022linear} defined in the next section.

\subsection{Matrix algebra $\A^d_{p,q}$ of partially transposed permutations}\label{sec:Apq}

Let us fix a \emph{local dimension} $d \geq 2$.
Let $V_d^p \defeq (\mathbb{C}^d)^{\otimes p}$ denote the tensor product space of $p$ qudits, and let $V_d^{p,q} \defeq V_d^p \otimes (V_d^q)^\ast$.
The walled Brauer algebra has a natural matrix representation $\psi^d_{p,q}\colon \B^d_{p,q} \to \End(V_d^{p,q})$ in which the generators $\sigma_i,\dotsc,\sigma_{p+q-1} \in \B^d_{p,q}$ act on strings $x_1,\dotsc,x_{p+q} \in [d]$ as follows:
\begin{equation}
    \label{eq:Brauer action}
    \psi^d_{p,q}(\sigma_i) \, \ket{x_1,\dotsc,x_{p+q}}
    \defeq
    \begin{cases}
        \ket{x_1,\dotsc,x_{i+1},x_i,\dotsc,x_{p+q}}, &
        \text{$i \neq p$}, \\
        \ket{x_1,\dotsc,x_{p-1}}
        \x
        \delta_{x_p,x_{p+1}}
        \sum_{k=1}^d \ket{k,k}
        \x
        \ket{x_{p+2},\dotsc,x_{p+q}}, &
        \text{$i = p$}.
    \end{cases}
\end{equation}
One can check that $\psi^d_{p,q}$ preserves all relations in \cref{def:Brauer}.
By multiplying the generators and using linearity, the action \eqref{eq:Brauer action} can be extended to the entire walled Brauer algebra $\B^d_{p,q}$.

The image of $\psi^d_{p,q}$ is a quotient algebra over the two-sided ideal $\ker(\psi^d_{p,q})$ of $\B^d_{p,q}$:
\begin{equation}
    \label{eq:Apqd}
    \A_{p,q}^d
    \defeq \psi^d_{p,q}(\B_{p,q}^d)
    % \cong \bigslant{\B_{p,q}^d}{\ker(\psi_{p,q}^d)},
    \cong \B_{p,q}^d \,/\, \ker(\psi_{p,q}^d),
\end{equation}
which we call the \emph{matrix algebra of partially transposed permutations}.
It is important to make a distinction between the algebras $\A^d_{p,q}$ and $\B^d_{p,q}$.
When the local dimension $d$ is small, i.e., $d < p + q$, those two algebras are not isomorphic.
As a finite-dimensional star matrix algebra, $\A^d_{p,q}$ is always semisimple
\cite{etingof2011introduction}.
On the other hand, $\B^d_{p,q}$ is not semisimple for integer $d < p + q - 1$ \cite{cox2008blocks}.

As an example, consider the walled Brauer algebras
$\B^d_{2,0}$ and $\B^d_{1,1}$, which are both generated by one generator, $\sigma_1$, which is a transposition in the first case and a contraction in the second.
The image of $\sigma_1$ under the map $\psi^d_{p,q}$ is shown in \cref{tab:TranspositionVSCOntraction}.

\begin{table}[!h]
\centering
\begin{tabular}[t]{c|cc}
\toprule
& Algebra $\B^d_{2,0}$
& Algebra $\B^d_{1,1}$ \\
\midrule
Diagram of the generator
& $\sigma_1 = \diagSWAP$
& $\sigma_1 = \diagContract$  \\
& (transposition) & (contraction) \\
&& \\
Action in $\A^d_{p,q}$
& $\psi^d_{2,0}(\sigma_1)\colon \ket{i} \ket{j} \mapsto \ket{j} \ket{i}$
& $\psi^d_{1,1}(\sigma_1)\colon \ket{i} \ket{j} \mapsto \delta_{i,j} \sum_{k=1}^{d} \ket{k} \ket{k}$ \\
&& \\
Matrix representation when $d=2$ &
  $\psi^2_{2,0}\of{\sigma_1}
  = \mx{
      1 & 0 & 0 & 0 \\
      0 & 0 & 1 & 0 \\
      0 & 1 & 0 & 0 \\
      0 & 0 & 0 & 1
    }$&
  $\psi^2_{1,1}\of{\sigma_1}
  = \mx{
      1 & 0 & 0 & 1 \\
      0 & 0 & 0 & 0 \\
      0 & 0 & 0 & 0 \\
      1 & 0 & 0 & 1
    }$\\
\bottomrule
\end{tabular}
\caption{\label{tab:TranspositionVSCOntraction}Comparison of algebras $\B^d_{2,0}$ and $\B^d_{1,1}$.
They both have a single generator $\sigma_1$ which is a transposition and a contraction, respectively.
The corresponding matrix algebras
$\A^d_{2,0}$ and $\A^d_{1,1}$
are two-dimensional and spanned by the identity matrix and $\psi^d_{p,q}\of{\sigma_1}$.}
\end{table}

\subsection{Representation theory of the algebra $\A^d_{p,q}$}\label{sec:rep_theory_A}

In this section, we describe the representation theory of the matrix algebra $\A^d_{p,q}$ of partially transposed permutations \cite{grinko2022linear}.
For this purpose, we will employ the notions of the Bratteli diagram and Gelfand--Tsetlin basis introduced in \cref{sec:BrattDiag}.

First, similar to the so-called Okounkov--Vershik approach, we consider the following sequence of algebra inclusions:
\begin{equation}
    \label{eq:A inclusions}
    \A_{0,0}^d \hookrightarrow \A_{1,0}^d \hookrightarrow \cdots \hookrightarrow \A_{p,0}^d \hookrightarrow \A_{p,1}^d \hookrightarrow \cdots \hookrightarrow \A_{p,q}^d,
\end{equation}
where each inclusion corresponds to tensoring the previous algebra with the identity matrix $I_d$ on the right or, equivalently, applying $\psi^d_{p,q}$ onto \cref{eq:B inclusions}.
All algebras in this sequence are semisimple, and their inclusions are multiplicity-free \cite{etingof2011introduction}.
For convenience, we will sometimes use the following single-subscript notation:
\begin{equation}
    \label{eq:A inclusions simplified}
    \A_{0}^d \hookrightarrow \A_{1}^d \hookrightarrow \cdots \hookrightarrow \A_{p+q}^d
\end{equation}
where the algebras $\A^d_k$ are defined as follows:
\begin{equation}
  \label{eq:Ak}
  \A^d_k \defeq
  \begin{cases}
    \C & \text{if } k = 0, \\
    \A^d_{k,0} & \text{if } 1 \leq k \leq p, \\
    \A^d_{p,k-p} & \text{if } p+1 \leq k \leq p+q.
  \end{cases}
\end{equation}

The irreducible representations of $\A_{p,q}^d$ are labelled by the following set of mixed Young diagrams $(\lambda_l,\lambda_r)$:
\begin{equation}
    \label{eq:IrrA}
    \Irr{\A_{p,q}^{d}} \defeq
    \set[\Big]{
    \lambda = (\lambda_l,\lambda_r) :
    0 \leq k \leq \min(p,q), \;
    \lambda_l \pt p - k, \;
    \lambda_r \pt q - k, \;
    \ell(\lambda_l)+\ell(\lambda_r) \leq d}.
\end{equation}
The Bratteli diagram corresponding to the sequence \eqref{eq:A inclusions} is as follows \cite{FusionProcedure}.
Its levels are labelled by $0,\dotsc,p+q$.
The vertices at level $s+t$, where $(s,t)$ ranges over subscripts in \cref{eq:A inclusions}, are labelled by $\Irr{\A_{s,t}^{d}}$.
An edge $\lambda \rightarrow \mu$ between $\lambda = (\lambda_l,\lambda_r) \in \Irr{\A_{s,t}^{d}}$ and $\mu = (\mu_l,\mu_r) \in \Irr{\A_{s',t'}^{d}}$ is present if and only if $\lambda$ and $\mu$ are in consecutive levels, i.e., $s' + t' = s + t + 1$, and
\begin{enumerate}
    \item if $s \leq p$, the diagram $\mu$ can be obtained from $\lambda$ by adding one cell to $\lambda_l$,
    \item if $s > p$, the diagram $\mu$ can be obtained from $\lambda$ by either adding one cell to $\lambda_r$ or removing one cell from $\lambda_l$.
\end{enumerate}
This is equivalent to saying that $\Paths(\lambda)$ in this Bratteli diagram correspond precisely to the set of mixed standard Young tableaux of shape $\lambda$ (see \cref{sec:mixed diagrams}).
For example, \cref{fig:A32} shows the Bratteli diagram for the sequence \eqref{eq:A inclusions} ending with algebra $\A^3_{3,2}$.

\begin{figure}
\begin{tikzpicture}[> = latex,
 MT/.style = {draw = blue!40, line width = 3pt},
  every node/.style = {inner sep = 1pt}]
  \def\W{2.0cm}
  \def\H{1cm}
  % Level labels
  \foreach \i/\p/\q in {0/0/0, 0.7/1/0, 1.5/2/0, 2.5/3/0, 3.5/3/1, 4.5/3/2} {
    \node at (\i*\W,3.2*\H) {$\A^3_{\p,\q}$};
  }
  \draw[dashed] (3*\W,3.4*\H) -- (3*\W,-2.7*\H);
  % Level 0
  \node (0!0)   at (0.0*\W, 0.0*\H) {$\br{\0,\0}$};
  % Level 1
  \node (1!0)   at (0.7*\W, 0.0*\H) {$\br{\yd{1},\0}$};
  % Level 2
  \node (2!0)   at (1.5*\W, 1.0*\H) {$\br{\yd{2},\0}$};
  \node (11!0)  at (1.5*\W,-1.0*\H) {$\br{\yd{1,1},\0}$};
  % Level 3
  \node (3!0)   at (2.5*\W, 1.5*\H) {$\br{\yd{3},\0}$};
  \node (21!0)  at (2.5*\W, 0.0*\H) {$\br{\yd{2,1},\0}$};
  \node (111!0) at (2.5*\W,-1.5*\H) {$\br{\yd{1,1,1},\0}$};
  % Level 4
  \node (3!1)   at (3.5*\W, 1.8*\H) {$\br{\yd{3},\yd{1}}$};
  \node (2!0')  at (3.5*\W, 0.5*\H) {$\br{\yd{2},\0}$};
  \node (21!1)  at (3.5*\W,-0.5*\H) {$\br{\yd{2,1},\yd{1}}$};
  \node (11!0') at (3.5*\W,-1.8*\H) {$\br{\yd{1,1},\0}$};
  % Level 5
  \node (3!2)   at (4.5*\W, 2.5*\H) {$\br{\yd{3},\yd{2}}$};
  \node (3!11)  at (4.5*\W, 1.5*\H) {$\br{\yd{3},\yd{1,1}}$};
  \node (2!1)   at (4.5*\W, 0.5*\H) {$\br{\yd{2},\yd{1}}$};
  \node (1!0')  at (4.5*\W,-0.5*\H) {$\br{\yd{1},\0}$};
  \node (21!2)  at (4.5*\W,-1.5*\H) {$\br{\yd{2,1},\yd{2}}$};
  \node (11!1)  at (4.5*\W,-2.5*\H) {$\br{\yd{1,1},\yd{1}}$};
  % Irrep dimensions
  \path (5.1*\W, 2.5*\H) node {$1$};
  \path (5.1*\W, 1.5*\H) node {$1$};
  \path (5.1*\W, 0.5*\H) node {$6$};
  \path (5.1*\W,-0.5*\H) node {$6$};
  \path (5.1*\W,-1.5*\H) node {$2$};
  \path (5.1*\W,-2.5*\H) node {$5$};
  % Edges 0 -> 1
  \draw[->] (0!0) -- (1!0);
  % Edges 1 -> 2
  \draw[->] (1!0.north east) -- (2!0.190);
  \draw[->] (1!0.south east) -- (11!0.170);
  % % Edges 2 -> 3
  \draw[->] (2!0.north east) -- (3!0.190);
  \draw[->] (2!0.south east) -- (21!0.170);
  \draw[->] (11!0.north east) -- (21!0.190);
  \draw[->] (11!0.south east) -- (111!0.170);
  % % Edges 3 -> 4
  \draw[->] (3!0) -- (3!1.180);
  \draw[->] (3!0) -- (2!0'.180);
  \draw[->] (21!0) -- (2!0'.-170);
  \draw[->] (21!0) -- (21!1.180);
  \draw[->] (21!0) -- (11!0'.140);
  \draw[->] (111!0) -- (11!0'.180);
  % % Edges 4 -> 5
  \draw[->] (3!1) -- (3!2);
  \draw[->] (3!1) -- (3!11);
  \draw[->] (3!1) -- (2!1);
  \draw[->] (2!0') -- (2!1);
  \draw[->] (2!0') -- (1!0');
  \draw[->] (21!1) -- (2!1);
  \draw[->] (21!1) -- (21!2);
  \draw[->] (21!1) -- (11!1);
  \draw[->] (11!0') -- (1!0');
  \draw[->] (11!0') -- (11!1);
  % Background layer
  \begin{scope}[on background layer]
      % Paths in \M{T}
      \draw[MT] (0!0) -- (1!0);
      \draw[MT] (1!0.north east) -- (2!0.190);
      \draw[MT] (2!0.north east) -- (3!0.190);
      \draw[MT] (2!0.south east) -- (21!0.170);
      \draw[MT] (3!0) -- (2!0'.180);
      \draw[MT] (21!0) -- (2!0'.-170);
      \draw[MT] (2!0') -- (2!1);
   \end{scope}
\end{tikzpicture}
\caption{The Bratteli diagram associated with the family~\eqref{eq:A inclusions} of partially transposed permutation matrix algebras when $p=3$, $q=2$, and local dimension $d=3$.
For a chosen path $T = \of[\big]{\br{\0,\0},\br{\yd{1},\0},\br{\yd{2},\0},\br{\yd{3},\0},\br{\yd{2},\0},\br{\yd{2},\yd{1}}}$, we have highlighted in blue the set $\M{T}$ of all paths that agree with $T$ everywhere except for level $p$.
For each leaf, we have indicated the number of paths from the root to that leaf, which coincides with the dimension of the corresponding irrep.}
\label{fig:A32}
\end{figure}

For any $\lambda \in \Irr{\A_{p,q}^d}$, we will denote the corresponding irrep by
\begin{equation}
    \psi_\lambda \colon \A_{p,q}^d \rightarrow \End\of[\big]{V_\lambda^{\A_{p,q}^d}}
    \quad\text{where}\quad
    V_\lambda^{\A_{p,q}^d} \defeq \C^{\Paths(\lambda)}.
\end{equation}
Our main technical result (see \cref{thm:main} below) provides an explicit formula for $\psi_\lambda(\sigma_i)$, for any generator $\sigma_i$ of $\A_{p,q}^d$. 
This effectively describes how the matrix algebra $\A_{p,q}^d$ acts on the Gelfand--Tsetlin basis vectors $\ket{T}$, where $T$ is any root-leaf path in the corresponding Bratteli diagram.

Before presenting our formula, we introduce some auxiliary notation.
Consider a path $T=(T^{0},\dotsc,T^{p+q}) \in \Paths(\lambda)$.
Similar to \cref{eq:conti}, we define the \emph{walled content} of a cell containing $i$ in $T$ as
\begin{equation}
    \wcont_i(T) \defeq
    \begin{cases}
        \cont(T^{i}_l \backslash T^{i-1}_l) &\text{if } i \leq p,\\
        \cont(T^{i}_r \backslash T^{i-1}_r)+d &\text{if } i > p \text{ and } T^{i}_l=T^{i-1}_l,\\
        -\cont(T^{i-1}_l \backslash T^{i}_l) &\text{if } i > p \text{ and } T^{i}_r=T^{i-1}_r.
    \end{cases}
\end{equation}
This definition is chosen so that $\wcont_i(T)$ matches the spectrum of Jucys--Murphy elements in \cref{lem:JMaction}.
For example, the path $T$ given in \cref{eq:T,example:mixed_tab} has the following values of walled content:
$\wcont_1(T) = 0$,
$\wcont_2(T) = 1$,
$\wcont_3(T) = 2$,
$\wcont_4(T) = -2$,
$\wcont_5(T) = d$.
We define the \emph{walled axial distance} between cells containing $i$ and $i+1$ in $T \in \Paths(\lambda)$ as
\begin{equation}
    r_i(T) \defeq \wcont_{i+1}(T)-\wcont_{i}(T).
    \label{eq:ri}
\end{equation}
The walled axial distance has a simple combinatorial interpretation. 
Indeed, $r_i(T ) $ is the axial distance between cells containing $i$ and $i + 1$ in a staircase representation $\widetilde{T}$ of a mixed Young diagram $T$, see \cref{fig:WalledConcatenation,example:mixed_tab}.

Furthermore, we denote by
\begin{equation}
    \M{T} \defeq
    \begin{cases}
        \set{
            (T^{0},\dotsc,T^{p-1},\mu,T^{p+1},\dotsc,T^{p+q}) \in \Paths(T^{p+q})
            \mid
            \mu \in \Irr{\A_{p,0}^{d}}
        }
        & \text{if $T^{p-1} = T^{p+1}$}, \\
        \0 & \text{otherwise}
    \end{cases}
    \label{eq:mobile elements}
\end{equation}
the set of all paths in the Bratteli diagram differing from $T$ only at the $p$-th level (see \cref{fig:A32} for an example).
For a given path $T$, we call the cell
\begin{equation}
    a_T \defeq
    \begin{cases}
        T^{p}_l \backslash T^{p-1}_l & \text{if $T^{p-1} = T^{p+1}$}, \\
        \0 & \text{otherwise},
    \end{cases}
\end{equation}
a \emph{mobile element}\footnote{The terms \emph{mobile cell} and \emph{mobile element} are from \cite{semikhatov2017quantum}.},
where $T_l^k$ denotes the left diagram in $T^k = (T^k_l,T^k_r)$.
For a mixed Young diagram $(\lambda_l,\lambda_r)$, we define the sets of \emph{removable/addable cells} based on the left tableau, i.e.,
$\RC((\lambda_l,\lambda_r)) \defeq \RC(\lambda_l)$ and
$\AC((\lambda_l,\lambda_r)) \defeq \AC(\lambda_l)$.
With this notation at hand, we can state our main technical result.

\begin{restatable}[Gelfand--Tsetlin basis for $\A^d_{p,q}$]{theorem}{mainthm}\label{thm:main}
For any $\lambda \in \Irr{\A_{p,q}^{d}}$, the following map
$\psi_\lambda \colon \A_{p,q}^d \rightarrow \End\of[\big]{\C^{\Paths(\lambda)}}$
is an irreducible representation of $\A_{p,q}^d$.
Given a generator $\sigma_i$ of $\A_{p,q}^d$,
$i = 1, \dotsc, p+q-1$,
the matrix $\psi_\lambda(\sigma_i)$ acts on the Gelfand--Tsetlin basis vectors $\ket{T}$ with $T \in \Paths(\lambda)$ as follows:
\begin{align}
    \psi_\lambda(\sigma_i) \, \ket{T}
    &= \frac{1}{r_i(T)} \, \ket{T} + \sqrt{1 - \frac{1}{r_i(T)^2}} \, \ket{\sigma_i T},
    & \text{for $i \neq p$},
    \label{gtbasis:transpositions}
    \\
    \psi_\lambda(\sigma_p) \, \ket{T}
    &= c(T) \, \ket{v_T}, \quad
    \ket{v_T} \defeq \sum_{T^\prime \in \M{T}} c(T^\prime) \, \ket{T^\prime},
    & \text{for $i = p$},
\end{align}
where
$r_i(T)$ is the walled axial distance defined in \cref{eq:ri},
$\sigma_i T$ denotes the mixed standard Young tableau $T$ with cell fillings permuted according to $\sigma_i$ (see \cref{sec:mixed diagrams}),
and the coefficient $c(T) \in \R$ is given by
\begin{equation}
    c(T) \defeq \sqrt{
        \of[\big]{d+\cont(a_T)}
        \frac{
            \prod_{c \in \RC(T^{p-1})}
            \of[\big]{\cont(a_T)-\cont(c)}
        }{
            \prod_{a \in \AC(T^{p-1}) \setminus a_T}
            \of[\big]{\cont(a_T)-\cont(a)}
        }
    }
\end{equation}
where $a_T$ is the mobile element of $T$, and $\RC$/$\AC$ are the sets of removable/addable cells.
\end{restatable}

\begin{remark}
Due to \cref{lem:content_into_sym_u_dims,lem:ration_into_u_dims} the coefficient $c(T)$ presented in \cref{thm:main} has the following combinatorial expression:
\begin{equation}
 c(T) =
 \sqrt{ \frac{ m_{T^{p}}}{ m_{T^{p-1}} }},
\end{equation}
where $m_{\lambda}$ is the dimension of the irreducible representation of the unitary group $\U{d}$ corresponding to the Young diagram $\lambda$.
It is well-known that the dimension $m_{\lambda}$ equals the number of semistandard Young tableaux $\SSYT(\lambda,d)$ of shape $\lambda$ and entries in $[d]$, see \cref{eq:hook content formula}.
\end{remark}

\noindent
We delegate the proof of \cref{thm:main} to \cref{proof_main_gt_theorem}.
\Cref{tab:A32} gives an example of how \cref{thm:main} can be used to compute all irreps of $\A^3_{3,2}$ using the Bratteli diagram shown in \cref{fig:A32}.
Moreover, we provide a code in \textit{Wolfram Mathematica} to generate the Gelfand--Tsetlin basis for arbitrary $\A_{p,q}^d$ \cite{github}.

%%%%%%%%%%%%%%%%%%%%%%%%%%%%%%%%%%%%%%%%%%%%%%%%%%%%%%%%%%%%%%%
%%%%%%%%%%%%%%%%%%%%%%%%%%%%%%%%%%%%%%%%%%%%%%%%%%%%%%%%%%%%%%%

\begin{table}
\begin{tabular}{c|cccc}
  % Irreps \backslash generators
  & $\sigma_1$ & $\sigma_2$ & $\sigma_3$ & $\sigma_4$ \\ \hline

  $\br{\yd{3},\yd{2}}$   & $\mx{1}$ & $\mx{1}$ & $\mx{0}$ & $\mx{ 1}$ \\[2pt]

  $\br{\yd{3},\yd{1,1}}$ & $\mx{1}$ & $\mx{1}$ & $\mx{0}$ & $\mx{-1}$ \\[2pt]

  $\br{\yd{2},\yd{1}}$ &

    $\smx{
    	 -1 & 0 & 0 & 0 & 0 & 0 \\
    	 0 & -1 & 0 & 0 & 0 & 0 \\
    	 0 & 0 & 1 & 0 & 0 & 0 \\
    	 0 & 0 & 0 & 1 & 0 & 0 \\
    	 0 & 0 & 0 & 0 & 1 & 0 \\
    	 0 & 0 & 0 & 0 & 0 & 1
    }$ &
	$\smx{
	   \frac{1}{2} & 0 & \frac{\sqrt{3}}{2} & 0 & 0 & 0 \\
	   0 & \frac{1}{2} & 0 & \frac{\sqrt{3}}{2} & 0 & 0 \\
	   \frac{\sqrt{3}}{2} & 0 & -\frac{1}{2} & 0 & 0 & 0 \\
	   0 & \frac{\sqrt{3}}{2} & 0 & -\frac{1}{2} & 0 & 0 \\
	   0 & 0 & 0 & 0 & 1 & 0 \\
	   0 & 0 & 0 & 0 & 0 & 1
	}$ &
	$\smx{
	   0 & 0 & 0 & 0 & 0 & 0 \\
	   0 & 0 & 0 & 0 & 0 & 0 \\
	   0 & 0 & \frac{4}{3} & 0 & \frac{2 \sqrt{5}}{3} & 0 \\
	   0 & 0 & 0 & 0 & 0 & 0 \\
	   0 & 0 & \frac{2 \sqrt{5}}{3} & 0 & \frac{5}{3} & 0 \\
	   0 & 0 & 0 & 0 & 0 & 0
	}$ &
	$\smx{
	   \frac{1}{2} & \frac{\sqrt{3}}{2} & 0 & 0 & 0 & 0 \\
	   \frac{\sqrt{3}}{2} & -\frac{1}{2} & 0 & 0 & 0 & 0 \\
	   0 & 0 & \frac{1}{2} & \frac{\sqrt{3}}{2} & 0 & 0 \\
	   0 & 0 & \frac{\sqrt{3}}{2} & -\frac{1}{2} & 0 & 0 \\
	   0 & 0 & 0 & 0 & \frac{1}{5} & \frac{2 \sqrt{6}}{5} \\
    	 0 & 0 & 0 & 0 & \frac{2 \sqrt{6}}{5} & -\frac{1}{5}
	}$

  \\

  $\br{\yd{1},\0}$ &

	$\smx{
	   -1 & 0 & 0 & 0 & 0 & 0 \\
	   0 & -1 & 0 & 0 & 0 & 0 \\
	   0 & 0 & -1 & 0 & 0 & 0 \\
	   0 & 0 & 0 & 1 & 0 & 0 \\
	   0 & 0 & 0 & 0 & 1 & 0 \\
	   0 & 0 & 0 & 0 & 0 & 1
	}$ &
	$\smx{
	   -1 & 0 & 0 & 0 & 0 & 0 \\
	   0 & \frac{1}{2} & 0 & \frac{\sqrt{3}}{2} & 0 & 0 \\
	   0 & 0 & \frac{1}{2} & 0 & \frac{\sqrt{3}}{2} & 0 \\
	   0 & \frac{\sqrt{3}}{2} & 0 & -\frac{1}{2} & 0 & 0 \\
	   0 & 0 & \frac{\sqrt{3}}{2} & 0 & -\frac{1}{2} & 0 \\
	   0 & 0 & 0 & 0 & 0 & 1
	}$ &
	$\smx{
	   \frac{1}{3} & \frac{2 \sqrt{2}}{3} & 0 & 0 & 0 & 0 \\
	   \frac{2 \sqrt{2}}{3} & \frac{8}{3} & 0 & 0 & 0 & 0 \\
	   0 & 0 & 0 & 0 & 0 & 0 \\
	   0 & 0 & 0 & 0 & 0 & 0 \\
	   0 & 0 & 0 & 0 & \frac{4}{3} & \frac{2 \sqrt{5}}{3} \\
	   0 & 0 & 0 & 0 & \frac{2 \sqrt{5}}{3} & \frac{5}{3}
	}$ &
	$\smx{
	   -1 & 0 & 0 & 0 & 0 & 0 \\
	   0 & \frac{1}{2} & \frac{\sqrt{3}}{2} & 0 & 0 & 0 \\
	   0 & \frac{\sqrt{3}}{2} & -\frac{1}{2} & 0 & 0 & 0 \\
	   0 & 0 & 0 & \frac{1}{2} & \frac{\sqrt{3}}{2} & 0 \\
	   0 & 0 & 0 & \frac{\sqrt{3}}{2} & -\frac{1}{2} & 0 \\
	   0 & 0 & 0 & 0 & 0 & 1
	}$

  \\

  $\br{\yd{2,1},\yd{2}}$ &

    $\mx{-1 & 0 \\ 0 & 1}$ &
    $\smx{
       \frac{1}{2} & \frac{\sqrt{3}}{2} \\
       \frac{\sqrt{3}}{2} & -\frac{1}{2}
     }$ &
    $\mx{0 & 0 \\ 0 & 0}$ &
    $\mx{1 & 0 \\ 0 & 1}$

  \\

  $\br{\yd{1,1},\yd{1}}$ &

	$\smx{
    	 -1 & 0 & 0 & 0 & 0 \\
    	 0 & -1 & 0 & 0 & 0 \\
    	 0 & 0 & -1 & 0 & 0 \\
    	 0 & 0 & 0 & 1 & 0 \\
    	 0 & 0 & 0 & 0 & 1
	}$ &
	$\smx{
    	 -1 & 0 & 0 & 0 & 0 \\
    	 0 & \frac{1}{2} & 0 & \frac{\sqrt{3}}{2} & 0 \\
    	 0 & 0 & \frac{1}{2} & 0 & \frac{\sqrt{3}}{2} \\
    	 0 & \frac{\sqrt{3}}{2} & 0 & -\frac{1}{2} & 0 \\
    	 0 & 0 & \frac{\sqrt{3}}{2} & 0 & -\frac{1}{2}
	}$ &
	$\smx{
    	 \frac{1}{3} & \frac{2 \sqrt{2}}{3} & 0 & 0 & 0 \\
    	 \frac{2 \sqrt{2}}{3} & \frac{8}{3} & 0 & 0 & 0 \\
    	 0 & 0 & 0 & 0 & 0 \\
    	 0 & 0 & 0 & 0 & 0 \\
    	 0 & 0 & 0 & 0 & 0
	}$ &
	$\smx{
    	 1 & 0 & 0 & 0 & 0 \\
    	 0 & \frac{1}{4} & \frac{\sqrt{15}}{4} & 0 & 0 \\
    	 0 & \frac{\sqrt{15}}{4} & -\frac{1}{4} & 0 & 0 \\
    	 0 & 0 & 0 & \frac{1}{4} & \frac{\sqrt{15}}{4} \\
    	 0 & 0 & 0 & \frac{\sqrt{15}}{4} & -\frac{1}{4}
	}$

\end{tabular}
\caption{\label{tab:A32}Example of \cref{thm:main} in action: the matrices $\psi_\lambda(\sigma_i)$ in the Gelfand--Tsetlin basis for all irreps $\psi_\lambda$ and generators $\sigma_i$ of the matrix algebra $\A^3_{3,2}$.
Rows and columns correspond to irrep labels $\lambda = (\lambda_l,\lambda_r)$ and generators $\sigma_i$, respectively.
Note that irreps with $|\lambda_l| = 3$ and $|\lambda_r| = 2$ vanish on the contraction $\sigma_3$ since they are irreps of $\S_3 \times \S_2$.
The Bratteli diagram for $\A^3_{3,2}$ is shown in \cref{fig:A32}.}
\end{table}

%%%%%%%%%%%%%%%%%%%%%%%%%%%%%%%%%%%%%%%%%%%%%%%%%%%%%%%%%%%%%%%
%%%%%%%%%%%%%%%%%%%%%%%%%%%%%%%%%%%%%%%%%%%%%%%%%%%%%%%%%%%%%%%

\section{Mixed quantum Schur transform}\label{sec:mixed Schur}

\subsection{Mixed Schur--Weyl duality}\label{sec:mixed SW}

As before, let $V_d^p \defeq (\C^d)\xp{p}$ and $V_d^{p,q} \defeq V_d^p \x (V_d^q)^\ast$.
Consider the natural action $\phi^d_{p,q} \colon \U{d} \to \End(V_d^{p,q})$ of the unitary group $\U{d}$ on $V_d^{p,q}$ given by
\begin{equation}
    \label{eq:UnitaryAction2}
    \phi^d_{p,q}(U) \, \ket{x_1,\dotsc,x_{p+q}}
    =
    \of[\big]{U\xp{p} \x \bar{U}\xp{q}} \,
    \ket{x_1,\dotsc,x_{p+q}}
\end{equation}
where $U \in \U{d}$, $\bar{U}$ is the complex conjugate of $U$, and $x_1, \dotsc, x_{p+q} \in [d]$.
This action is relevant in several contexts in quantum information, such as unitary-equivariant quantum channels with $p$ input and $q$ output qudits (see \cref{sec:Unitary-equivariant quantum channels}).
The action \eqref{eq:UnitaryAction2} commutes with the aforementioned action $\psi^d_{p,q} \colon \B_{p,q}^d \to \End(V_d^{p,q})$ from \cref{eq:Brauer action} which defines the algebra $\A_{p,q}^d$ of partially transposed permutations on $V_d^{p,q}$.\footnote{This is a consequence of the fact that the maximally entangled state $\frac{1}{\sqrt{d}} \sum_{k \in [d]} \ket{k,k}$ is invariant under $U \x \bar{U}$ for any $U \in \U{d}$. On the other side, the permutation of registers between spaces $V_d^p$ and $(V_d^p)^\ast$ does not commute with action \cref{eq:UnitaryAction2}.}

The actions $\psi^d_{p,q}$ and $\phi^d_{p,q}$ of the matrix algebra $\A_{p,q}^d$ and the unitary group $\U{d}$ commute, i.e.,
\begin{equation}
    \sof[\big]{
        \psi^d_{p,q}(\sigma),
        \phi^d_{p,q}(U)
    } = 0,
    \quad\text{for all}\quad
    \sigma \in \A_{p,q}^d
    \text{ and }
    U \in \U{d}.
\end{equation}
Moreover, the space $V_d^{p,q}$ decomposes into a direct sum of irreducible modules of both actions:
\begin{equation}
    \label{eq:mSW}
    V_d^{p,q}
    \cong
    \bigoplus_{\lambda \in \Irr{\A_{p,q}^{d}}} V_\lambda^{\A_{p,q}^d} \x V_\lambda^{\U{d}}
    \quad\text{where}\quad
    V_\lambda^{\A_{p,q}^d} \defeq \C^{\Paths(\lambda)}
    \text{ and }
    V_\lambda^{\U{d}} \defeq \C^{\GT(\lambda,d)}.
\end{equation}
Here the direct sum ranges over all irreducible representations $\lambda$ of $A_{p,q}^d$, see \cref{eq:IrrA}, and the spaces $V_\lambda^{\A_{p,q}^d}$ and $V_\lambda^{\U{d}}$ correspond to irreducible representations of $\A_{p,q}^d$ and $\U{d}$, respectively.
We will denote their dimensions by
\begin{align}
    d_\lambda
    &\defeq \dim V_\lambda^{\A_{p,q}^d}
    = \abs{\Paths(\lambda)},
    \label{eq:d_lambda} \\
    m_{\lambda}
    &\defeq \dim V_\lambda^{\U{d}}
    = \abs{\GT(\lambda,d)}.
    \label{eq:m_lambda}
\end{align}
Both these irrep dimensions implicitly depend on the local dimension $d$.\footnote{For $d_\lambda$ the $d$-dependence is more subtle since it results from truncating the Bratteli diagram. Recall from \cref{eq:IrrA} that for small values of $d$ some vertices are removed from the Bratteli diagram, which can result in fewer paths reaching a given leaf. This phenomenon does not occur when $q = 0$, hence in the regular Schur--Weyl duality the $\S_p$ irrep dimension $d_\lambda$ does not depend on $d$.}
By comparing the dimensions on both sides of \cref{eq:mSW} we obtain a non-trivial combinatorial identity
\begin{equation}
    d^{p+q}
    = \sum_{\lambda \in \Irr{\A_{p,q}^{d}}}
    d_\lambda \, m_{\lambda},
\end{equation}
which is a consequence of the Robinson--Schensted--Knuth correspondence for mixed tensors \cite[Section~4]{stembridge1987rational}.

\Cref{eq:mSW} is known as \emph{mixed Schur--Weyl duality} and is usually expressed in the following more abstract way: the algebras generated by the actions $\psi^d_{p,q}$ and $\phi^d_{p,q}$ of $\A_{p,q}^d$ and $\U{d}$ on $V^{p,q}_d$ are complete mutual centralizers of each other within $\End(V^{p,q}_d)$ \cite{koike1989decomposition,bchlls}, see also \cite{grinko2022linear}.
The decomposition~\eqref{eq:mSW} follows from the double centralizer theorem \cite[Theorem~4.54]{etingof2011introduction} applied to this pair of algebras.

\subsection{Mixed Schur transform}\label{sec:SchTrans}

We would like to understand in more detail the basis change~\eqref{eq:mSW} that simultaneously block-diagonalizes both actions.
This is a unitary transformation
\begin{equation}
    \label{eq:Sch}
    \UschPQ \colon
    V^{p,q}_d
    \rightarrow
    \bigoplus_{\lambda \in \Irr{\A_{p,q}^{d}}}
    V_\lambda^{\A_{p,q}^d} \x V_\lambda^{\U{d}}
\end{equation}
known as \emph{mixed Schur transform} \cite{grinko2022linear}.
It maps the computational basis to a new basis composed of irreducible representations of $\A_{p,q}^d$ and $\U{d}$:
\begin{align}
    \UschPQ \; \psi^d_{p,q}(\sigma) \; \UschPQ\ct
    &= \bigoplus_{\lambda \in \Irr{\A_{p,q}^{d}}}
    \psi_\lambda(\sigma) \x I_{m_\lambda}, \\
    \UschPQ \; \phi^d_{p,q}(U) \; \UschPQ\ct
    &= \bigoplus_{\lambda \in \Irr{\A_{p,q}^{d}}}
    I_{d_\lambda} \x \phi_\lambda(U),
\end{align}
where $\psi^d_{p,q}$ and $\phi^d_{p,q}$ are defined in \cref{eq:UnitaryAction2,eq:Brauer action},
$\psi_\lambda$ is described in \cref{thm:main},
and $\phi_\lambda$ is an 

Recall from \cref{eq:mSW} that we can label the standard bases of $V_\lambda^{\A_{p,q}^d}$ and $V_\lambda^{\U{d}}$
by paths $T \in \Paths(\lambda)$ and Gelfand--Tsetlin patterns $M \in \GT(\lambda,d)$, respectively:
\begin{align}
    V_\lambda^{\A_{p,q}^d}
    &\defeq \spn_\C \set{\ket{T} : T \in \Paths(\lambda)}, \\
    V_\lambda^{\U{d}}
    &\defeq \spn_\C \set{\ket{M} : M \in \GT(\lambda,d)}.
\end{align}
When $q = 0$ we can equivalently think of $T$ and $M$ as a standard and a semistandard Young tableau, respectively (see \cref{sec:tableaux,sec:GT}).
For general values of $p,q$, one can interpret $T$ as a mixed standard Young tableau (see \cref{sec:mixed diagrams}) while $M$ can still be interpreted as a semistandard tableau by adding an appropriate constant to all entries of the Gelfand--Tsetlin pattern so that they become non-negative.
Note that the irrep label $\lambda \in \Irr{\A_{p,q}^{d}}$ is implicit in both $T$ and $M$.
Indeed, it can be recovered from the final vertex of the path $T$ as well as from the first row of the pattern $M$.
Hence, we can treat the output space in \cref{eq:Sch} as a formal linear span of all pairs $(T,M)$:
\begin{equation}
    \bigoplus_{\lambda \in \Irr{\A_{p,q}^{d}}}
    V_\lambda^{\A_{p,q}^d} \x V_\lambda^{\U{d}}
    = \spn_\C \set[\big]{\ket{(T,M)} : \lambda \in \Irr{\A_{p,q}^{d}}, T \in \Paths(\lambda), M \in \GT(\lambda,d)}.
\end{equation}
Note that the output space of $\UschPQ$ does \emph{not} have a tensor product structure, i.e., one should not treat $\ket{(T,M)}$ as $\ket{T} \x \ket{M}$.
Nevertheless, for each individual $\lambda$ the corresponding subspace $\C^{d_\lambda} \x \C^{m_\lambda}$ is indeed a tensor product.

We would like to find an efficient way of computing the mixed Schur transform matrix entries
\begin{equation}
    \bra{(T,M)} \, \UschPQ \, \ket{x_1,\dotsc,x_{p+q}}.
    \label{eq:Schur matrix entries}
\end{equation}
Here the rows are labelled by pairs $(T,M)$, where
$T \in \Paths(\lambda)$ and
$M \in \GT(\lambda,d)$
for all choices of
$\lambda \in \Irr{\A_{p,q}^{d}}$,
while the columns are labelled by strings
$x_1,\dotsc,x_{p+q} \in [d]$.\footnote{Note that \cref{eq:Sch} is a passive transformation, as a change of coordinates system from the computational basis vectors into the basis labelled by tuples $(T,M)$. As the number of such tuples matches the dimension of $V_d^{p,q}$ space, it can be seen as a transformation of $V_d^{p,q}$ onto itself, after applying any bijection between tuples $(T,M)$ and computational basis vectors. As such, the Schur transform is a unitary transformation (also active transformation).}%
$^{,}$%
\footnote{Note that Schur transform, defined as a transformation (\ref{eq:Sch}) which simultaneously decomposes $V_d^{p,q}$ into irreducible modules of $\A_{p,q}^d$ and $\U{d}$ is not uniquely defined. Indeed, arbitrary change of basis within modules $V_\lambda^{\A_{p,q}^d}$ and $ V_\lambda^{\U{d}}$ does not affect decomposition (\ref{eq:Sch}). On the other hand, such a change of basis within modules is the only degree of freedom for the Schur transform. In this paper, we present the most common form of Schur transform related to a Gelfand--Tsetlin basis corresponding to a sequence of algebras (\ref{eq:A inclusions}) for $\A_{p,q}^d$ and sequence $\U{1}\hookrightarrow \cdots\hookrightarrow \U{d}$ for unitary group $U(d)$. Moreover, for $d=2$, it can be seen as a consecutive composition of spin-$\tfrac{1}{2}$ particles into a system with well-defined global spin.}
The matrix entries \eqref{eq:Schur matrix entries} can be arranged into a matrix as follows:
\begin{equation}
    \UschPQ =
    \sum_{\lambda \in \Irr{\A_{p,q}^{d}}}
    \sum_{T \in \Paths(\lambda)}
    \sum_{M \in \GT(\lambda,d)}
    \sum_{x_1,\dotsc,x_{p+q} \in [d]}
    \bra{(T,M)} \, \UschPQ \, \ket{x_1,\dotsc,x_{p+q}} \cdot
    \ketbra{(T,M)}{x_1,\dotsc,x_{p+q}}.
\end{equation}
In particular, the mixed quantum Schur transform of any standard basis vector $\ket{x_1,\dotsc,x_{p+q}}$ is
\begin{equation}
    \label{eq:Sch2}
    \UschPQ \, \ket{x_1,\dotsc,x_{p+q}}
    =
    \sum_{\lambda \in \Irr{\A_{p,q}^{d}}}
    \sum_{T \in \Paths(\lambda)}
    \sum_{M \in \GT(\lambda,d)}
    \bra{(T,M)} \, \UschPQ \, \ket{x_1,\dotsc,x_{p+q}} \cdot
    \ket{(T,M)},
\end{equation}
while the Schur basis vectors $\bra{(T,M)}$ can be expressed as
\begin{equation}
    \label{eq:Schur basis vector}
    \bra{(T,M)}
    =
    \sum_{x_1,\dotsc,x_{p+q} \in [d]}
    \bra{(T,M)} \, \UschPQ \, \ket{x_1,\dotsc,x_{p+q}} \cdot
    \bra{x_1,\dotsc,x_{p+q}}.
\end{equation}

The most common way of implementing the regular $q = 0$ Schur transform is by a sequence of Clebsch--Gordan transforms \cite{HarrowThesis,bch2006quantumschur,kirby2018practical}.
Since this approach can be generalized to any $q \geq 0$ in a straightforward way, we can derive an explicit formula for the matrix entries of the mixed quantum Schur transform by inductively decomposing it as a sequence of Clebsch--Gordan transforms.

We start with a single qudit in state $\ket{x_1}$.
At each step $k = 2, \dotsc, p+q$ we use the Clebsch--Gordan transform $\CG^k \in \U{d^k}$ to couple the output state from the previous iteration with an additional qudit in state $\ket{x_k}$.
The input and output spaces of $\CG^k$ can be decomposed as follows:
\begin{equation}
    \label{def:CGk_map}
    \CG^k:
    \of[\bigg]{
        \bigoplus_{\lambda \in \Irr{\A_{k-1}^{d}}}
        V_\lambda^{\A_{k-1}^d} \x V_\lambda^{\U{d}}
    }
    \x \C^d
    \to
    \bigoplus_{\mu \in \Irr{\A_{k}^{d}}}
    V_\mu^{\A_{k}^d} \x V_\mu^{\U{d}},
\end{equation}
where we are using the single-subscript convention \eqref{eq:Ak} for the sequence of algebras $\A_k^d$ from \cref{eq:A inclusions simplified}.
For any path $T = (T^{0},\dotsc,T^{k-1})$ and pattern $M \in \GT(T^{k-1},d)$ the action of $\CG^k$ is defined as
\begin{equation}
    \label{def:CG_classical}
    \CG^{k} \of[\Big]{\ket{(T,M)} \otimes \ket{x_k}}
    \defeq \sum_{\mu : T^{k-1} \rightarrow \mu}
    \sum_{\substack{N \in \GT(\mu,d) \\ w(N) = w(M) + w(x_k)}}
    c_{N,M}^{x_k} \ket{(T \circ \mu, N)},
\end{equation}
where $T^{k-1} \rightarrow \mu$ means that $\mu$ can be reached from $T^{k-1}$ by one step in the Bratteli diagram,
and $c_{N,M}^{x_k}$ are the so-called \emph{Clebsch--Gordan coefficients} (see \cref{sec:CGcoefficients} for an explicit formula).
The mixed quantum Schur transform $\UschPQ$ is given by a cascade of $p+q-1$ Clebsch--Gordan transforms:
\begin{equation}
    \label{def:Usch_classical}
    \UschPQ = \CG^{p+q} \cdot
    \of*{\CG^{p+q-1} \otimes I} \cdots
    \of*{\CG^{3} \otimes I\xp{p+q-3}} \cdot
    \of*{\CG^{2} \otimes I\xp{p+q-2}}.
\end{equation}
In particular, the usual Schur transform satisfies
$\Usch(k,0) = \CG^k \of[\big]{\Usch(k-1,0) \x I}$
where $k = 2, \dotsc, p$.

Let us use \cref{def:Usch_classical} to derive an explicit formula for the matrix entries of the mixed quantum Schur transform.
For an arbitrary path
$T = (T^{0},\dotsc,T^{p+q}) \in \Paths(\lambda)$
and Gelfand--Tsetlin pattern
$M \in \GT(\lambda,d)$,
we get from \cref{def:CG_classical,def:Usch_classical} that
\begin{equation}
    \label{eq:USch3}
    \bra{(T,M)} \, \UschPQ \, \ket{x_1,\dotsc,x_{p+q}} =
        \bra{M}
        C_{T^{p+q} T^{p+q-1}}^{x_{p+q}}
        \cdots
        C_{T^{2} T^{1}}^{x_{2}}
        \ket{x_1}
\end{equation}
where $\ket{M}$ is a Gelfand--Tsetlin basis vector of the unitary group irrep corresponding to the staircase $\lambda$,
and $C_{T^{k} T^{k-1}}^{x_{k}}$ for $k = 2, \dotsc, p+q$ are rectangular matrices with rows labelled by $N \in \GT(T^{k},d)$ and columns labelled by $M \in \GT(T^{k-1},d)$, with the corresponding matrix entry equal to the Clebsch--Gordan coefficient:
\begin{equation}
    \label{eq:cg_matrix_entires}
    \bra{N} C_{T^{k} T^{k-1}}^{x_{k}} \ket{M} \defeq c_{N,M}^{x_k}.
\end{equation}
Hence, according to \cref{eq:USch3}, any matrix entry of $\UschPQ$ can be computed by applying a sequence of matrices onto $\ket{x_1}$.
The complexity of this computation depends on the dimensions of the matrices $C_{T^{k} T^{k-1}}^{x_{k}}$.

In practice, one can take advantage of the fact that the matrices $C_{T^{k} T^{k-1}}^{x_{k}}$ are block-diagonal. By tracking which blocks contribute non-trivially to a given $\bra{(T,M)}$, \cref{eq:USch3} can be modified as follows:
\begin{equation}
    \label{eq:USch4}
    \bra{(T,M)} \, \UschPQ \, \ket{x_1,\dotsc,x_{p+q}} =
        \bra{M}
        C_{T^{p+q} T^{p+q-1}}^{x_{p+q}, w(x_{p+q-1},\dotsc,x_1)}
        C_{T^{p+q-1} T^{p+q-2}}^{x_{p+q-1}, w(x_{p+q-2},\dotsc,x_1)}
        \cdots
        C_{T^{2} T^{1}}^{x_{2}}
        \ket{x_1},
\end{equation}
where $C_{T^{k} T^{k-1}}^{x_{k}, w(x_{k-1},\dotsc,x_1)}$ are submatrices of $C_{T^{k} T^{k-1}}^{x_{k}}$ with rows labelled only by $N \in \GT(T^{k},d)$ of weight $w(N) = w(x_k,\dotsc,x_1)$ and columns labelled only by $M \in \GT(T^{k-1},d)$ of weight $w(M) = w(x_{k-1},\dotsc,x_1)$.

%%%%%%%%%%%%%%%%%%%%%%%%%%%%%%%%%%%%%%%%%%%%%%%%%%%%%%%%%%%%%%%%%%%%%%%%%%%%%%%%%%%%%%%%%%%%%%%%%%%%%%%%%%%%%%%%%%%%%%%%%%%%%%%%%%%%%%%%%%%%%%%%%%%%%%%%%%%%%%%

\subsection{MPS representation of mixed Schur basis vectors}\label{sec:MPS}

Notice that \cref{eq:USch3} presents the Schur basis vectors $\ket{(T,M)}$ as matrix product states (MPS) with bond dimensions given by $D_k \defeq |\GT(T^{k},d)|$, see \cref{fig:MPS}. For fixed local dimension $d$, bond dimensions $D_k$ are upper-bounded by $(p+q)^{O(d^2)}$.%
\footnote{This can be easily seen by counting the number of Gelfand--Tsetlin patterns $\GT(T^{k},d)$. Indeed, for any pair of Young diagrams $\lambda=(\lambda_l,\lambda_r)$ satisfying $\ell (\lambda_l)+\ell (\lambda_r) \leq d$, and of size $|\lambda_l|\leq p$, $|\lambda_r|\leq q$, a corresponding set of Gelfand--Tsetlin patterns $\GT(\lambda,d)$ consists of patterns (\ref{eq:GTpattern}) with $\tfrac{d(d+1)}{2}$ entries $m_{ij}$, see \cref{eq:GTpatternOfAPair}. Notice that $m_{1d}=\lambda_{l,1}$ and $m_{dd}=-\lambda_{r,1}$. Obviously, $|\lambda_{l,1}|\leq|\lambda|\leq p$, and $|\lambda_{r,1}|\leq|\lambda|\leq q$. Therefore, $m_{1d}\leq p$ and $m_{dd}\geq -q$. It is easy to see, that all entries $m_{ij}$ of the Gelfand--Tsetlin pattern are bounded by $m_{dd}\leq m_{ij} \leq m_{1d}$, hence for arbitrary entry of the pattern $m_{ij}$, we have $-q\leq m_{ij} \leq p$. As the number of entries is $\tfrac{d(d+1)}{2}$, the size of the set $\GT(\lambda,d)$ can be upper bound by $(p+q)^{O(d^2)}$.}
Note that the length of an MPS is given by $(p+q)$, and bond dimensions $D_k$ are upper bounded by $(p+q)^{O(d^2)}$.%
\footnote{In principle, for $d>p+q$, the scaling should be $(p+q)^{O((p+q)^2)}$ by adapting a strategy outlined in \cite{HarrowThesis}. Indeed, in such case, any semistandard Young tableau $M \in \SSYT(\lambda,d)$ of size $|\lambda|=p+q$ and with letters in $[d]$ has effectively at most $p+q$ different symbols. Therefore, it can be encoded as $M' \in \SSYT(\lambda,n)$ by some encoding function $f_{\text{enc}}:[n]\rightarrow [d]$ which takes the values of the original tableau $M$.}
As a consequence, for fixed local dimension $d$ the computational complexity of computing $\bra{(T,M)} \, \UschPQ \, \ket{x_1,\dotsc,x_{p+q}}$ is upper bounded by
\begin{equation}
    \label{eq:upperb}
    (p+q)^{O(d^2)}.
\end{equation}
Not that the complexity of computing the entries of matrices $C_{T^{k} T^{k-1}}^{x_{k}}$ (see \cref{sec:CGcoefficients}) is absorbed into the above bound.
Thus, we have established the following result.

\begin{theorem}[MPS representation of mixed Schur basis vectors]\label{thm:MPS}
    The Schur basis states $\ket{(T,M)}$ (or rows of mixed quantum Schur transform $\UschPQ$) admit a matrix product state representation with bond dimension $(p+q)^{O(d^2)}$.
    Hence, the matrix entries $\bra{(T,M)} \, \UschPQ \, \ket{x_1,\dotsc,x_{p+q}}$ of mixed Schur transform $\Usch$ can be computed in time $(p+q)^{O(d^2)}$.
    In particular, for constant local dimension $d$, this is polynomial in the system size $p+q$.
\end{theorem}

Moreover, notice that \cref{eq:USch4} provides a more refined approach for computing entries of vectors $\ket{(T,M)}$ in the computational basis. Although it is no longer an MPS, as the choice of consecutive matrices $C_{T^{k} T^{k-1}}^{x_{k}, w(x_{k-1},\dotsc,x_1)}$ does not depend on $x_k$ only. notice that matrices in this product are of the size given by $ |\{M\in \GT(T^{k},d): w(M)=w(x_1,\dotsc,x_k)\}|$ the number of Gelfand--Tsetlin patterns with a given weight $w(x_1,\dotsc,x_k)$. This number is known as a \emph{Kostka number} $K_{T^{k},w(x_1,\dotsc,x_k)}$ depending on two integer partitions $T^{k}$ and $w(x_1,\dotsc,x_k)$.
Clearly the size of those matrices is smaller than in \cref{eq:USch3}, however, we did not find an improvement in the asymptotic upper bound. As previously, $K_{T^{k},w(x_1,\dotsc,x_k)}$ can be upper bounded by $(p+q)^{O(d^2)}$, which leads to the same upper bound (\ref{eq:upperb}).

\begin{figure}[t]
\begin{tikzpicture}[
amp/.style = {regular polygon, regular polygon sides=3,
              draw, fill=white, text width=0.8em,
              inner sep=0.4mm, outer sep=0mm,
              shape border rotate=90}
                        ]
\begin{scope}[shift={(0,3.5)}]
\node (amp) at (-2.5,0) {$M$};
\node[rectangle,minimum width=4em,draw] (vp+q) at (0,0) {$C_{T^{p+q} T^{p+q-1}}$};
 \draw [thick] (amp.east) -- (vp+q.west);
\node (ip+q) at (0,1.5) {};
\draw (vp+q.north) -- (ip+q) ;
%  \draw (vp+q.north) -- ++ (5,0.5) ;
 \node (ivp+q) at (0,1.7) {$x_{p+q}$};
\draw [thick] (vp+q.east) -- ++ (0.5,0);
\node[rectangle,minimum width=4em,draw] (v2) at (5,0) {$C_{T^{3} T^{2}}$};
\node (i2) at (5,1.5) {};
\draw (v2.north) -- (i2) ;
 \node (iv2) at (5,1.7) {$x_{3}$};
 \draw [thick] (v2.west) -- ++ (-0.5,0);
\node[rectangle,minimum width=4em,draw] (v1) at (7.5,0) {$C_{T^{2} T^{1}}$};
\node (i1) at (7.5,1.5) {};
\draw (v1.north) -- (i1) ;
 \node (iv1) at (7.5,1.7) {$x_{2}$};
  \draw [thick] (v2.east) -- (v1.west);
  \draw  (v1.east) -- ++ (0.8,0);
\node (iv0) at (9.4,0) {$x_{1}$};
\node (cdots) at (2.8,0) {$\cdots$};
\end{scope}
\end{tikzpicture}
\caption{Tensor network representation of $\UschinvPQ \ket{(T,M)}$ from \cref{eq:USch3} as a matrix product state.
Tensors $C_{T^{k} T^{k-1}}$ have three indices. The matrices $C_{T^{k} T^{k-1}}^{x_{k}}$ from \cref{eq:cg_matrix_entires} are obtained from $C_{T^{k} T^{k-1}}$ by fixing the index corresponding to $x_k$. Indices $x_k \in [d]$ indicate the computational basis states $\ket{x_k}$. The bond dimensions $D_k = |\GT(T^k,d)|$ are equal to the number of Gelfand--Tsetlin patterns of a shape $T^k$. Asymptotically the maximal value of $D_k$ for different $k$ is upper bounded by $(p+q)^{O(d^2)}$.}
\label{fig:MPS}
\end{figure}

%%%%%%%%%%%%%%%%%%%%%%%%%%%%%%%%%%%%%%%%%%%%%%%%%%%%%%%%%%%%%%%%%%%%%%%%%%%%%%%%%%%%%%%%%%%%%%%%%%%%%%%%%%%%%%%%%%%%%%%%%%%%%%%%%%%%%%%%%%%%%%%%%%%%%%%%
\subsection{Mixed Schur transform achieves the Gelfand--Tsetlin basis}\label{sec:schur_gt_basis}

By construction \cite{Vilenkin1992}, the Clebsch--Gordan transform \cref{def:CG_classical} achieves the Gelfand--Tsetlin basis of the unitary group on the unitary group registers $V_\lambda^{\U{d}}$ in \cref{eq:Sch}.
As a consequence, the same holds true also for the mixed Schur transform \eqref{def:Usch_classical}.
However, it is also true that the mixed Schur transform yields the Gelfand--Tsetlin basis of the algebra $\A_{p,q}^d$ in the relevant register $V_\lambda^{\A_{p,q}^d}$.

Because the Clebsch--Gordan transform $\CG^{k}$ from \cref{def:CGk_map,def:CG_classical} acts only on the unitary irrep registers upon adding a qudit $\C^d$, we get for every $\sigma \in \A^d_{k-1}$
\begin{align}
    \CG^{k} \of[\bigg]{ \of[\bigg]{ \bigoplus_{ \mu \in \Irr{\A^d_{k-1}} } \psi_{\mu}(\sigma) \otimes I^{\U{d}}_\mu } \otimes I} {\CG^{k}}\ct
    = \bigoplus_{\lambda \in \Irr{\A^d_{k}}} \of[\bigg]{ \bigoplus_{\substack{\mu \in \Irr{\A^d_{k-1}} \\ \mu \,:\, \mu \rightarrow \lambda}} \psi_{\mu}(\sigma) } \otimes I_\lambda^{\U{d}}.
\end{align}
We see that the $\CG^{k}$ does not change the action inside the registers $V_\lambda^{\A_{k-1}^d}$, meaning that it is naturally implementing a subalgebra-adapted basis, namely a Gelfand--Tsetlin basis. However, there is still a degree of freedom of choosing phases for the Gelfand--Tsetlin basis vectors. In \cite{jordan2009permutational,HarrowThesis} it was argued, that our choice of the Clebsch--Gordan transforms $\CG^{k}$ implements exactly the same Gelfand--Tsetlin basis in $\A_{p,q}^d$ register of the mixed Schur--Weyl duality as in \cref{thm:main} for permutations $\sigma_i$, $i \neq p$. For the contraction generator $\sigma_p \in \A_{p,q}^d$ we numerically observe the same correspondence. This, in principle, can be proved by directly contracting two tensor networks from \cref{fig:MPS} corresponding to two different paths $S,T \in \Paths(\lambda)$ for every $\lambda \in \Irr{\A_{p,q}^d}$.

The knowledge of explicit action of the generators of $\A_{p,q}^d$ in the Gelfand--Tsetlin basis is useful for quantum computing applications. For example, it yields the efficient quantum circuit for port-based teleportation, see \cref{sec:PBT}. 
\Cref{def:Usch_classical} suggests not only a way of classically computing the matrix entries of the mixed quantum Schur transform unitary but also a quantum circuit for implementing the corresponding isometry (we will use the same notation for both). In the next section, we describe a quantum circuit which implements the needed basis transformation.

%%%%%%%%%%%%%%%%%%%%%%%%%%%%%%%%%%%%%%%%%%%%%%%%%%%%%%%%%%%%%%%%%%%%%%%%%%%%%%%%%%%%%%%%%%%%%%%%%%%%%%%%%%%%%%%%%%%%%%%%%%%%%%%%%%%%%%%%%%%%%%%%%%%%%%%%
\subsection{Quantum circuit for mixed Schur transform}\label{sec:SchTransQuantum}

 In this section, we describe a quantum circuit that implements the mixed quantum Schur transform isometry $\UschPQ$:
\begin{equation}\label{eq:quantum_schur_domain_range}
    \UschPQ \colon
    (\C^d)^{p+q} \to
    \rlap{$\overbrace{\phantom{\C^{\abs{\Irr{\A_2^d}}} \x \dotsb \x
    \C^{\abs{\Irr{\A_{p+q-1}^d}}} \x
    \C^{\abs{\Irr{\A_{p+q}^d}}}}}^T$}
    \C^{\abs{\Irr{\A_2^d}}} \x \dotsb \x
    \C^{\abs{\Irr{\A_{p+q-1}^d}}} \x
    \underbrace{\C^{\abs{\Irr{\A_{p+q}^d}}} \x
    \C^{\abs{\IrrU{d-1}^{p+q}}} \x \dotsb \x
    \C^{\abs{\IrrU{1}^{p+q}}}}_M,
\end{equation}
where the first $p+q-1$ registers of the output correspond to a path $T \in \Paths(\lambda)$ labelling corresponding irreps of $\A_k^d$, the last $d$ registers correspond to a Gelfand--Tsetlin pattern $M \in \GT(\lambda,d)$ for some $\lambda \in \Irr{\A_{p,q}^d}$ and $\IrrU{k}^{p+q}$ is a set of staircases of bounded size, labelling the irreps of the group $\U{k}$. The last vertex $T^{p+q}$ of the path $T$ coincides with the top row $\m_d$ of the Gelfand--Tsetlin pattern $M$, i.e., $T^{p+q}=\m_d=\lambda$; this explains the overlap between $T$ and $M$ in \cref{eq:quantum_schur_domain_range}.

More precisely, for every integer $n \in [p+q]$ and $k \in [d]$ we define the set $\IrrU{k}^n$ of staircases labelling some irreps of $\U{k}$ as
\begin{equation}
    \IrrU{k}^n \defeq 
    \begin{cases}
        \set[\big]{ \lambda \in \IrrU{k} \, \big| \, \abs{\lambda} \leq n} & n<p+q \text{ or } k<d, \\
        \Irr{\A_{p+q}^d} & n=p+q \text{ and } k=d.
    \end{cases}
\end{equation}
The output is \cref{eq:quantum_schur_domain_range} of the quantum mixed Schur transform has tensor product structure for both path $T \in \Paths(\lambda)$ and the Gelfand--Tsetlin pattern $M \in \GT(\lambda,d)$. The natural way of interpreting the rows of a general Gelfand--Tsetlin pattern $M = (\m_d,\m_{d-1},\m_{d-2},\dotsc,\m_{1}) \in \GT(\m_d,d)$ is to use the staircase notation for each row. This means that one should think of the quantum states $\ket{T}$ and $\ket{M}$ as follows:
\begin{align}
    \ket{T} &= \ket{T^2} \otimes \ket{T^3} \otimes \dotsb \otimes \ket{T^{p+q-1}} \otimes \ket{T^{p+q}}, \\
    \ket{M} &= \ket{\m_d} \otimes \ket{\m_{d-1}} \otimes \dotsb \otimes \ket{\m_2} \otimes \ket{\m_1}.
\end{align}
Note that we suppress the registers corresponding to $\ket{T^0}$ and $\ket{T^1}$ of the path $T \in \Paths(\lambda)$ since they are always one-dimensional. Moreover, it is useful to define a shorthand notation $M_{(k)}$ to indicate only the bottom $k$ rows of the Gelfand--Tsetlin pattern $M$ and the corresponding quantum state for any $k \in [d]$ as
\begin{equation}
    \ket{M_{(k)}} \defeq \ket{\m_{k}} \otimes \ket{\m_{k-1}} \otimes \dotsb \otimes \ket{\m_{1}}.
\end{equation}

Our construction is a slight modification of the original quantum Schur transform \cite{bch2006quantumschur,HarrowThesis} for the classical Schur--Weyl duality (i.e., the $q=0$ case of our formalism).
It involves a cascade of Clebsch--Gordan isometries $\CGqc{+}{d}$ (see \cref{fig:CG}), where each isometry implements one of the unitary Clebsch--Gordan transforms $\CG^{k}$ from \cref{def:CGk_map}.
The only modification we make to the original construction is to replace $\CGqc{+}{d}$ by a similarly defined \emph{dual Clebsch--Gordan isometry} $\CGqc{-}{d}$ in the second half of the circuit (see \cref{fig:mixed_sch_circuit}). This immediately leads to the following result which agrees with \cite{bch2006quantumschur,HarrowThesis} in the $q = 0$ case.

\begin{theorem}[Mixed quantum Schur transform]\label{thm:mixed_schur}
The mixed quantum Schur transform isometry (see \cref{fig:mixed_sch_circuit}) for block-diagonalizing the algebra $\A_{p,q}^d$ has a quantum circuit with $\poly(d,p+q,\log1/\epsilon)$ gates, where $\epsilon$ is the desired error, $d$ is the local dimension, and $p$ and $q$ are the parameters of $\A_{p,q}^d$.
\end{theorem}

The main new building block of the mixed quantum Schur transform circuit is the dual Clebsch--Gordan isometry $\CGqc{-}{d}$.
It can be obtained by a small modification of the usual $\CGqc{+}{d}$ isometry, which can be taken directly from \cite{bch2006quantumschur,HarrowThesis}.
While both isometries have the same structure, their representation-theoretic interpretation is slightly different:
$\CGqc{+}{d}$ is used when $k \leq p$ to add a new box to $T^{k-1}_l$ producing a new diagram $T^{k}_l$, while $\CGqc{-}{d}$ is used when $k > p$ either to remove a box from $T^{k-1}_l$ or to add a box to $T^{k-1}_r$.
In the staircase notation from \cref{sec:mixed diagrams}, this is equivalent to adding a box to the staircase $T^{k-1}$ when $k \leq p$ and removing a box when $k > p$.
We describe efficient quantum circuits for both isometries in a unified way in \cref{sec:CG circuit}.
Their complexity is $\poly(d,\log(p+q),\log1/\epsilon)$, which leads to the $(p+q)\poly(d,\log(p+q),\log1/\epsilon)$ complexity for the mixed quantum Schur transform $\UschPQ$.

\begin{figure}[!ht]
\centering
\begin{quantikz}[wire types = {q,b,q}, classical gap = 2pt]
\lstick{$\ket{T^{k-1}}$}   & \gate[wires=3]{\CGqc{\pm}{d}} & \rstick{$\ket{T^{k-1}}$} \\
\lstick{$\ket{M_{(d-1)}}$} &                               & \setwiretype{q} \rstick{$\ket{T^{k}}$} \\
\lstick{$\ket{x_{k}}$}     &                               & \setwiretype{b} \rstick{$\ket{N_{(d-1)}}$}
\end{quantikz}
\caption{\label{fig:CG}Input and output registers of the Clebsch--Gordan isometry $\CGqc{\pm}{d}$ which appears in \cref{fig:mixed_sch_circuit}.
Here $T^{k-1}$ and $T^{k}$ denote the incoming and outgoing unitary group irreps (they are denoted by $\lambda$ and $\mu$ in \cref{def:CGk_map}).
Moreover, $T^{k-1}$ is the first row of a Gelfand--Tsetlin pattern $M$ of length $d$ whose remaining rows $M_{(d-1)}$ are encoded by the tensor product state
$\ket{M_{(d-1)}} \defeq \ket{\m_{d-1}} \cdots \ket{\m_{1}}$.
Similarly, $T^{k}$ is the first row of a Gelfand--Tsetlin pattern $N$ of length $d$ whose remaining rows $N_{(d-1)}$ are encoded by
$\ket{N_{(d-1)}} \defeq \ket{\n_{d-1}} \cdots \ket{\n_{1}}$.}
\end{figure}
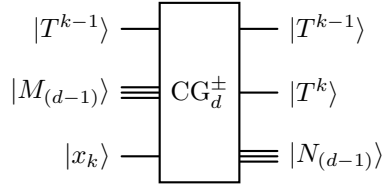

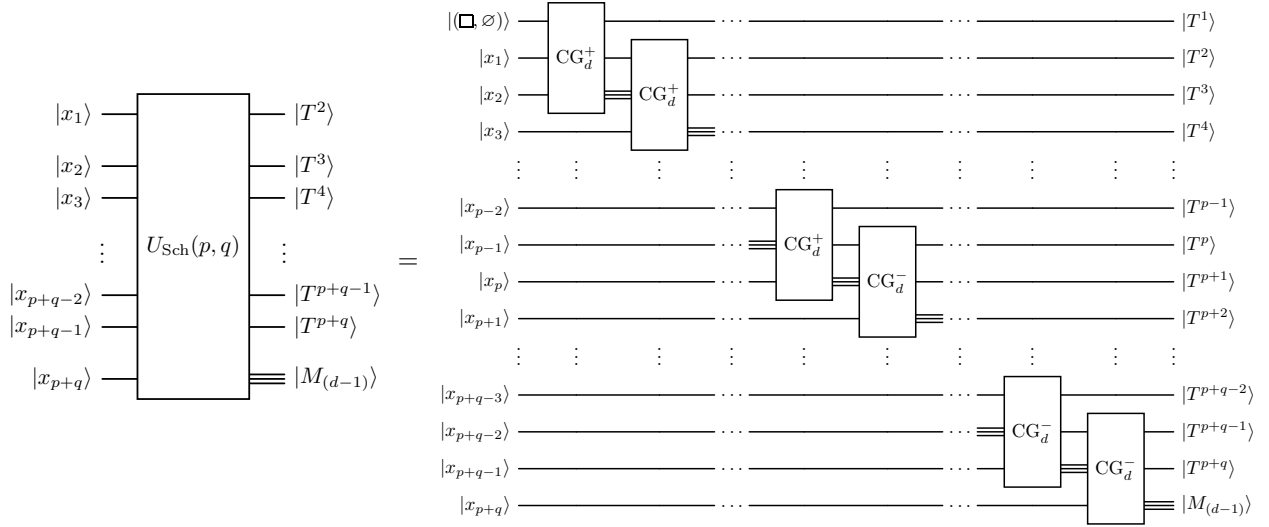
\begin{figure}[!ht]
\centering
\newcommand{\UC}[1]{\gate[wires=3]{\CGqc{#1}{d}}} % Clebsch-Gordan transform gate
\resizebox{0.3\textwidth}{!}{\begin{quantikz}[classical gap = 2pt]
\lstick{$\ket{x_1}$}       & \gate[wires=7]{\UschPQ} & \rstick{$\ket{T^{2}}$} \\
\lstick{$\ket{x_2}$}       &                         & \rstick{$\ket{T^{3}}$} \\
\lstick{$\ket{x_3}$}       &                         & \rstick{$\ket{T^{4}}$} \\
\setwiretype{n} \vdots     &                         & \vdots                       \\
\lstick{$\ket{x_{p+q-2}}$} &                         & \rstick{$\ket{T^{p+q-1}}$} \\
\lstick{$\ket{x_{p+q-1}}$} &                         & \rstick{$\ket{T^{p+q  }}$} \\
\lstick{$\ket{x_{p+q  }}$} &                         & \qb \rstick{$\ket{M_{(d-1)}}$} \\
\end{quantikz}}
%\quad
$=$
%\quad
\resizebox{0.65\textwidth}{!}{
\begin{quantikz}[row sep = 0pt, classical gap = 2pt]
\lstick{$\ket{\of*{\ytableausetup{boxsize=6pt}\ydiagram{1},\0}}$} & \UC+   &         &    \HDots &         &         &    \HDots &         &         &        \rstick{$\ket{T^{1}}$} \\
\lstick{$\ket{x_1}$}       &        & \UC+    &    \HDots &         &         &    \HDots &         &         &        \rstick{$\ket{T^{2}}$} \\
\lstick{$\ket{x_2}$}       &        &     \qb & \qq\HDots &         &         &    \HDots &         &         &        \rstick{$\ket{T^{3}}$} \\
\lstick{$\ket{x_3}$}       &        &         & \qb\HDots & \qq     &         &    \HDots &         &         &        \rstick{$\ket{T^{4}}$} \\
\setwiretype{n} \vdots     & \vdots & \vdots  &    \vdots & \vdots  & \vdots  &    \vdots & \vdots  & \vdots  & \vdots \\[6pt]
\lstick{$\ket{x_{p-2}}$}   &        &         &    \HDots & \UC+    &         &    \HDots &         &         &        \rstick{$\ket{T^{p-1}}$} \\
\lstick{$\ket{x_{p-1}}$}   &        &         &    \HDots &     \qb & \UC-\qq &    \HDots &         &         &        \rstick{$\ket{T^{p  }}$} \\
\lstick{$\ket{x_{p  }}$}   &        &         &    \HDots &         &     \qb & \qq\HDots &         &         &        \rstick{$\ket{T^{p+1}}$} \\
\lstick{$\ket{x_{p+1}}$}   &        &         &    \HDots &         &         & \qb\HDots & \qq     &         &        \rstick{$\ket{T^{p+2}}$} \\
\setwiretype{n} \vdots     & \vdots & \vdots  &    \vdots & \vdots  & \vdots  &    \vdots & \vdots  & \vdots  & \vdots \\[6pt]
\lstick{$\ket{x_{p+q-3}}$} &        &         &    \HDots &         &         &    \HDots & \UC-    &         &        \rstick{$\ket{T^{p+q-2}}$} \\
\lstick{$\ket{x_{p+q-2}}$} &        &         &    \HDots &         &         &    \HDots &     \qb & \UC-\qq &        \rstick{$\ket{T^{p+q-1}}$} \\
\lstick{$\ket{x_{p+q-1}}$} &        &         &    \HDots &         &         &    \HDots &         &     \qb & \qq    \rstick{$\ket{T^{p+q  }}$} \\
\lstick{$\ket{x_{p+q  }}$} &        &         &    \HDots &         &         &    \HDots &         &         & \qb    \rstick{$\ket{M_{(d-1)}}$} \\
\end{quantikz}
}
\caption{Schematic depiction of the mixed quantum Schur transform $\UschPQ$ and its implementation by a cascade of Clebsch--Gordan transforms $\CGqc{\pm}{d}$.
Note that we switch from $\CGqc{+}{d}$ to $\CGqc{-}{d}$ starting at input $\ket{x_{p+1}}$.
Since the topmost register on the right is one-dimensional, i.e., its input and output values are fixed to $\ket{T^{1}} = \ket{(\yd{1},\0)}$, we suppress it in $\UschPQ$.}
\label{fig:mixed_sch_circuit}
\end{figure}

%%%%%%%%%%%%%%%%%%%%%%%%%%%%%%%%%%%%%%%%%%%%%%%%%%%%%%%%%%%%%%%%%%%%%%%%%%%%%%%%%%%%%%%%%%%%%%%%%%%%%%%%%%%%%%%%%%%%%%%%%%%%%%%%%%%%%%%%%%%%%

\subsection{Quantum circuit for Clebsch--Gordan transforms}\label{sec:CG circuit}

In this section, we describe a recursive quantum circuit of complexity $\poly(d,\log(p+q),\log1/\epsilon)$ for the (dual) Clebsch--Gordan isometry $\CGqc{\pm}{d}$.
The only difference between $\CGqc{-}{d}$ and $\CGqc{+}{d}$ is the formulas for their matrix entries and the labelling scheme of their input and output basis vectors.
Our construction is based on \cite{bch2006quantumschur,HarrowThesis} and can be seen as a consequence of the fact that both the usual and the dual Clebsch--Gordan coefficients can be expressed as products of \emph{reduced Wigner coefficients}\footnote{In \cite{bch2006quantumschur,HarrowThesis} this fact is proved based on the Wigner--Eckart theorem. In our approach, the starting point is the fact that Clebsch--Gordan coefficients can be expressed as products of reduced Wigner coefficients, see \cite{Vilenkin1992}. Consequently, we try to keep the notation from \cite{Vilenkin1992}, see \cref{sec:CGcoefficients}.}, see \cref{sec:CGcoefficients}.

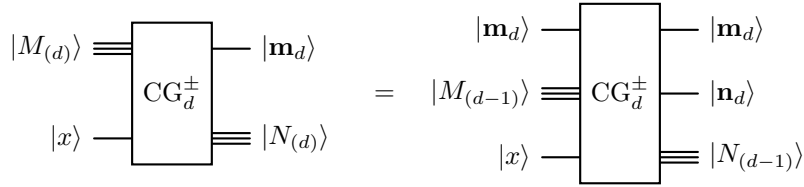
\begin{figure}[!ht]
\centering
%%%%%%%%%%%%%%%%%%%%%%%%%%%%%%%%%%%%%%
\begin{quantikz}[wire types = {b,q}, classical gap = 2pt]
\lstick{$\ket{M_{(d)}}$} & \gate[wires=2]{\CGqc{\pm}{d}} & \setwiretype{q} \rstick{$\ket{\m_d}$} \\
\lstick{$\ket{x}$}       &                               & \setwiretype{b} \rstick{$\ket{N_{(d)}}$}
\end{quantikz}
%%%%%%%%%%%%%%%%%%%%%%%%%%%%%%%%%%%%%%
\quad$ = $\quad
%%%%%%%%%%%%%%%%%%%%%%%%%%%%%%%%%%%%%%
\begin{quantikz}[wire types = {q,b,q}, classical gap = 2pt]
\lstick{$\ket{\m_d}$}      & \gate[wires=3]{\CGqc{\pm}{d}} & \rstick{$\ket{\m_d}$} \\
\lstick{$\ket{M_{(d-1)}}$} &                               & \setwiretype{q} \rstick{$\ket{\n_d}$} \\
\lstick{$\ket{x}$}         &                               & \setwiretype{b} \rstick{$\ket{N_{(d-1)}}$}
\end{quantikz}
%%%%%%%%%%%%%%%%%%%%%%%%%%%%%%%%%%%%%%
\caption{Input and output registers of the quantum gate $\CGqc{\pm}{d}$ that implements the $\CG^{k}$ transform from \cref{def:CGk_map}.
Here $\m_d,\n_d$ denote the top rows of the Gelfand--Tsetlin patterns $M$ and $N$, indicating irrep labels. We denote the last $d-1$ rows of $M$ by $M_{(d-1)}$ and encode it as a tensor product state: $\ket{M_{(d-1)}} = \ket{\m_{d-1}} \cdots \ket{\m_{1}}$.}
\end{figure}

Let us fix an arbitrary level $k = 2, \dotsc, p+q$ in the Bratteli diagram of $\A_{p,q}^d$, which is the same as fixing a position in the cascade of Clebsch--Gordan transforms in \cref{fig:mixed_sch_circuit}. Then for any $x \in [d]$ and $\n_d, \m_d \in \IrrU{d}^n$,
we define
\begin{equation}
    C^{x,\pm}_{\n_{d},\m_{d}} \ket{M_{(d-1)}} \defeq \sum_{\substack{ N \in \GT(\n_{d},d)}} c^{x,\pm}_{N,M} \ket{N_{(d-1)}}
\end{equation}
where $c^{x,\pm}_{N,M}$ are defined in \cref{sec:CGcoefficients} and $\pm$ refers to either dual or direct Clebsch--Gordan coefficients. The operators $C^{x,\pm}_{\n_{d},\m_{d}}$ defined above are essentially the same as classical matrices $C^{x_{j}}_{T^{j}, T^{j-1}}$ defined in \cref{eq:cg_matrix_entires}. Now we can define quantum Clebsch--Gordan transforms $\CGqc{\pm}{d}$, which are quantum analogues of $\CG^{k}$ from \cref{def:CGk_map}, as
\begin{equation}
    \CGqc{\pm}{d} \ket{M_{(d)}} \ket{x} \defeq \ket{\m_{d}} \sum_{\n_{d} \,:\, \m_{d} \rightarrow_{\pm} \n_{d} } \ket{\n_d} \of[\big]{C^{x,\pm}_{\n_{d},\m_{d}}\ket{M_{(d-1)}}},
\end{equation}
where the notation $\n_{d} \,:\, \m_{d} \rightarrow_{\pm} \n_{d} $ means that we obtain a staircase $\n_{d}$ by either adding or removing a box from the staircase $\m_d$, depending on the chosen $\CGqc{\pm}{d}$ Clebsch--Gordan transform.
More explicitly, we can write
\begin{equation}\label{eq:cg_qc_def}
    \CGqc{\pm}{d} = \sum_{x \in [d]} \sum_{\substack{\m_{d},\n_{d} \in \IrrU{d}^{p+q} }} \sum_{\substack{ M \in \GT(\m_{d},d) \\ N \in \GT(\n_{d},d) }} c^{x,\pm}_{N_{(d)}, M_{(d)}} \ketbra{\m_{d},\n_{d},N_{(d-1)}}{\m_{d},M_{(d-1)},x}.
\end{equation}
Our goal is to recursively implement $\CGqc{\pm}{d}$ in terms of $\CGqc{\pm}{d-1}$.
Using \cref{eq:cg_qc_def} we can easily write the Clebsch--Gordan transform for $d-1$:
\begin{equation}
    \CGqc{\pm}{d-1} = \sum_{x \in [d-1]}  \sum_{\substack{\m_{d-1},\n_{d-1} \in \IrrU{d-1}^{p+q}}} \sum_{\substack{\\ M \in \GT(\m_{d-1},d-1) \\ N \in \GT(\n_{d-1},d-1) }} c^{x,\pm}_{N_{(d-1)}, M_{(d-1)}} \ketbra{\m_{d-1},\n_{d-1},N_{(d-2)}}{\m_{d-1},M_{(d-2)},x}.
\end{equation}
Recall from \cref{sec:CGcoefficients}
that for every $k \in [d]$ and Gelfand--Tsetlin patterns of length $k-1$,
\begin{equation}
    c^{k,\pm}_{N_{(k-1)},M_{(k-1)}} \defeq \delta_{N_{(k-1)},M_{(k-1)}}.
\end{equation}
Moreover, for any
symbol
$x \in [k]$,
staircases
$\m_k, \n_k \in \IrrU{k}^{p+q}$,
and Gelfand--Tsetlin patterns
$M \in \GT(\m_{k},k)$,
$N \in \GT(\n_{k},k)$
the following recursive identity holds:
\begin{equation}
    c^{x,\pm}_{N_{(k)},M_{(k)}} =
        \rwpm{\m_k}{\n_k}{\m_{k-1}}{\n_{k-1}} \cdot c^{x,\pm}_{N_{(k-1)},M_{(k-1)}}.
\end{equation}
Moreover, note that
\begin{align}
    & \ketbra{\m_{d},\n_{d},N_{(d-1)}}{\m_{d},M_{(d-1)},x} \nonumber \\
    &= \ketbra{\m_{d},\n_{d},\n_{d-1},N_{(d-2)}}{\m_{d},\m_{d-1},M_{(d-2)},x} \\
    &= \of[\big]{\ketbra{\m_d,\n_d,\n_{d-1}}{\m_d,\m_{d-1},\n_{d-1}} \otimes I} \cdot
    \of[\big]{I \otimes \ketbra{\m_{d-1},\n_{d-1},N_{(d-2)}}{\m_{d-1},M_{(d-2)},x}},
\end{align}
which implies
\begin{align}
     & c^{x,\pm}_{N_{(d)},M_{(d)}} \ketbra{\m_{d},\n_{d},N_{(d-1)}}{\m_{d},M_{(d-1)},x} \nonumber \\
     &= \of[\big]{\rwpm{\m_d}{\n_d}{\m_{d-1}}{\n_{d-1}} \ketbra{\m_d,\n_d,\n_{d-1}}{\m_d,\m_{d-1},\n_{d-1}} \otimes I} \cdot
     \nonumber \\
     & \hspace{+3cm} \cdot \of[\big]{c^{x,\pm}_{N_{(d-1)},M_{(d-1)}} I \otimes \ketbra{\m_{d-1},\n_{d-1},N_{(d-2)}}{\m_{d-1},M_{(d-2)},x} }.
\end{align}
Together, these observations allow us rewrite \cref{eq:cg_qc_def} as
\begin{align}
    \CGqc{\pm}{d} &=  \sum_{x \in [d]}  \sum_{\substack{\m_{d},\n_{d} \in \IrrU{d}^{p+q} }} \sum_{\substack{ M \in \GT(\m_{d},d) \\ N \in \GT(\n_{d},d) }} \ketbra{\m_{d},\n_{d},N_{(d-1)}}{\m_{d},M_{(d-1)},x} \\
    &=\sum_{x \in [d]} \sum_{\substack{\m_{d},\n_{d} \in \IrrU{d}^{p+q} \\ \m_{d-1},\n_{d-1} \in \IrrU{d-1}^{p+q} }} \sum_{\substack{ M \in \GT(\m_{d-1},d-1) \\ N \in \GT(\n_{d-1},d-1) }} \of[\big]{\rwpm{\m_d}{\n_d}{\m_{d-1}}{\n_{d-1}} \ketbra{\m_d,\n_d,\n_{d-1}}{\m_d,\m_{d-1},\n_{d-1}} \otimes I} \cdot \\
    & \hspace{+6cm} \cdot \of[\big]{c^{x,\pm}_{N_{(d-1)},M_{(d-1)}} I \otimes \ketbra{\m_{d-1},\n_{d-1},N_{(d-2)}}{\m_{d-1},M_{(d-2)},x} } \nonumber \\
    &= \of[\Bigg]{  \sum_{\substack{\m_{d},\n_{d} \in \IrrU{d}^{p+q} \\ \m_{d-1},\n_{d-1} \in \IrrU{d-1}^{p+q} }} \rwpm{\m_d}{\n_d}{\m_{d-1}}{\n_{d-1}} \ketbra{\m_d,\n_d,\n_{d-1}}{\m_d,\m_{d-1},\n_{d-1}} \otimes I} \cdot \\
    & \hspace{+10pt} \cdot \of[\Bigg]{I \otimes  \sum_{x \in [d]} \sum_{\substack{\m_{d-1},\n_{d-1} \in \IrrU{d-1}^{p+q}}} \sum_{\substack{\\ M \in \GT(\m_{d-1},d-1) \\ N \in \GT(\n_{d-1},d-1) }} c^{x,\pm}_{N_{(d-1)},M_{(d-1)}} \ketbra{\m_{d-1},\n_{d-1},N_{(d-2)}}{\m_{d-1},M_{(d-2)},x} } \nonumber \\
    &= \of[\big]{C^{\pm}_d \otimes I} \cdot \of[\big]{I \otimes \CGqcTilde{\pm}{d-1}}, \label{eq:main_obs_recursion}
\end{align}
where we introduced the following two operators
\begin{align}
    \label{def:dualCG_C_op}
    C^{\pm}_d &\defeq \sum_{\substack{\m_{d},\n_{d} \in \IrrU{d}^{p+q} \\ \m_{d-1},\n_{d-1} \in \IrrU{d-1}^{p+q} }} \rwpm{\m_d}{\n_d}{\m_{d-1}}{\n_{d-1}} \ketbra{\m_d,\n_d,\n_{d-1}}{\m_d,\m_{d-1},\n_{d-1}}, \\
    \CGqcTilde{\pm}{d-1} &\defeq \sum_{x \in [d]} \sum_{\substack{\m_{d-1},\n_{d-1} \in \IrrU{d-1}^{p+q}}} \sum_{\substack{\\ M \in \GT(\m_{d-1},d-1) \\ N \in \GT(\n_{d-1},d-1) }} c^{x,\pm}_{N_{(d-1)},M_{(d-1)}} \ketbra{\m_{d-1},\n_{d-1},N_{(d-2)}}{\m_{d-1},M_{(d-2)},x} \nonumber \\
             &= \CGqc{\pm}{d-1} + \sum_{\m_{d-1} \in \IrrU{d-1}^{p+q}} \sum_{\substack{M \in \GT(\m_{d-1},d-1)}} \ketbra{\m_{d-1},\m_{d-1},M_{(d-2)}}{\m_{d-1},M_{(d-2)},d},
\end{align}
where the last term corresponds to $x = d$.
We can translate \cref{eq:main_obs_recursion} into a quantum circuit shown in \cref{fig:dCG_recurse}. This procedure can be continued recursively on the parameter $d$.

\begin{figure}[!ht]
\centering
\begin{quantikz}[row sep = 20pt, classical gap = 2pt]
\lstick{$\ket{\m_{d}}$}    & \gateCGpm{d}    &     \rstick{$\ket{\m_{d}}$} \\
\lstick{$\ket{M_{(d-1)}}$} &             \qb & \qq \rstick{$\ket{\n_{d}}$}  \\
\lstick{$\ket{x}$}         &                 & \qb \rstick{$\ket{N_{(d-1)}}$}
\end{quantikz}
\qquad$=$\qquad
\begin{quantikz}[row sep = 0pt, classical gap = 2pt]
\lstick{$\ket{\m_{d}}$}        &                        & \gateCpm{d} & \rstick{$\ket{\m_{d}}$} \\
\lstick{$\ket{\m_{d-1}}$}      & \gateCGtilde{\pm}{d-1} &             & \rstick{$\ket{\n_{d}}$} \\
\lstick{$\ket{M_{(d-2)}}$} \qb &                        & \qq         & \rstick{$\ket{\n_{d-1}}$} \\
\lstick{$\ket{x}$}             &                        & \qb         & \rstick{$\ket{N_{(d-2)}}$}
\end{quantikz}
\caption{Recursive implementation of $\CGqc{\pm}{d}$. The registers $\ket{M_{(d-1)}}$ and $\ket{N_{(d-1)}}$ on the left-hand side should be understood as tensor products $\ket{M_{(d-1)}} = \ket{\m_{d-1}}\ket{M_{(d-2)}}$ and $\ket{N_{(d-1)}} = \ket{\n_{d-1}}\ket{N_{(d-2)}}$.}
\label{fig:dCG_recurse}
\end{figure}
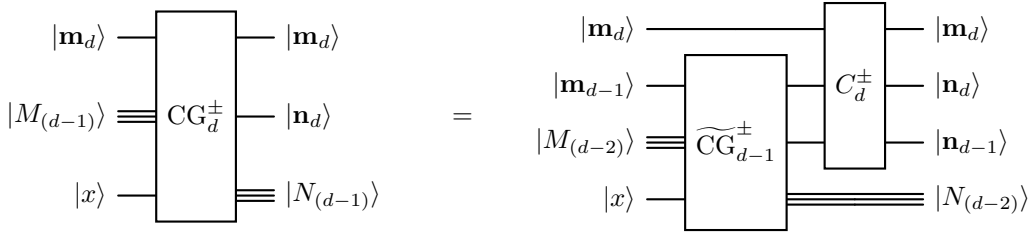

Note from \cref{def:dualCG_C_op} that $C^{\pm}_d$ is a controlled operation acting on the middle register:
\begin{equation}
    C^{\pm}_d = \sum_{\m_{d} \in \IrrU{d}^{p+q}} \sum_{\n_{d-1} \in \IrrU{d-1}^{p+q}} \proj{\m_d} \otimes C^{\pm}_{\m_{d},\n_{d-1}} \otimes \proj{\n_{d-1}},
\end{equation}
where we define $C^{\pm}_{\m_{d},\n_{d-1}}$ as
\begin{align}
    \label{def:CG_C_op_small}
    C^{\pm}_{\m_{d},\n_{d-1}} &\defeq  \sum_{\m_{d-1} \in \IrrU{d-1}^{p+q}} \sum_{\n_{d} \in \IrrU{d}^{p+q}} \rwpm{\m_d}{\n_d}{\m_{d-1}}{\n_{d-1}} \ketbra{\n_d}{\m_{d-1}}.
\end{align}
We can think of $C^{\pm}_d$ as the quantum circuit shown in \cref{fig:dCG_recurse_2}. Crucially, the $C^{\pm}_{\m_{d},\n_{d-1}}$ operator is essentially a $d \times d$ matrix since most of the coefficients $\rwpm{\m_d}{\n_d}{\m_{d-1}}{\n_{d-1}}$ are zero, see \cref{sec:CGcoefficients}. The operator $C^{\pm}_{\m_{d},\n_{d-1}}$ admits an efficient implementation as a quantum circuit because the reduced Wigner coefficients $\rwpm{\m_d}{\n_d}{\m_{d-1}}{\n_{d-1}}$ are efficiently computable, see \cref{def:reduced_wigner-0,def:reduced_wigner--,def:reduced_wigner+0,def:reduced_wigner++}.

From these formulas, we see that computing the reduced Wigner coefficients has complexity $\poly(d,\log(p+q))$. Therefore, the complexity of implementing the operator $C^{\pm}_d$ to accuracy $\epsilon$ is $\poly(d,\log(p+q),\log(1/\epsilon))$. This is the same complexity as in \cite{bch2006quantumschur,HarrowThesis}.

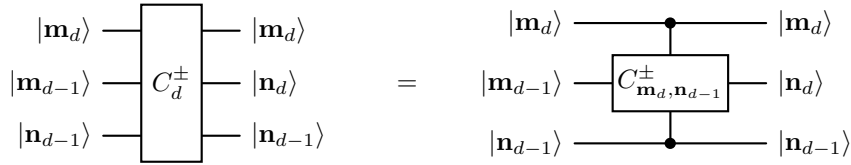
\begin{figure}[!ht]
\centering
\begin{quantikz}[row sep = 10pt, classical gap = 2pt]
\lstick{$\ket{\m_{d}}$}      \qq & \gateCpm{d} & \rstick{$\ket{\m_{d}}$} \\
\lstick{$\ket{\m_{d-1}}$}    \qq &             & \rstick{$\ket{\n_{d}}$} \\
\lstick{$\ket{\n_{d-1}}$}    \qq  &             & \rstick{$\ket{\n_{d-1}}$}
\end{quantikz}
\qquad$=$\qquad
\begin{quantikz}[row sep = 10pt, classical gap = 2pt]
\lstick{$\ket{\m_{d}}$}       \qq &    \ctrl{1}                           & \rstick{$\ket{\m_{d}}$} \\
\lstick{$\ket{\m_{d-1}}$}     \qq &    \gate{C^{\pm}_{\m_{d},\n_{d-1}}}   & \rstick{$\ket{\n_{d}}$} \\
\lstick{$\ket{\n_{d-1}}$}     \qq &    \ctrl{-1}                          & \rstick{$\ket{\n_{d-1}}$}
\end{quantikz}
\caption{Circuit for the $C^{\pm}_d$ operator from \cref{def:dualCG_C_op}.}
\label{fig:dCG_recurse_2}
\end{figure}

%%%%%%%%%%%%%%%%%%%%%%%%%%%%%%%%%%%%%%%%%%%%%%%%%%%%%%%%%%%%%%%%%%%%%%%%%%%%%%%%%%%%%%%%%%%%%%%%%%%%%%%%%%%%%%%%%%%%%%%%%%%%%%%%%%%%%%%%%%%%%%%%%%%%%%%%%%%%%%%%%%%%%%%%%%%%%%%%%%%%%%%%%%%%%%%%%%%%%%%%%%%%%%%%%%%%
%%%%%%%%%%%%%%%%%%%%%%%%%%%%%%%%%%%%%%%%%%%%%%%%%%%%%%%%%%%%%%%%%%%%%%%%%%%%%%%%%%%%%%%%%%%%%%%%%%%%%%%%%%%%%%%%%%%%%%%%%%%%%%%%%%%%%%%%%%%%%%%%%%%%%%%%%%%%%%%%%%%%%%%%%%%%%%%%%%%%%%%%%%%%%%%%%%%%%%%%%%%%%%%%%%%%

\section{Unitary-equivariant SDPs}\label{sec:SDP}

\subsection{Unitary-equivariant quantum channels}
\label{sec:Unitary-equivariant quantum channels}

A quantum channel $\Phi\colon \End(V^p) \to \End(V^q)$ is \emph{locally $\U{d}$-equivariant} or simply \emph{unitary-equivariant} if
\begin{equation}
    \Phi\of[\big]{U\xp{p} \, \rho \, U\ctxp{p}}
    = U\xp{q} \, \Phi(\rho) \, U\ctxp{q}
    \label{eq:Phi equivariance}
\end{equation}
for every $U \in \U{d}$ and quantum state $\rho$.
A convenient way of representing any quantum channel is by its \emph{Choi matrix}.
In particular, for a quantum channel $\Phi\colon \End(V^p) \to \End(V^q)$, its Choi matrix $X^\Phi \in \End(V^{p,q}_d)$ is defined as
\begin{equation}
    X^\Phi \defeq \sum_{\substack{i_1,\dotsc,i_p \in [d] \\ j_1,\dotsc,j_p \in [d]}}
    \ketbra{i_1,\dotsc,i_p}{j_1,\dotsc,j_p} \x
    \Phi\of[\big]{\ketbra{i_1,\dotsc,i_p}{j_1,\dotsc,j_p}}
\end{equation}
where the sum runs over the computational basis of $V^p_d$.
Recall that the action of $\Phi$ on state $\rho$ can be recovered from its Choi matrix $X^\Phi$ via $\Phi(\rho) = \Tr_{V^p_d} \sof[\big]{X^\Phi (\rho\tp \x I_{V^q_d})}$.
Furthermore, a given matrix $X \in \End(V^{p,q}_d)$ describes a quantum channel if and only if
% it is positive semidefinite and trace-preserving, i.e.,
$X \succeq 0$ and $\Tr_{V^q_d} (X) = I_{V^p_d}$ \cite{watrous}.

Let $X^\Phi \in \End(V^{p,q}_d)$ be the Choi matrix of a channel $\Phi$. Then $\Phi$ is unitary-equivariant if and only if
\begin{equation}
    \sof[\big]{
        X^\Phi,
        U\xp{p} \x \bar{U}\xp{q}
    } = 0,
    \label{eq:Choi equivariance}
\end{equation}
for all $U \in \U{d}$ \cite{grinko2022linear}.
Thanks to Schur's lemma, applying the mixed quantum Schur transform (\ref{eq:Sch}) decomposes $X^\Phi$ as follows:
\begin{equation}
    \label{eq:matrix_units}
    \UschPQ
    X^\Phi
    \UschinvPQ
    =
    \bigoplus_{\lambda \in \Irr{\A_{p,q}^{d}}}
    \sum_{T,S \in \Paths(\lambda)}
    x_{TS}
    \ketbra{T}{S} \otimes I_{m_{\lambda}}
\end{equation}
where $m_{\lambda}$ is the dimension of the $\lambda$-irrep of $\U{d}$, $I_{m_{\lambda}} = \sum_{M \in \GT(\lambda,d)} \proj{M}$ is the identity matrix on the unitary irrep register $\lambda$, and
\begin{equation}
    \E_{TS} \defeq \UschinvPQ \of[\bigg]{ \bigoplus_{\mu \in \Irr{\A_{p,q}^{d}}} \delta_{\lambda \mu} \ketbra{T}{S} \otimes I_{m_{\lambda}}} \UschPQ
\end{equation}
is the so-called \emph{matrix unit} algebra $\A_{p,q}^d$ corresponding to the pair of paths $T,S \in \Paths(\lambda)$.
These matrix units form a basis of $\A_{p,q}^d$ and span the whole algebra:
\begin{equation}
    \A_{p,q}^d = \spn \Set*{ \E_{TS} \given \lambda \in \Irr{\A_{p,q}^{d}},\, S,T \in \Paths(\lambda) }
\end{equation}

Notice that the trace of a matrix unit for $T,S \in \Paths(\lambda)$ is
\begin{equation}
    \Tr(\E_{TS}) = \delta_{TS} {m_{\lambda}},
\end{equation}
and their product satisfies
\begin{equation}
    \E_{ST} \cdot \E_{T'S'} = \delta_{TT'} \E_{SS'}.
\end{equation}

\subsection{Full trace}\label{sec:Trace}

In this section, we present methods to efficiently compute the trace $\Tr (YX)$ of the product of two matrices $X$ and $Y$, where $X$ is of the form (\ref{eq:matrix_units}), i.e.,
\begin{equation}
    \label{eq:Xmatrix}
    X
    = \sum_{\lambda \in \Irr{\A_{p,q}^{d}}} \sum_{S,T \in \Paths(\lambda)} x_{ST} \E_{ST},
\end{equation}
and $Y$ is assumed to be of one of three special forms: either corresponding to a matrix unit, to an entry in the matrix written in the computational basis or corresponding to a walled Brauer diagram.
Those methods will be crucial for semi-definite
optimization problems over unitary equivariant quantum channels.

%%%%%%%%%%%%%%%%%%%%%%%%%%%%%%%%%%%%%%%%%%%%%%%%%%%%%%%%%%%%%%%%%%%%%%%%%%%%%%%%%%%%%%%%%%%%%%%%%%%%%%%%%%%%%%%%%%%%%%%%%%%%%%%%%%%%%%%%%%%%%%%%%%%%%%%%%%%%%%%
\subsubsection{}
Firstly, assume that matrix $Y$ is also a matrix unit, i.e. $Y =\E_{S,T}$ for some $\E_{S, T} \in \Epq$ with $S, T \in \Paths(\lambda)$. $X = \sum_{S',T'} x_{S',T'} \E_{S',T'}$. As the matrix $X$ is of the form (\ref{eq:Xmatrix}), by elementary properties of matrix units, we have
\begin{equation}
  \Tr(Y X) =   \sum_{T',S'} x_{T',S'} \Tr(\E_{S,T} \E_{T',S'} ) = x_{T,S}  \cdot m_{\lambda}.
\end{equation}
Hence, the computational complexity of computing all coefficients in front of $x_{T,S}$ of this operation equals the complexity of the computation of the number $m_{\lambda}$, which is $O((p+q)\log^2(p+q+d))$.

%%%%%%%%%%%%%%%%%%%%%%%%%%%%%%%%%%%%%%%%%%%%%%%%%%%%%%%%%%%%%%%%%%%%%%%%%%%%%%%%%%%%%%%%%%%%%%%%%%%%%%%%%%%%%%%%%%%%%%%%%%%%%%%%%%%%%%%%%%%%%%%%%%%%%%%%%%%%%%%
\subsubsection{}
Secondly, assume that matrix $Y$ is written in the computational basis and has only a single non-zero entry, i.e. $Y=\ketbra{I}{J}$ for some $I=i_1,\dotsc,i_{p+q}$ and $J =j_1,\dotsc,j_{p+q}$. In order to compute the trace $\Tr(Y X)$, we need to re-write the entries of matrix (\ref{eq:Xmatrix}) into computational basis by applying Schur transform:
\begin{equation}
\label{eq:a1}
    \Tr(\ketbra{I}{J} X) = \bra{J}X\ket{I} = \sum_{\lambda\in \Irr{\A_{p,q}^{d}}} \sum_{S,T \in \Paths(\lambda)} x_{TS} \sum_{M \in \GT(\lambda,d)} \bra{J} \UschinvPQ \ket{(T,M)} \bra{(S,M)} \UschPQ \ket{I}
\end{equation}
Using previously derived form of Schur transform (\ref{eq:USch4}), the computation of $\bra{(T,M)} \UschPQ \ket{I} = \bra{I} \UschinvPQ \ket{(T,M)}$ is reduced to the multiplication of Clebsch--Gordan matrices:
\begin{equation}
    \bra{I} \UschinvPQ \ket{(T,M)} =
    \begin{cases}
        \bra{M} C_{T^{p+q} T^{p+q-1}}^{i_{p+q}, w(i_{p+q-1},\dotsc,i_1)}
\cdots
C_{T^{2} T^{1}}^{i_{2}} \ket{i_1} &\, w(I) = w(M), \\
        0 &\, w(I) \neq w(M), \\
    \end{cases}
\end{equation}
where $\ket{M}$ is a Gelfand--Testlin basis vector for the irreducible representation of a unitary group corresponding to $\lambda$.
Notice that $\sum_{M\in \GT(\lambda,d)} \ketbra{M}{M} = I_{m_{\lambda}}$, hence we can rewrite (\ref{eq:a1}) as
\begin{equation}
\label{eq:a3}
    \bra{J} X \ket{I} =
    \begin{cases}
         \sum_{S,T} x_{S,T} \bra{j_1} \of*{C_{S^{2} S^{1}}^{i_{2}}} \ct \cdots \of*{C_{S^{p+q} S^{p+q-1}}^{j_{p+q}, w(j_{p+q-1},\dotsc,j_1)}} \ct C_{T^{p+q} T^{p+q-1}}^{i_{p+q}, w(i_{p+q-1},\dotsc,i_1)}
\cdots
C_{T^{2} T^{1}}^{i_{2}} \ket{i_1} &\,  w(I) = w(J),\\
         0 &\, w(I) \neq w(J).
    \end{cases}
\end{equation}
As it we demonstrated in \cref{sec:MPS}, matrices in the above equation have dimensions given by Kostka numbers $K_{T^{(k)},w(i_k,\dotsc,i_1)}$ and $K_{S^{(k)},w(j_k,\dotsc,j_1)}$ respectively, and computation of complexity of computing (\ref{eq:a3}) is given by $(p+q)^{O(d^2)}$.

%%%%%%%%%%%%%%%%%%%%%%%%%%%%%%%%%%%%%%%%%%%%%%%%%%%%%%%%%%%%%%%%%%%%%%%%%%%%%%%%%%%%%%%%%%%%%%%%%%%%%%%%%%%%%%%%%%%%%%%%%%%%%%%%%%%%%%%%%%%%%%%%%%%%%%%%%%%%%%%
\subsubsection{}\label{sec:trace_comp_pi}
Lastly, assume that the matrix $Y$ is an image of a single Brauer diagram $\pi \in \B_{p,q}^d$, i.e. $Y = \psi_{p,q}^d (\pi)\in \A_{p,q}^d$. In order to compute trace $\Tr(\psi_{p,q}^d (\pi) X)$, we shall compute matrix entries of $\psi_{p,q}^d (\pi)$ in the Gelfand--Tstelin basis of all irreducible representations of $\A_{p,q}^d$. For that purpose, we shall present diagram $\pi\in \B_{p,q}^d$ in terms of generators $\sigma_i\in \B_{p,q}^d$ of the walled Brauer algebra $\B_{p,q}^d$.
Assume that diagram $\pi$ has exactly $k$ contractions.
By applying certain permutations $\sigma^u_l, \sigma^d_l \in \S_p$ and $\sigma^u_r, \sigma^d_r \in \S_q$ acting on the left and rights sides of the diagram $\pi$, it can be represented as
\begin{equation}
    \pi = (\sigma^u_l \otimes \sigma^u_r) \overline{\sigma}_k (\sigma^d_l \otimes \sigma^d_r)
\end{equation}
where $\overline{\sigma}_k \defeq \prod_{i=0}^{k-1} \overline{(p-i,p+1+i)}$ is represented by the diagram with $k$ contractions located near the wall, i.e
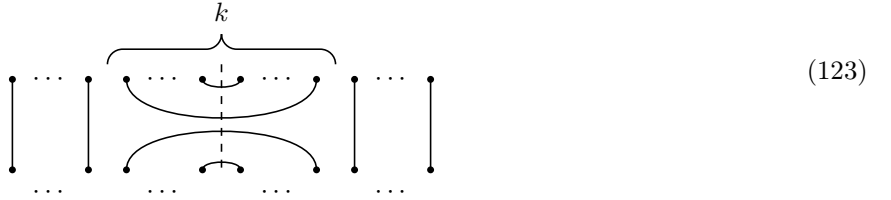
\begin{equation}
  \begin{tikzpicture}[baseline = 1cm]
    \BrauerTikZStyle
    % p and q
    \curlybrace{0.5*6.5}{1.4}{1.4} \node at (0.5*6.5,1.9) {$k$};
    % Nodes
    \foreach \i in {1,3,4,6,7,9,10,12} {
      \fill (0.5*\i,1) circle [dot] coordinate (A\i);
      \fill (0.5*\i,-0.2) circle [dot] coordinate (B\i);
    }
    % Edges
    \draw (A1) to [l] (B1);
    \draw (A3) to [l] (B3);
    \draw (A4) to [u] (A9);
    \draw (A6) to [u] (A7);
    \draw (B4) to [n] (B9);
    \draw (B6) to [n] (B7);
    \draw (A10) to [n] (B10);
    \draw (A12) to [l] (B12);
    % Wall
    \draw [dashed] (0.5*6.5,1.2) -- (0.5*6.5,-0.2);
     \node at (0.5*2,1) {$\cdots$};
     \node at (0.5*2,-0.5) {$\cdots$};
     \node at (0.5*5,1) {$\cdots$};
     \node at (0.5*5,-0.5) {$\cdots$};
     \node at (0.5*8,1) {$\cdots$};
     \node at (0.5*8,-0.5) {$\cdots$};
     \node at (0.5*11,1) {$\cdots$};
     \node at (0.5*11,-0.5) {$\cdots$};
  \end{tikzpicture}
\end{equation}
As $\sigma^u_l, \sigma^d_l \in \S_p$ they might be written as a product of at most $p^2$ generators $\sigma_i$ of the algebra $\B_{p,q}^d$.
Similarly, $\sigma^u_r, \sigma^d_r \in \S_q$ can be written as a product of $q^2$ generators $\sigma_i$ of the algebra $\B_{p,q}^d$.
Furthermore, the diagram $\overline{\sigma}_k$ can be decomposed into generators $\sigma_i$ of the algebra $\B_{p,q}^d$ as follows:
\begin{equation}
 \overline{\sigma}_k =
  (\sigma_p\sigma_{p+1}\cdots
  \sigma_{p-k+1}\cdots\sigma_{p-1}\sigma_{p+k-1})
\cdots
(\sigma_p\sigma_{p+1}\sigma_{p+2}\sigma_{p-1}\sigma_{p-2})
 (\sigma_p\sigma_{p+1}\sigma_{p-1})
 \sigma_p
\end{equation}
which is of the length given by $k^2$.
Altogether, the decomposition of an arbitrary diagram $\pi$ into generators of the algebra $\B_{p,q}^d$ requires $O((p+q)^2)$ multiplications of these generators.
We shall multiply those generators separately in the Gelfand--Tsetlin basis for each irreducible representation related to $\lambda \in \Irr{\A_{p,q}^{d}}$. Such a multiplication has a complexity $O(d_\lambda^3)$, where $d_\lambda$ is a dimension of irreducible representation corresponding to $\lambda$.
Due to the fact \cite{grinko2022linear}
\begin{equation}
    \dim{\A_{p,q}^d} = \sum_{\lambda \in \Irr{\A_{p,q}^{d}}} d_\lambda^2,
\end{equation}
complexity of computing the aforementioned product of generators is $O\of*{\of*{\dim{\A_{p,q}^d}}^{3/2}}$
\footnote{It seems, that the factor $\tfrac{3}{2}$ in this estimation could be dropped by applying Fast Fourier Transform (FFT). Indeed, FFT was successfully adapted to the setting of finite groups and some finite-dimensional semisimple algebras \cite{FFT2,maslen2018efficient} and could be easily adapted for the full walled Brauer algebra $\B_{p,q}^d$, when it is semisimple. However, for the algebra of partially transposed permutations $\A_{p,q}^d$, it is not clear to us how to present the set of vectors which span the entire algebra $\A_{p,q}^d$. The non-triviality of the ideal $\ker(\psi_{p,q}^d)$ makes the adaptation of \cite{FFT2,maslen2018efficient} highly nontrivial, so we leave it for future work.}.
Combing everything together, the total complexity of computing all matrix units of $\psi_{p,q}^d (\pi) $, i.e. computing $\psi(\pi)_{S,T} \defeq \Tr (\psi_{p,q}^d (\pi) \E_{S,T})$ is $O\of*{(p+q)^2\of*{\dim{\A_{p,q}^d}}^{3/2}}$.
Note that having the matrix units of $\psi_{p,q}^d (\pi)$ allows us to write
\begin{equation}
 \Tr(YX)=
    \Tr \of*{ \psi_{p,q}^d (\pi) X }
    = \sum_{\lambda \in \Irr{\A_{p,q}^{d}}} \sum_{S,T \in \Paths(\lambda)} \psi(\pi)_{S,T} \Tr(\E_{S,T}X)
    = \sum_{\lambda \in \Irr{\A_{p,q}^{d}}} \sum_{S,T \in \Paths(\lambda)} \psi(\pi)_{S,T} \cdot x_{T,S} \cdot m_{\lambda}.
\end{equation}

%%%%%%%%%%%%%%%%%%%%%%%%%%%%%%%%%%%%%%%%%%%%%%%%%%%%%%%%%%%%%%%%%%%%%%%%%%%%%%%%%%%%%%%%%%%%%%%%%%%%%%%%%%%%%%%%%%%%%%%%%%%%%%%%%%%%%%%%%%%%%%%%%%%%%%%%%%%%%%%
\subsection{Partial trace}\label{sec:Partial trace}

In this section, we present methods to efficiently compute the partial trace $\Tr_{S}(XY)$ of a product of two matrices where $X$ is presented as a linear combination of matrix units (\ref{eq:Xmatrix}), and $Y$ is also of one of the special form. For simplicity, we assume that traced out systems are always the last system, i.e. we compute $\Tr_{S_{k}}(X)$ for the sets $S_{k} \defeq \set{k+1,\dotsc,p+q}$ for arbitrary $k$ such that $1\leq k\leq p+q-1$. In that case, we use the following general result by by Ram and Wenzl \cite{ram1992matrix}, which we adopt to our setting of algebras $\A_{p,q}^d$
\footnote{In fact, it should be possible to adapt the same result for computing the partial trace over all types of subsystems $S\subset [p+q]$. Indeed, one can "SWAP" given two subsystems using, so-called, \textit{$6j$-symbol}. In that way, subsystems in $S$ can be effectively swapped to the last positions. We leave for future work the details of such procedure.}.

\begin{lemma}[\cite{ram1992matrix}]
\label{Lemma:Ram and Wenzl}
Consider any irreducible representation $\lambda\in \A_{p,q}^d$, two paths $S,T \in \Paths(\lambda)$ and a corresponding matrix unit $\E_{S,T}$. One can decompose paths $S,T$ with respect to the last system, i.e write $S=\Bar{S}\circ \lambda$ and $T=\Bar{T}\circ \lambda$, and $T\in \Paths_{p+q-1}(\mu,)$ and $S\in \Paths_{p+q-1}(\mu)$ where $S^{p+q-1} = \mu$ and $T^{p+q-1} = \mu'$.
The partial trace of the last system for the matrix unit $\E_{S,T}$ reads:
    \begin{equation}
        \Tr_{p+q} \E_{S,T} =
        \begin{cases}
            \frac{m_{\lambda}}{m_{\mu}} \E_{\Bar{S},\Bar{T}} &\,  \mu = \mu', \\
            0 &\, \mu \neq \mu'.
        \end{cases}
    \end{equation}
\end{lemma}

\noindent
For simplicity of the notation, in this subsection, we rewrite the sequence of inclusions (\ref{eq:A inclusions}) in the following way:
\begin{equation}
    \A_{0}^d \hookrightarrow \A_{1}^d \hookrightarrow \cdots \hookrightarrow \A_{p+q-1}^d \hookrightarrow \A_{p+q}^d ,
\end{equation}
i.e. $\A_{k}^d \defeq \A_{k,0}^d$ for $k \leq p$ and $\A_{k}^d \defeq \A_{p,k-p}^d$ for $k\geq p$.
As previously, we assume that matrix $X$ is of the form (\ref{eq:matrix_units}), and matrix $Y$ is of one of the special forms: either identity matrix, or corresponding to a matrix unit, or corresponding to a walled Brauer diagram.

%%%%%%%%%%%%%%%%%%%%%%%%%%%%%%%%%%%%%%%%%%%%%%%%%%%%%%%%%%%%%%%%%%%%%%%%%%%%%%%%%%%%%%%%%%%%%%%%%%%%%%%%%%%%%%%%%%%%%%%%%%%%%%%%%%%%%%%%%%%%%%%%%%%%%%%%%%%%%%%

\subsubsection{}
Firstly, assume that $Y = I$ is an identity matrix. Applying \cref{Lemma:Ram and Wenzl} recursively, we have
\begin{align}
    \Tr_{S_k}(X) &= \sum_{\lambda \in \Irr{\A_{p,q}^{d}}} \sum_{S,T \in \Paths(\lambda)} x_{S,T} \Tr_{k} \Tr_{k+1} \dotsc \Tr_{p+q} \E_{S,T} \\
    &= \sum_{\lambda \in \Irr{\A_{p,q}^{d}}} \sum_{\mu \in \Irr{\A_{k}^{d}}} \sum_{\substack{S,T \in \Paths(\lambda) \\ S^{k} = T^{k} = \mu \\ \forall i \geq k  \,:\, S^{i} = T^{i}}} x_{S,T} \frac{m_{\lambda}}{m_{\mu}} \E_{\Bar{S},\Bar{T}} ,
\end{align}
where $\Bar{S} \in \Paths_{k}(\mu), \, \Bar{T} \in \Paths_{k}(\mu)$ are truncations of $S,T \in \Paths(\lambda)$ to the first $k$ subsystems.

%%%%%%%%%%%%%%%%%%%%%%%%%%%%%%%%%%%%%%%%%%%%%%%%%%%%%%%%%%%%%%%%%%%%%%%%%%%%%%%%%%%%%%%%%%%%%%%%%%%%%%%%%%%%%%%%%%%%%%%%%%%%%%%%%%%%%%%%%%%%%%%%%%%%%%%%%%%%%%%

\subsubsection{}
Secondly, assume that matrix $Y = \E_{S',T'}$ is a matrix unit. By elementary properties of matrix units and by applying \cref{Lemma:Ram and Wenzl} recursively, we have
\label{sec:partial_trace_comp_EST}
\begin{align}
    \Tr_{S_k}(\E_{S',T'}X) &= \sum_{\lambda \in  \Irr{\A_{p,q}^{d}}} \sum_{T \in \Paths(\lambda)} x_{T',T} \Tr_{S_k}(\E_{S',T}) = \sum_{\lambda \in  \Irr{\A_{p,q}^{d}}} \sum_{\substack{T \in \Paths(\lambda) \\ T^{k} = \mu \\ \forall i \geq k  \,:\, T^{i} = S'^{i}}} x_{T',T}  \frac{m_{\lambda}}{m_{\mu}} \E_{\Bar{S'},\Bar{T}}
\end{align}
where $\Bar{S}' \in \Paths_{k}(\mu), \, \Bar{T} \in \Paths_{k}(\mu)$ are truncations of $S',T \in \Paths(\lambda)$ respectively to the first $k$ subsystems.

%%%%%%%%%%%%%%%%%%%%%%%%%%%%%%%%%%%%%%%%%%%%%%%%%%%%%%%%%%%%%%%%%%%%%%%%%%%%%%%%%%%%%%%%%%%%%%%%%%%%%%%%%%%%%%%%%%%%%%%%%%%%%%%%%%%%%%%%%%%%%%%%%%%%%%%%%%%%%%%

\subsubsection{}\label{sec:partial_trace_comp_pi}
Lastly, assume that the matrix $Y$ is an image of a single Brauer diagram $\pi \in \B_{p,q}^d$, i.e. $Y = \psi_{p,q}^d (\pi)\in \A_{p,q}^d$. As In \cref{sec:trace_comp_pi}, we shown how to compute all matrix units of $\psi_{p,q}^d (\pi)$, i.e. computing $\psi(\pi)_{S,T} \defeq \Tr (\psi_{p,q}^d (\pi) \E_{S,T})$ in time $O\of*{(p+q)^2\of*{\dim{\A_{p,q}^d}}^{3/2}}$. Having done it, we can use the same method as in \cref{sec:partial_trace_comp_EST}:
\begin{align}
\nonumber
 \Tr_{S_k}(YX)
   & = \Tr_{S_k}(\psi_{p,q}^d (\pi)X)
     = \sum_{\lambda \in \Irr{\A_{p,q}^{d}}} \sum_{S,T \in \Paths(\lambda)} \psi_{p,q}^d (\pi)_{S,T} \Tr_{S_k}(\E_{S,T}X)
     \\
    &
    =  \sum_{\lambda \in \Irr{\A_{p,q}^{d}}} \sum_{S,T \in \Paths(\lambda)}  \psi_{p,q}^d (\pi)_{S,T} \sum_{\lambda \in \Irr{\A_{p,q}^{d}}} \sum_{\mu \in \Irr{\A_{k}^{d}}} \sum_{\substack{T' \in \Paths(\lambda) \\ S^{k} = T'^{k} = \mu \\ \forall i \geq k  \,:\, T'^{i} = S^{i}}} x_{T,T'}  \frac{m_{\lambda}}{m_{\mu}} \E_{\Bar{S},\Bar{T}'}
\end{align}
in order to compute the partial trace.

\subsection{Unitary-equivariant SDPs}\label{sec:subSDP}

Semidefinite programming is an important subfield of optimization \cite{handbookSDP} that has numerous applications in quantum information theory \cite{SDPs,watrous}.
A particularly common class of semidefinite programs (SDPs) that occur in quantum information is one where the matrix variable $X$ has a continuous unitary symmetry group:
\begin{equation}
\label{eq:equiSDP}
    \sof[\big]{X, U\xp{p} \x \bar{U}\xp{q}} = 0,
    \quad \forall U \in \U{d}.
\end{equation}
For example, this occurs when $X$ describes a unitary-equivariant quantum channel with $p$ inputs and $q$ outputs, each of dimension $d$ (see \cref{sec:Unitary-equivariant quantum channels} for more details).
We will focus on the following general class of SDPs with this symmetry:
\begin{equation}
    \begin{aligned}
        \max_X \quad & \Tr(CX) \\
        \textrm{s.t.} \quad
        & \Tr(A_k X) \leq b_k, & \forall k &\in [m_1], \\
        & \Tr_{S_{i_k}} (D_k X) = B_k, & \forall k &\in [m_2], \\
        & \sof[\big]{X, U\xp{p} \x \bar{U}\xp{q}} = 0, & \forall U &\in \U{d}, \\
        & X \succeq 0,
    \end{aligned}
    \label{eq:input SDP}
\end{equation}
Here $m_1$ and $m_2$ denote the number of inequality and equality constraints, respectively.
The matrices $A_k, D_k, C$ are Hermitian, and we consider only a specific choice of the sets $S_{i_k}$, namely $S_{i_k} \defeq \set{i_k,\dotsc,p+q}$, furthermore, we assume that matrices $D_k$ are written as linear combinations of matrix units \eqref{eq:matrix_units}.

Our goal is to devise a symmetry reduction procedure that translates the above SDP into one where the irrelevant degrees of freedom in the matrix $X$ variable are eliminated.
Thanks to mixed Schur--Weyl duality, the unitary equivariance constrain \eqref{eq:equiSDP} on matrix $X$ implies that it can be written as a linear combination of the matrix units \eqref{eq:Xmatrix} of the matrix algebra $\A_{p,q}^d$, i.e.,
\begin{equation}
    X = \sum_{S,T} x_{ST} \E_{ST}
    = \sum_{\lambda \in \Irr{\A_{p,q}^{d}}} \sum_{S,T \in \Paths(\lambda)} x_{ST} \E_{ST}
    .
\end{equation}
By rewriting the constraints in \eqref{eq:input SDP} into the mixed Schur basis, we can translate the input problem in \cref{eq:input SDP} into the following equivalent SDP problem:
\begin{equation}
    \begin{aligned}
        \max_{x_{ST}} \quad & f_C (x_{ST}) \\
        \textrm{s.t.} \quad
        & f_{A_k} (x_{ST})\leq b_k, & \forall k &\in [m_1], \\
        & g_{D_k, B_k}^{i_k,T_{i_k},S_{i_k}} (x_{ST}) = 0, & \forall k &\in [m_2],\, \\
        & X_\lambda \succeq 0,
    \end{aligned}
    \label{eq:output SDP}
\end{equation}
where $T_{i_k},S_{i_k}\in \Paths_{i_k}(\mu)$ for some $\mu \in \A_k^d $, and $X_\lambda \defeq \sum_{S,T\in \Paths(\lambda)} x_{ST} \E_{ST}$,
and $f_C$, $f_{A_k}$, $g_{D_k, B_k}^{i_k,T_{i_k},S_{i_k}}$ are some affine functions depending on indicated matrices.
The above optimisation problem has $\dim \of*{\A_{p,q}^d} = \sum_{\lambda \in \Irr{\A_{p,q}^{d}}} d^2_\lambda$ degrees of freedom.
For comparison, the original optimization problem has $d^{2(p+q)}$ degrees of freedom.
In order to make the optimization problem trackable, we make some further assumptions on the form and sparseness of matrices $A_k, C, D_k, B_k$ in the original problem. Hence, we analyze some particular cases in which the aforementioned matrices are given in one of the following forms:
\begin{enumerate}
    \item they are arbitrary linear combinations of matrix units of $\A_{p,q}^d$,
    \item they are sparse linear combinations of computational basis matrix units,
    \item they are sparse linear combinations of diagrams that span $\B_{p,q}^d$.
\end{enumerate}

Firstly, notice that if all matrices $C,A_k,D_k,B_k$ are written as a linear combination of matrix units $\E_{ST}$, the optimization problem (\ref{eq:input SDP}) can be rewritten to the form (\ref{eq:output SDP}) trivially.

Furthermore, we can use the methods presented in \cref{sec:Trace,sec:Partial trace} to efficiently compute the (partial) traces of products of matrices of different forms.
Indeed, by summarizing results form \cref{sec:Trace,sec:Partial trace}, we obtain the following generalization of the main result of \cite{grinko2022linear} from linear to semidefinite programming.

\begin{theorem}\label{thm:SDP}
The computational complexity of rewriting the input SDP (\ref{eq:input SDP}) with $d^{2(p+q)}$ variables to the reduced SDP (\ref{eq:output SDP}) with $\dim (\A_{p,q}^d)$ variables is
\begin{itemize}
    \item $O(s)$ if $C,A_k,D_k,B_k$ are given as $s$-sparse linear combinations of matrix units $\E_{ST}$ of $\A_{p,q}^d$, or are the identity matrix,
    \item $s \cdot (p+q)^{O(d^2)}$ if $C,A_k,D_k,B_k$ are given as $s$-sparse linear combinations of computational basis matrix units, while the matrices $D_k$ are linear combinations of matrix units,
    \item $O \of*{s (p+q)^2 \of*{\dim{\A_{p,q}^d}}^{3/2}}$ if $C,A_k,D_k,B_k$ are $s$-sparse linear combinations of diagrams in $\B_{p,q}^d$.
\end{itemize}
\end{theorem}

In the first two cases, the complexity scales polynomially in the system size $p+q$ when the local dimension $d$ is constant.
In the third case, the complexity scales in $\dim(\A_{p,q}^d)$ as opposed to $\dim(\B_{p,q}^d) = (p+q)!$, which can make a significant difference for small values of the local dimensions $d$ (see \cite{grinko2022linear} for more discussion).
Note that scaling in
$\dim(\A_{p,q}^d)$ versus
$\dim(\B_{p,q}^d)$
results from writing the matrix variable as a formal linear combination of matrix units $\E_{ST}$ instead of walled Brauer algebra diagrams.

%%%%%%%%%%%%%%%%%%%%%%%%%%%%%%%%%%%%%%%%%%%%%%%%%%%%%%%%%%%%%%%%%%%%%%%%%%%%%%%%%%%%%%%%%%%%%%%%%%%%%%%%%%%%%%%%%%%%%%%%%%%%%%%%%%%%%%%%%%%%%%%%%%%%%%%%%%%%%%%%%%%%%%%%%%%%%%%%%%%%%%

\Cref{thm:SDP} can be applied to a number of important optimization problems in quantum information which are all of the form \eqref{eq:input SDP}: optimization over PPT-extendible channels \cite{HoldsworthVishalWilde}, SDP relaxation hierarchy for bilinear optimization \cite{berta2021semidefinite} for approximate quantum error correction \cite{chee2023efficient}, optimization over quantum combs \cite{quintino2019probabilistic,quintino2019reversing,yoshida2021universal,quintino2021deterministic,grinko2022linear}, unitary-equivariant Boolean functions \cite{buhrman2016quantum},
monogamy of entanglement
\cite{Monogamy},
and many others
(see \cite{grinko2022linear} for an extensive list of references).
In most of these settings, there are additional discrete symmetries as well, which can be utilized together with unitary-equivariance to simplify the problem even further.
We leave finding efficient ways of combining these symmetries with unitary equivariance for future work.

%%%%%%%%%%%%%%%%%%%%%%%%%%%%%%%%%%%%%%%%%%%%%%%%%%%%%%%%%%%%%%%%%%%%%%%%%%%%%%%%%%%%%%%%%%%%%%%%%%%%%%%%%%%%%%%%%%%%%%%%%%%%%%%%%%%%%%%%%%%%%%%%%%%%%%%%%%%%%%%%%%%%%%%%%%%%%%%%%%%%%%%
%%%%%%%%%%%%%%%%%%%%%%%%%%%%%%%%%%%%%%%%%%%%%%%%%%%%%%%%%%%%%%%%%%%%%%%%%%%%%%%%%%%%%%%%%%%%%%%%%%%%%%%%%%%%%%%%%%%%%%%%%%%%%%%%%%%%%%%%%%%%%%%%%%%%%%%%%%%%%%%%%%%%%%%%%%%%%%%%%%%%%%%

\section{Efficient quantum circuit for optimal port-based teleportation}\label{sec:PBT}

Quantum teleportation is a cornerstone of quantum information \cite{bennett1993teleporting}.
However, one drawback of the original teleportation protocol is that the receiving party needs to perform a correction operation on the received state.
\emph{Port-based teleportation} (PBT) gets around this limitation \cite{IshizakaHiroshima,ishizaka2009quantum}.
In PBT, Alice and Bob share an entangled resource state distributed evenly among $p$ quantum systems called \emph{ports} on each side.
To teleport an unknown quantum state, Alice measures it together with her share of the ports
The measurement outcome, which she communicates to Bob, indicates to which of Bob's ports the state has teleported to.
Bob does not need to perform any correction but simply discard the remaining ports.

Port-based teleportation possesses the crucial feature of unitary equivariance, meaning it remains effective when Bob applies the same unitary operation to his port systems before the protocol starts. However, due to finite resources \cite{NoProgramming}, unitarily equivariant PBT protocols can only achieve approximate teleportation. Nevertheless, certain PBT protocols become asymptotically faithful as the number of ports increases \cite{beigi2011simplified,mozrzymas2018optimal,christandl2021asymptotic}. PBT has diverse applications in non-local quantum computation and quantum communication \cite{beigi2011simplified,buhrman2016quantum,may2022complexity}, channel discrimination \cite{pirandola2019fundamental}, channel simulation \cite{pereira2021characterising}, and holography in high-energy physics \cite{may2019quantum,may2022complexity}. PBT has also been extended to multi-port teleportation \cite{studzinski2020efficient,kopszak2020multiport,mozrzymas2021optimal}. The resource requirements for PBT have been studied further in \cite{studzinski2022square,strelchuk2023minimal}.

Usually, two types of PBT protocols are considered: probabilistic exact and deterministic inexact. The optimal entanglement fidelity for deterministic protocol is related to the success probability in probabilistic version \cite{leditzky2020optimality}. Typical resource states considered for PBT are either $p$ maximally entangled pairs of states or a nontrivial optimized state (which achieves the best possible optimal entanglement fidelity) \cite{studzinski2017port,mozrzymas2018optimal}.

It turns out that the same \emph{pretty good measurement} measurement is optimal for both cases \cite{studzinski2017port,mozrzymas2018optimal,leditzky2020optimality}. We denote this positive operator-valued measure (POVM) by $E = \set{E_k}_{k=0}^{p}$. Denote Alices' ports by $A_1,\dotsc,A_p$. The input register $p+1$ on Alices' side is for the state $\ket{\psi} \in \C^d$ to be teleported to Bob. Alice measures all her registers and if she obtains outcome $k \in [p]$, then this is the number of the port where Bob should find the teleported state $\ket{\psi}$. Otherwise, upon measuring $k=0$ she aborts the protocol (for probabilistic exact PBT) or sends a random classical outcome $k \in [p]$ to Bob (deterministic inexact PBT). The optimal POVM $E$ is given by \cite{studzinski2017port,mozrzymas2018optimal,leditzky2020optimality}
\begin{equation}\label{def:PGM_PBT}
    E_k = \rho^{-1/2} \rho_k \rho^{-1/2} \text{ for every $k \in [p]$}, \qquad E_0 = I-\sum_{k=1}^{p} E_k.
\end{equation}
Here $\rho^{-1}$ should be understood as the generalized inverse of $\rho \defeq \sum_{k=1}^p \rho_k$ where
\begin{equation}
    \rho_k \defeq \psi^d_{p,q} \of*{\pi^{k} \sigma_p \pi^{-k}},
\end{equation}
where $\psi^d_{p,q}$ denotes the map defined in \cref{eq:Brauer action}, $\pi \defeq \sigma_{1}\sigma_{2}\dots\sigma_{p-2}\sigma_{p-1}$ is the cyclic shift permutation, and $\sigma_p$ the contraction between systems $p$ and $p+1$.

While the form of the optimal measurement is known, an efficient quantum circuit for implementing it was not known until our work. Our main result is
\begin{theorem}\label{thm:pbt}
    The pretty good measurement $E$ for the port-based teleportation protocol from \cref{def:PGM_PBT} can be implemented by a quantum circuit with gate complexity $\poly(p,d)$, where $p$ is the number of ports and $d$ is the dimension of the teleported quantum state.
\end{theorem}

The proof is explained in detail in the next two sections. Our construction provides the first efficient implementation of the optimal measurement $E$ as a quantum circuit of gate complexity $\poly(p,d)$. This is an exponential improvement over a trivial implementation
which has complexity $\poly(d^p)$.

The setting of port-based teleportation is naturally suited to the use of representation theory of the algebra $\A_{p,1}^d$. It is also natural to work in the mixed Schur basis, which can be achieved by applying the mixed quantum Schur transform from \cref{sec:SchTransQuantum}.\footnote{Starting from now on, for brevity we will not mention explicitly the matrix representation of the walled Brauer algebra $\psi^d_{p,q}$ and we assume that we are working in the mixed Schur basis, i.e., when using, for example, a diagram $\sigma \in \B^d_{p,q}$ it should be understood as $U_{\mathrm{Sch}} \psi^d_{p,q}(\sigma) U^\dagger_{\mathrm{Sch}}$. Moreover, we ignore multiplicity registers $V_\lambda^{\U{d}}$ corresponding to unitary group irreps in the mixed Schur--Weyl duality, i.e., we write all expressions in the Gelfand--Tsetlin basis from \cref{thm:main}.}

%%%%%%%%%%%%%%%%%%%%%%%%%%%%%%%%%%%%%%%%%%%%%%%%%%%%%%%%%%%%%%%%%%%%%%%%%%%%%%%%%%%%%%%%%%%%%%%%%%%%%%%%%%%%%%%%%%%%%%%%%%%%%%%%%%%%%%%%%%%%%%%%%%%%%%%%%%%%%%%
\subsection{Naimark's dilation}\label{sec:pbt:dilation}

Before we present our circuit, we need to explain how to dilate the POVM $E$ to a projective measurement $\Pi$ (projection-valued measure or PVM for short). That, in principle, is possible for any POVM due to the Naimark's dilation theorem. However, a simple and efficient dilation is not obvious to achieve and implement in general. After we explain how to construct such dilation explicitly, we present a construction of an efficient circuit for $E$ in the next section.

Note that the unnormalized state $\rho$ coincides with the shifted Jucys--Murphy element $d-J_{p+1}$ of $\A_{p,1}^d$, so its spectrum can be easily obtained, see \cref{lem:JMaction}. More concretely, $\rho$ is diagonal in the Gelfand--Tsetlin basis and due to \cref{lem:JMaction}:
\begin{equation}
    \rho = \sum_{\substack{\lambda \in \Irr{\A_{p,1}^{d}} \\ \lambda_r=\0 }} \rho_{\lambda}, \qquad \rho_{\lambda} \defeq \sum_{T \in \Paths(\lambda)} \of[\big]{d + \cont(T^{p} \backslash \lambda_l)} \proj{T}.
\end{equation}
Note that $\rho$ is zero on irreps $\lambda$ for which $\lambda_r \neq \0$.\footnote{Therefore our convention, from now on, is that we will drop subscript $l$ from $\lambda_l$ and will refer to it simply by $\lambda$. Moreover, all vertices $\mu$ in all levels up to $p$ in the Bratteli diagram do have the property $\mu_r = \0$, so a similar convention applies to all such $\mu$.}
Also note that due to \cref{thm:main} the generator $\sigma_p$ in the Gelfand--Tsetlin basis can be written as
\begin{equation}
    \sigma_p = \sum_{\substack{\lambda \in \Irr{\A_{p,1}^{d}} \\ \lambda_r=\0 }} \sum_{S \in \Paths_{p-1}(\lambda)} \ketbra{v_{S,\lambda}}{v_{S,\lambda}}, \quad \text{where} \quad \ket{v_{S,\lambda}} \defeq \sum_{T \in \M{S,\lambda}} c(T) \ket{T}
\end{equation}
where $c(T) = \sqrt{\frac{m_{T^{p}}}{ m_{T^{p-1}}}}$ and we define for $S \in \Paths_{p-1}(\lambda)$:
\begin{equation}
    \M{S,\lambda} \defeq \Set*{ T \in \Paths(\lambda) \given  \exists \, \mu \in \Irr{\A_{p,0}^{d}} : T = (S^{0},S^{1},\dotsc,S^{p-2},\lambda,\mu,\lambda)}.
\end{equation}
Note that $\M{S,\lambda}$ is in bijection with a subset $\AC_d(\lambda) \subseteq \AC(\lambda)$ of addable boxes to $\lambda \pt p-1$ formally defined as
\begin{equation}
    \AC_d(\lambda) \defeq \Set{ a \in \AC(\lambda) \given \ell(\lambda \cup a) \leq d }.
\end{equation}
Now we can rewrite $E_p$ in the Gelfand--Tsetlin basis as follows:
\begin{align}\label{eq:pbt_main}
    E_p = \rho^{-1/2} \sigma_p \rho^{-1/2} = \sum_{\substack{\lambda \in \Irr{\A_{p,1}^{d}} \, : \, \lambda_r=\0 \\ S \in \Paths_{p-1}(\lambda)}} \rho^{-1/2}_{\lambda} \ketbra{v_{S,\lambda}}{v_{S,\lambda}} \rho^{-1/2}_{\lambda} = \sum_{\substack{\lambda \in \Irr{\A_{p,1}^{d}} \, : \, \lambda_r=\0 \\ S \in \Paths_{p-1}(\lambda)}} \ketbra{w_{S,\lambda}}{w_{S,\lambda}},
\end{align}
with
\begin{equation}
    \ket{w_{S,\lambda}} \defeq \sum_{T \in \M{S,\lambda}} \sqrt{\frac{d_{T^{p}}}{p \cdot d_{T^{p-1}}}} \ket{T} = \sum_{a \in \AC_d(\lambda)} \sqrt{\frac{d_{\lambda \cup a}}{p \cdot d_{\lambda}}} \ket{S \circ \of{\lambda \cup a} \circ \lambda},
\end{equation}
where we used \cref{thm:main,lem:content_into_sym_u_dims,lem:ration_into_u_dims}. Since $\rho$ commutes with $\A_{p,0}^d$, the other POVM elements $E_k$ for $k \in [p]$ can be written as
\begin{equation}
    E_k = \rho^{-1/2} \pi^{k} \sigma_p \pi^{-k} \rho^{-1/2} = \pi^{k} E_p \pi^{-k}.
\end{equation}
We will denote the restriction of $E_k$ to irrep $\lambda \in \Irr{\A_{p,1}^{d}}$ by $E_k^\lambda$.

Note that for any irrep $\lambda \in \Irr{\A_{p,1}^{d}}$ with $\lambda_r = \0$ and every path $T \in \Paths(\lambda)$ the dimensions $d_{T^{p}}$ and $d_{T^{p-1}}$ coincide with the dimensions of the corresponding irreps of the symmetric groups $\S_p$ and $\S_{p-1}$, respectively.
Now recall that the Bratteli diagram of the symmetric group is the \emph{Young lattice} or \emph{Young graph} \cite{Sagan}, and the following identity holds for every $\lambda \pt p-1$ \cite{stanley2013algebraic} in the Young lattice:
\begin{equation}
    p \cdot d_\lambda = \sum_{a \in \AC(\lambda)}  d_{\lambda \cup a},
    \label{pbt:naimrak_sym_induction}
\end{equation}
where the notation $\lambda \cup a$ denotes the Young diagram in the Young lattice obtained by adding a box $a$ to $\lambda$, and $d_\lambda$ is the dimension of the symmetric group irrep $\lambda$.\footnote{The dimension $d_\lambda$ can be both understood as the number of paths from the root to a vertex $\lambda$ in the Young lattice as well as in the Bratteli diagram of $\A_{p,0}^d$, since up to level $p$ the Bratteli diagram is a subset of the full Young lattice and the procedure of adding a cell is monotonic with respect to the number of rows in $\lambda$ along a given path in the Young lattice.}

The main observation of this section is that for a Young diagram $\lambda \pt p-1$, we have $\AC_d(\lambda) = \AC(\lambda)$ if $\lambda_d = 0$ and $\AC_d(\lambda) \neq \AC(\lambda)$ if $\lambda_d > 0$. In particular, when $\lambda_d = 0$ this implies that
\begin{equation}
    \norm{\ket{w_{S,\lambda}}}^2 = \sum_{T \in \M{S,\lambda}} \frac{d_{T^{p}}}{p \cdot d_{T^{p-1}}} = \sum_{a \in \AC_d(\lambda)} \frac{d_{\lambda \cup a}}{p \cdot d_\lambda} = \sum_{a \in \AC(\lambda)} \frac{d_{\lambda \cup a}}{p \cdot d_\lambda} = 1,
\end{equation}
so $E^\lambda_p$ is an orthogonal projector. Since the cyclic shift $\pi$ acts unitarily, all $E^\lambda_i$ are orthogonal projectors as well. Because $E^\lambda$ provides a resolution of the identity in the irreducible representation $\lambda$, the POVM $E^\lambda$ restricted to the irreducible representation $\lambda$ with $\lambda_d = 0$ is actually a PVM on that irreducible representation. We will replace $E^\lambda$ by $\Pi^\lambda$ from now on to indicate that $E^\lambda$ is actually a PVM.

However, for the irreps $\lambda$ with $\lambda_d > 0$ the POVM $E^\lambda$ is not a PVM because $\AC(\lambda) = \AC_d(\lambda) \sqcup \set{(d+1,1)}$ and the vectors $\ket{w_{S,\lambda}}$ are not normalized anymore:
\begin{equation}\label{eq:pbt_norm_<1}
    \norm{\ket{w_{S,\lambda}}}^2 =  \sum_{T \in \M{S,\lambda}} \frac{d_{T^{p}}}{p \cdot d_{T^{p-1}}} = \sum_{a \in \AC_d(\lambda)} \frac{d_\mu}{p \cdot d_\lambda} = \of[\bigg]{\sum_{a \in \AC(\lambda)} \frac{d_\mu}{p \cdot d_\lambda}} - \frac{d_{\lambda \cup (d+1,1)}}{p \cdot d_\lambda} = 1 - \frac{d_{\lambda \cup (d+1,1)}}{p \cdot d_\lambda} < 1,
\end{equation}
where $\lambda \cup (d+1,1)$ denotes the Young diagram obtained from $\lambda$ by adding a cell with coordinates $(d+1,1)$, so that $\ell \of*{\lambda \cup (d+1,1)}=d+1$. The vertex corresponding to this Young diagram does not exist in the Bratteli diagram of $\A_{p,1}^d$. Fortunately, \cref{eq:pbt_norm_<1} suggests immediately how to construct a Naimark's dilation $\Pi^\lambda$ of $E^\lambda$ for $\lambda$ with $\lambda_d > 0$. For this construction, one needs to modify the Bratteli diagram of $\A_{p,1}^d$ by adding certain vertices to each level of the diagram. Then the set of all paths in this modified Bratteli diagram will define a new basis for the Naimark dilated Hilbert space.

More concretely, to each level $k \leq p$ of the Bratteli diagram of $\A_{p,1}^d$ we add all possible vertices labelled by all Young diagrams $\nu \pt k$ such that $\nu_{d+1} = 1$ (if such Young diagrams exist for a given level $k$).
An edge between a pair of Young diagrams in two consecutive levels is added if the latter diagram can be obtained by adding a cell to the previous one. This procedure ensures that all the levels up to $p$ of the new Bratteli diagram form a subset of the Young lattice, such that for every vertex at level $p$ the irrep dimensions still satisfy \cref{pbt:naimrak_sym_induction}. The basis for the Naimark dilated Hilbert space consists of all paths in this modified Bratteli diagram, which we denote by $\widetilde{\Paths}$ which we formally define as follows. For every $\lambda \in \Irr{\A_{p,1}^{d}}$, if $\lambda_r = \0 $ we define
\begin{equation}
    \widetilde{\Paths}(\lambda) \defeq
        \set*{ T = (T^{1},\dotsc,T^{p},\lambda) \in \Irr{\A_{1,0}^{d+1}} \times \dotsb \times \Irr{\A_{p,0}^{d+1}} \times \Irr{\A_{p,1}^{d}} \mid T \text{ satisfies \cref{def:bratteli_rule}} },
\end{equation}
where
\begin{equation}\label{def:bratteli_rule}
    T^{k}_{l,d+1} \leq 1 \quad \forall k \, \in [p], \quad \text{ and } \quad T^{k} \backslash T^{k-1} \in \AC(T^{k-1})  \quad \forall k \in [p], \quad \text{ and } \quad T^{p} \backslash \lambda \in \RC(T^{p}).
\end{equation}
If $\lambda_r \neq \0$ then
\begin{equation}
     \widetilde{\Paths}(\lambda) \defeq \Paths(\lambda).
\end{equation}
The full modified Bratteli diagram can be thought of as a disjoint union of modified sets of paths for every $\lambda \in \Irr{\A_{p,1}^{d}}$:
\begin{equation}
    \widetilde{\Paths} \defeq \bigsqcup_{\lambda \in \Irr{\A_{p,1}^{d}}} \widetilde{\Paths}(\lambda).
\end{equation}

The action of the generators $\sigma_1,\dotsc,\sigma_{p-1}$ of $\S_p$ in this modified Bratteli diagram is given by \cref{thm:main} when $\lambda_r \neq \0$, and by the trivial generalization of these formulas to all paths in $\widetilde{\Paths}(\lambda)$ if $\lambda_r = \0$. For this new Bratteli diagram, we define the dilated versions $\ket{\widetilde{w}_{S,\lambda}}$ of vectors $\ket{w_{S,\lambda}}$, for $S \in \widetilde{\Paths}_{p-1}(\lambda)$ as follows:
\begin{equation}
    \ket{\widetilde{w}_{S,\lambda}} \defeq \sum_{T \in \widetilde{M}(S,\lambda)} \sqrt{\frac{d_{T^{p}}}{p \cdot d_{T^{p-1}}}} \ket{T},
\end{equation}
where $\widetilde{M}(S,\lambda)$ is defined in the same way as the set $\M{S,\lambda}$, but for the dilated Bratteli diagram $\widetilde{\Paths}$, i.e., for every $S \in \widetilde{\Paths}_{p-1}(\lambda)$:
\begin{equation}
    \widetilde{M}(S,\lambda) \defeq \Set*{ T \in \widetilde{\Paths}(\lambda) \given \exists \, \mu \in \Irr{\A_{p,0}^{d+1}} : T = (S^{0},S^{1},\dotsc,S^{p-2},\lambda,\mu,\lambda), \, \mu_{d+1} \leq 1}.
\end{equation}
Note that $\widetilde{M}(S,\lambda)$ is in bijection now with the set $\AC(\lambda)$ and we can rewrite $\ket{\widetilde{w}_{S,\lambda}}$ as
\begin{equation}
    \ket{\widetilde{w}_{S,\lambda}} = \sum_{a \in \AC(\lambda)} \sqrt{\frac{d_{\lambda \cup a}}{p \cdot d_{\lambda}}} \ket{S \circ \of{\lambda \cup a} \circ \lambda},
\end{equation}
Therefore in the dilated space since \cref{pbt:naimrak_sym_induction} holds and all vertices $\lambda \cup a$, which can be obtained from a vertex $\lambda$ at level $p-1$ by adding a box $a$, exist then
\begin{equation}
    \norm{\ket{\widetilde{w}_{S,\lambda}}}^2 = \sum_{T \in \widetilde{M}(S,\lambda)} \frac{d_{T^{p}}}{p \cdot d_{T^{p-1}}} = \sum_{a \in \AC(\lambda)} \frac{d_{\lambda \cup a}}{p \cdot d_{\lambda}} = 1.
\end{equation}
Consequently, in the dilated space our POVM $E^\lambda$ is actually a PVM, which we denote by $\Pi^\lambda$. From now on assume that we work in the dilated space and we want to implement the PVM $\Pi = \set{\Pi_k}_{k=0}^p$, where for every $k \in [p]$:
\begin{align}\label{def:pbt_dilated_pvm}
    \Pi_k = \hspace{-10pt} \sum_{\substack{\lambda \in \Irr{\A_{p,1}^{d}} \\ \lambda_r=\0 }} \Pi^\lambda_k, \qquad \Pi^\lambda_k \defeq \pi^{k} \Pi^\lambda_p \pi^{-k} = \hspace{-10pt} \sum_{S \in \widetilde{\Paths}_{p-1}(\lambda)} \hspace{-10pt} \pi^k \ketbra{\widetilde{w}_{S,\lambda}}{\widetilde{w}_{S,\lambda}} \pi^{-k}, \qquad  \Pi_0 = I - \sum_{k=1}^{p}\Pi_k.
\end{align}
We have provided a \textit{Wolfram Mathematica} notebook implementing our construction in \cite{github}.

%%%%%%%%%%%%%%%%%%%%%%%%%%%%%%%%%%%%%%%%%%%%%%%%%%%%%%%%%%%%%%%%%%%%%%%%%%%%%%%%%%%%%%%%%%%%%%%%%%%%%%%%%%%%%%%%%%%%%%%%%%%%%%%%%%%%%%%%%%%%%%%%%%%%%%%%%%%%%%%%%%%%%%%%%%%%%%%%%%%%%%%%%%%%%%%%%%
\subsection{Efficient quantum circuit for the pretty good measurement}

Using the results of \cref{sec:pbt:dilation}, our task now is to implement the PVM $\Pi$ from \cref{def:pbt_dilated_pvm}. To explain our construction, let us illustrate the main idea in a simpler example, so first we reformulate the problem in a more abstract language for a simplified setting of rank $1$ projectors.

Suppose we have the ability to implement $n+1$ orthogonal vectors $\ket{x_k}$ for $k \in \set{0,\dotsc,n}$ via some easy-to-implement unitaries $U_k$ starting from known basis vector $\ket{0}$ as $\ket{x_k} = U_k \ket{0}$. Assume, that we also know that these vectors comprise a PVM $X = \set*{\ketbra{x_k}{x_k}}_{k=0}^{n}$. We can implement PVM $X$ via the following two ideas:
\begin{enumerate}
    \item Define a unitary $V$ as
        \begin{equation}
            V \defeq \sum_{k=0}^{n} \omega^{k}_{n+1} \ketbra{x_k}{x_k}
        \end{equation}
    where $\omega_{n+1}$ is root of unity of order $n+1$. Note, that we can implement the unitary $V$ efficiently if we have easy-to-implement unitaries $U_k$ via the following circuit:
    
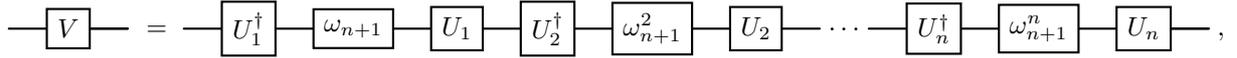
\begin{equation*}
        \begin{quantikz}
        \qw &\gate{V} & \qw
        \end{quantikz}
        \hspace{0.1em}=
        \begin{quantikz}
        \qw &\gate{U^\dagger_1} & \gate{\omega_{n+1}} & \gate{U_1}   &  \gate{U^\dagger_2} & \gate{\omega^{2}_{n+1}} & \gate{U_2}  & \HDots & \gate{U^\dagger_{n}} & \gate{\omega^{n}_{n+1}} & \gate{U_n}  & \qw
        \end{quantikz},
    \end{equation*}
    where \begin{quantikz}\gate{\omega^{k}_{n+1}}\end{quantikz} represents the gate $\omega^{k}_{n+1} \ketbra{0}{0} + \ketbra{0^\perp}{0^\perp}$.
    \item Note that implementing $V^i$ is easy:
    
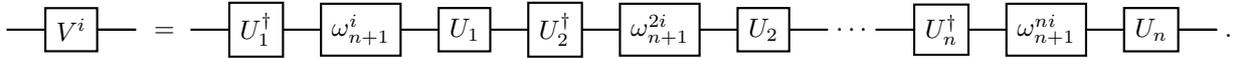
\begin{equation*}
        \begin{quantikz}
        \qw &\gate{V^i} & \qw
        \end{quantikz}
        \hspace{0.1em}=
        \begin{quantikz}
        \qw &\gate{U^\dagger_1} & \gate{\omega^{i}_{n+1}} & \gate{U_1}   &  \gate{U^\dagger_2} & \gate{\omega^{2i}_{n+1}} & \gate{U_2}  & \HDots & \gate{U^\dagger_{n}} & \gate{\omega^{ni}_{n+1}} & \gate{U_n}  & \qw
        \end{quantikz}.
    \end{equation*}
    Therefore we can use the standard phase estimation circuit to measure a given state $\ket{\Psi}$ with respect to the PVM $X$ as follows:
    
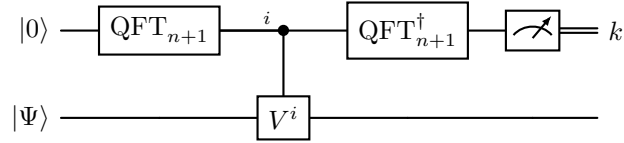
\begin{equation*}
    \begin{quantikz}%[classical gap = 2pt]
        \lstick{$\ket{0}$}  & \gate{\mathrm{QFT}_{n+1}} & \ctrl{1} \wire[l][1]["i"{above,pos=0.2}]{a} & \gate{\mathrm{QFT}_{n+1}^\dagger} & \meter{} & \setwiretype{c} \rstick{$k$} \\
        \lstick{$\ket{\Psi}$} & & \gate{V^i} & & &
    \end{quantikz}
    \end{equation*}
\end{enumerate}
The above ideas are trivially extended to higher-rank PVMs. Now before presenting our circuit for $\Pi$ from \cref{def:pbt_dilated_pvm}, we need to define a unitary $W$, which can be used to prepare states $\ket{\widetilde{w}_{S,\lambda}} = \ket{S}\ket{\widetilde{w}_{\lambda}}$ for every $\lambda \in \Irr{\A_{p,1}^{d}}$ with $\lambda_r = \0$ and $\S \in \widetilde{\Paths}_{p-1}(\lambda)$, where we denote
\begin{equation}\label{def:w_lambda_tilde}
    \ket{\widetilde{w}_{\lambda}} \defeq \sum_{a \in \AC(\lambda)} \sqrt{\frac{d_{\lambda \cup a}}{p \cdot d_{\lambda}}} \ket{\lambda \cup a}.
\end{equation}
Now we define $W$ as a unitary, which prepares $\ket{\widetilde{w}_{\lambda}}$ conditioned on $\ket{\lambda}$:
\begin{equation}\label{def:W_lambda}
    W \defeq \sum_{\lambda} W_\lambda \otimes \ketbra{\lambda}{\lambda}, \qquad W_\lambda \ket{0} \defeq \ket{\widetilde{w}_{\lambda}}
\end{equation}
$W_\lambda$ is a rotation matrix of size at most $(d+1) \times (d+1)$ with easy-to-compute coefficients determined from \cref{def:w_lambda_tilde} and \cref{lem:content_into_sym_u_dims}.

Assume we start at the state $\ket{S}\ket{0}\ket{\lambda} \defeq \ket{S^{0}} \ket{S^{1}} \ket{S^{2}} \dotsc \ket{S^{p-1}} \ket{0} \ket{\lambda}$ defined for arbitrary $S \in \widetilde{\Paths}_{p-1}$, where $\ket{0}$ is some basis state of the register, corresponding to the $p$-th level of the dilated basis $\widetilde{\Paths}$. Then we can prepare a state $\ket{S}\ket{\widetilde{w}_{\lambda}}\ket{\lambda}$ as follows:
\begin{equation}
    I \otimes W \ket{S}\ket{0}\ket{\lambda} = \ket{S}\ket{\widetilde{w}_{\lambda}}\ket{\lambda},
\end{equation}
where identity $I$ acts on the register $\ket{S}$.

Now following the outlined prescription we construct the efficient circuit for the pretty good measurement $E$ for the PBT protocol in \cref{fig:pbt_circuit}. The circuit should be understood as acting on the dilated space spanned by $(T^{0},T^{1},\dotsc, T^{p},\lambda) \in \widetilde{\Paths}$, which forms the Gelfand--Tsetlin basis. The ancilla registers used for the dilation should be understood as discarded after the computation. Let us comment on \cref{fig:pbt_circuit}.

Firstly, the Schur transform maps the standard basis into the mixed Schur basis which is labelled usually by $\ket{M,T}$, where $M$ is a Gelfand--Tsetlin pattern and $T$ is a mixed Young tableau. We assume a tensor product structure in the mixed Schur basis for different vertices $T^{i}$ of the path $T \in \widetilde{\Paths}$, where all registers $T^{2},\dotsc,T^{p}$ are assumed to be dilated, according to the procedure explained in \cref{sec:pbt:dilation}. Since $T^{0}$ and $T^{1}$ are always constant, we omit those registers from the diagram. The last level $T^{p+1}$ of the path $T$ is labelled by $\lambda$ and indicates the irreducible representation. The cyclic permutation gate $\pi = (12\dotsc p) = \sigma_{1}\sigma_{2}\dotsc\sigma_{p-1}$ acts only on $p-1$ wires of the dilated Gelfand--Tsetlin basis, and each of the transpositions $\sigma_i$ act only locally in the registers $T^{i-1}, T^{i}, T^{i+1}$ ($\sigma_1$ acts only on $T^{2}$, and $\sigma_2$ acts only on $T^{2}, T^{3}$). $W$ prepares the state $\ket{\widetilde{w}_{\lambda}}$ conditioned on $\lambda$, i.e. $W_{\lambda} \ket{0} = \ket{\widetilde{w}_{\lambda}}$ and are controlled on $\lambda$ as well. The phase gates $\omega^{ki}_{p+1}$ act only on the state $\ket{0}$ in the register $T^{p}$, when they are controlled on the condition $T^{p-1} = T^{p+1}= \lambda$. Finally, the measured outcome $k = 0$ corresponds to the failure of the protocol, otherwise $k \in [p]$ indicates the port, where Bob should find the teleported state $\ket{\psi}$ of dimension $d$.

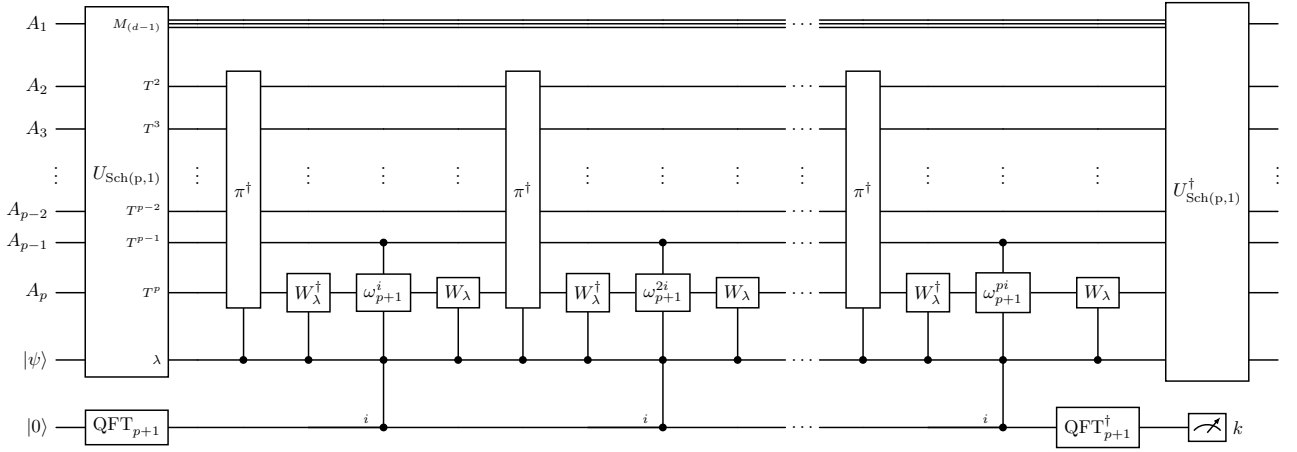
\begin{figure}[!ht]
\centering
\resizebox{\textwidth}{!}{
\begin{quantikz}[classical gap = 2pt]
\lstick{$A_1$} & \gateSchurPone \gateoutput{$M_{(d-1)}$}  & \qb         & & & & & & & & & \HDots & & & & & \gateSchurInvPone & \qq \\
\lstick{$A_2$} & \gateoutput{$T^{2}$}            & \qq    & \gateCyc  & & & & \gateCyc  & & & & \HDots & \gateCyc & & & & & \\
\lstick{$A_3$} & \gateoutput{$T^{3}$}            & \qq    &           & & & &           & & & & \HDots &          & & & & & \\
\setwiretype{n} \vdots &                         & \vdots &           & \vdots & \vdots & \vdots & & \vdots & \vdots & \vdots & \vdots & & \vdots & \vdots & \vdots & & \vdots \\
\lstick{$A_{p-2}$}     & \gateoutput{$T^{p-2}$}  & \qq    &           & & & & & & & & \HDots & & & & & & \\
\lstick{$A_{p-1}$}     & \gateoutput{$T^{p-1}$}  & \qq    &           &           & \ctrl{1}  &                   & & & \ctrl{1}  & & \HDots & &                   & \ctrl{1}   & & & \\
\lstick{$A_p$}         & \gateoutput{$T^{p}$}    & \qq    &           & \gateWinv & \gateR{i} & \gateW  & & \gateWinv & \gateR{2i} & \gateW & \HDots & & \gateWinv & \gateR{pi} & \gateW & & \\
\lstick{$\ket{\psi}$}  & \gateoutput{$\lambda$}  & \qq    & \ctrl{-1} & \ctrl{-1} & \ctrl{-1} & \ctrl{-1} & \ctrl{-1} & \ctrl{-1} & \ctrl{-1} & \ctrl{-1} & \HDots & \ctrl{-1}  & \ctrl{-1} & \ctrl{-1} & \ctrl{-1} & & \\
\lstick{$\ket{0}$} & \gateQFT & & & & \ctrl{-1} \wire[l][1]["i"{above,pos=0.2}]{a} & & & & \ctrl{-1} \wire[l][1]["i"{above,pos=0.2}]{a} & & \HDots & & & \ctrl{-1} \wire[l][1]["i"{above,pos=0.2}]{a} & \gateQFTinv & \meter{} \rstick{$k$}
\end{quantikz}
}
\caption{The circuit implementation of the pretty good measurement for port-based teleportation. The registers $T^{2},T^{3},\dotsc,T^{p}$ are assumed to be dilated according to \cref{sec:pbt:dilation}. The total qubit and gate cost is upper bounded by $\poly(p,d)$.
}
\label{fig:pbt_circuit}
\end{figure}

Now we argue that the complexity of our circuit is $\poly(p,d)$:
\begin{enumerate}
    \item The complexity of implementing the mixed Schur transform according to \cref{sec:SchTransQuantum} is $O(\poly(p,d)$. The number of ancillas qubits needed to implement the mixed Schur transform isometry and create a Naimark's dilation after the mixed Schur transfer is also polynomial $\poly(d,\log(p))$. This is so because the number of compositions of the integer $p$ into $d$ non-negative parts is $\binom{p+d-1}{p}$, so $\log\of*{\binom{p+d-1}{p}} = \poly(d,\log(p))$.
    \item The complexity of implementing $\pi = \sigma_{1}\sigma_{2}\dotsc\sigma_{p-1}$ is also $\poly(p,d)$, since each transposition operator $\sigma_{i}$ acts locally on the registers $T^{i-1}, T^{i}, T^{i+1}$. Namely $\sigma_{i}$ is a $2 \times 2$ rotation in the register $ T^{i}$ controlled from $T^{i-1}, T^{i+1}$ and $\lambda$. So using the standard methods (e.g. Given's rotations) it can be implemented with complexity upper bounded by $\poly(d,\log(p))$.\footnote{Alternatively, one can go back to the standard basis via inverse mixed Schur transform, implement a simple permutation gate corresponding to $\pi$ and then go back to the Gelfand--Tsetlin basis via mixed Schur transform. That also would also have complexity $\poly(p,d)$}
    \item $W$ operator implements $\lambda$-controlled $(d+1) \times (d+1)$ operator $W_\lambda$ on the register $T^{p}$. The number of qubits in the register $T^{p}$ is $\poly(d,\log(p))$. The coefficients of the matrix $W_\lambda$ are determined from \cref{def:w_lambda_tilde} and \cref{lem:content_into_sym_u_dims} and are easy to compute in classical time $\poly(d,\log(p))$. Therefore, implementing $W$ would have the gate complexity $\poly(d,\log(p))$.
    \item $\omega^{ki}_{p+1}$ denotes a simple-to-implement gate $\omega^{ki}_{p+1} \ketbra{0}{0} + \ketbra{0^\perp}{0^\perp}$ in the register $T^{p}$ conditioned on the registers $T^{p-1}=T^{p+1}=\lambda$. This has complexity $\poly(d,\log(p))$.
    \item Finally, the complexity of the Quantum Fourier Transform $\mathrm{QFT}_{p+1}$ is $\poly(p)$.
\end{enumerate}

%%%%%%%%%%%%%%%%%%%%%%%%%%%%%%%%%%%%%%%%%%%%%%%%%%%%%%%%%%%%%%%%%%%%%%%%%%%%%%%%%%%%%%%%%%%%%%%%%%%%%%%%%%%%%%%%%%%%%%%%%%%%%%%%%%%%%%%%%%%%%%%%%%%%%%%%%%%%%%%
\subsection{Exponentially improved lower bound for non-local quantum computation}

Port-based teleportation has interesting applications in holography and non-local quantum computation \cite{may2019quantum,may2022complexity}, where it was argued that the complexity of the local operation controls the amount of entanglement needed to implement it non-locally, using ideas from AdS/CFT correspondence. In particular, it was derived in \cite[Lemma 9]{may2022complexity} that port-based teleportation can be used to lower bound the amount of entanglement needed to implement a given channel (from a large class of one-sided quantum channels) non-locally in terms of the so-called \emph{interaction-class circuit complexity} $\mathcal{C}$ \cite[Definition 3]{may2022complexity} denoted by $\mathcal{C}$. Port-based teleportation can also be used to find an upper bound \cite{beigi2011simplified,speelman:LIPIcs:2016:6690,may2022complexity,}. The bounds read as
\begin{equation}
    \Omega \of*{ \log \log \mathcal{C}} \leq E_c \leq O \of*{\mathcal{C} \cdot 2^{\mathcal{C}}},
\end{equation}
where $E_c$ is the entanglement cost needed to implement non-locally a unitary with complexity $\mathcal{C}$ \cite{may2022complexity}. The derivation of the lower bound uses a trivial upper bound $\exp\of{O(p)}$ for the complexity of the port-based teleportation in terms of the number of ports $p$, see \cite[Equation 47]{may2022complexity}. It is pointed out in \cite[page 28]{may2022complexity} that a better implementation of the port-based teleportation protocol would lead to a better lower bound.

Complexity of our implementation of PBT protocol is $\poly(p)$, therefore this immediately translates, according to \cite[Lemma 9]{may2022complexity}, to a better lower bound:
\begin{equation}
    \Omega \of*{ \log \mathcal{C}} \leq E_c,
\end{equation}
thus improving exponentially upon the previous bound.\\
\par

\section*{Acknowledgements}

DG thanks Tudor Giurgica-Tiron, Quynh Nguyen, Aram Harrow, Hari Krovi, Philip Verduyn Lunel, Rene Allerstorfer and Florian Speelman for useful discussions.
DG, AB, and MO were supported by an NWO Vidi grant (Project No VI.Vidi.192.109).

%%%%%%%%%%%%%%%%%%%%%%%%%%%%%%%%%%%%%%%%%%%%%%%%%%%%%%%%%%%%%%%%%%%%%%%%%%%%%%%%%%%%%%%%%%%%%%%%%%%%%%%%%%%%%%%%%%%%%%%%%%%%%%%%%%%%%%%%%%%%%%%%%%%%%%%%%%%%%%%%

\printbibliography

%%%%%%%%%%%%%%%%%%%%%%%%%%%%%%%%%%%%%%%%%%%%%%%%%%%%%%%%%%%%%%%%%%%%%%%%%%%%%%%%%%%%%%%%%%%%%%%%%%%%%%%%%%%%%%%%%%%%%%%%%%%%%%%%%%%%%%%%%%%%%%%%%%%%%%%%%%%%%%%%
%%%%%%%%%%%%%%%%%%%%%%%%%%%%%%%%%%%%%%%%%%%%%%%%%%%%%%%%%%%%%%%%%%%%%%%%%%%%%%%%%%%%%%%%%%%%%%%%%%%%%%%%%%%%%%%%%%%%%%%%%%%%%%%%%%%%%%%%%%%%%%%%%%%%%%%%%%%%%%%%

\begin{appendix}
\markboth{APPENDIX}{APPENDIX}

\section{Clebsch--Gordan coefficients}\label{sec:CGcoefficients}

This appendix summarizes formulas from \cite[Chapter~18]{Vilenkin1992} for evaluating the (dual) Clebsch--Gordan\footnote{Also known as Wigner coefficients \cite{biedenharn1968pattern,HarrowThesis}.} coefficients of $\U{d}$.
%in terms of the \emph{scalar factors} $\rw{\m_n}{\pm i}{\m_{n-1}}{\pm j}$.
Recall from \cref{sec:GT} that a Gelfand--Tsetlin pattern $M \in \GT(\lambda,d)$ is a column vector
\begin{equation}
    M =
    \bmx{
        \m_d \\
        \vdots \\
        \m_1
    },
\end{equation}
where $\m_n = (m_{1,n}, \dotsc, m_{n,n})$ are row vectors of non-decreasing integers subject to interlacing relations~\eqref{eq:GTpatterns-relations}.
For any row $\m_n$ and integer $i \in \set{1, \dotsc, n}$, we denote by $\m_n^{\pm i}$ the vector $\m_n$ with entry $m_{i,n}$ replaced by $m_{i,n} \pm 1$.
Let us fix any symbol $x \in [d]$.
We define Gelfand--Tsetlin patterns $M^{+}$ and $M^{-}$ by modifying the top $d-x+1$ rows of $M$ as follows:
\begin{equation}
    \label{eq:Mprime}
    M =
    \bmx{
        \m_d \\
        \vdots \\
        \m_{x} \\
        \m_{x-1} \\
        \vdots \\
        \m_1
    }, \qquad
    M^{+} =
    \bmx{
        \m_d^{+i_d} \\
        \vdots \\
        \m_{x}^{+i_x} \\
        \m_{x-1} \\
        \vdots \\
        \m_1
    }, \qquad
    M^{-} =
    \bmx{
        \m_d^{-i_d} \\
        \vdots \\
        \m_{x}^{-i_x} \\
        \m_{x-1} \\
        \vdots \\
        \m_1
    },
\end{equation}
for some integers $i_x, \dotsc, i_d$ where $1 \leq i_j \leq j$.
Intuitively, this means that the semistandard tableau $M^{+}$ is obtained from $M$ by adding a box containing $x$ in the row $i_x$, and then consecutively bumping the entries $j$ from row $i_j$ downwards the tableau.

For any $x \in [d]$, the Gelfand--Tsetlin patterns corresponding to $x$ and its dual are defined as follows:
\begin{align}
    X^+ &\defeq
    \bmx{
      1 \; 0 \; 0 \dots 0 \; 0 \; 0 \\
      1 \; 0 \; 0 \dots 0 \; 0 \\
        \dots \dots \dots \\
      1 \; 0 \dots 0 \\
      0 \dots 0 \\
        \dots \\
      0 \; 0 \\
      0
    }
    \begin{array}{c}
        d \\ d-1 \\ \dots \\ x \\ x-1 \\ \dots \\ 2 \\ 1
    \end{array}, &
    X^- &\defeq
    \bmx{
      0 \; 0 \; 0 \dots 0 \; 0 \; -1 \\
      0 \; 0 \; 0 \dots 0 \; -1 \\
        \dots \dots \dots \\
      0 \; 0 \dots -1 \\
      0 \dots 0 \\
        \dots \\
      0 \; 0 \\
      0
    }
    \begin{array}{c}
        d \\ d-1 \\ \dots \\ x \\ x-1 \\ \dots \\ 2 \\ 1
    \end{array}.
\end{align}
Now we can define the Clebsch--Gordan coefficients $c_{M^{\pm},M}^{x,\pm}$ uniformly as
\begin{equation}
    c_{M^{\pm},M}^{x,\pm} \defeq c_{M^{\pm},M}^{X^{\pm}},
\end{equation}
where $c_{M^{\pm},M}^{X^{\pm}}$ is the product of \emph{reduced Wigner coefficients}:
%%%%%%%%%%%%%%%%%%%%%%%%%%%%%%%%
\begin{equation}
    c_{M^{\pm},M}^{X^{\pm}} =
    \prod_{n=x+1}^{d}
    \left(\begin{array}{cc|c}
        \m_{n}   & \mathbf{x}^{\pm}_{n}   & \m_n^{\pm i_n} \\
        \m_{n-1} & \mathbf{x}^{\pm}_{n-1} & \m_{n-1}^{\pm i_{n-1}}
    \end{array}\right).
\end{equation}
%%%%%%%%%%%%%%%%%%%%%%%%%%%%%%%%
The \emph{reduced Wigner coefficients} \cite[p.~152]{HarrowThesis} or \emph{scalar factors} \cite[p.~385]{Vilenkin1992} are defined as follows.
We take two consecutive rows $\m_n$ and $\m_{n-1}$
($1 < n \leq d$)
of a Gelfand--Tsetlin pattern $M$ and modify them at positions
$1 \leq i \leq n$ and
$1 \leq j \leq n-1$.
The corresponding reduced Wigner coefficients are
\begin{align}
    \label{def:reduced_wigner+0}
    \left(\begin{array}{cc|c}
        \m_n & (1, \mathbf{0}_{n-1}) & \m_n^{+i} \\
        \m_{n-1} & (0, \mathbf{0}_{n-2}) & \m_{n-1}
    \end{array}\right)
    &=
    \left|\frac{\prod_{j=1}^{n-1}\left(\ell_{j, n-1}-\ell_{i,n}-1\right)}{\prod_{j \neq i}\left(\ell_{j,n}-\ell_{i,n}\right)}\right|^{1/2}, \\
    \label{def:reduced_wigner++}
    \left(\begin{array}{cc|c}
        \m_n & (1,\mathbf{0}_{n-1}) & \m_n^{+i} \\
        \m_{n-1} & (1,\mathbf{0}_{n-2}) & \m_{n-1}^{+j}
    \end{array}\right)
    &=
    S(i, j)\left|\frac{\prod_{k \neq j}\left(\ell_{k, n-1}-\ell_{i, n}-1\right) \prod_{k \neq i}\left(\ell_{k, n}-\ell_{j, d-1}\right)}{\prod_{k \neq i}\left(\ell_{k,n}-\ell_{i, n}\right) \prod_{k \neq j}\left(\ell_{k, n-1}-\ell_{j, n-1}-1\right)}\right|^{1/2}, \\
    \label{def:reduced_wigner-0}
    \left(\begin{array}{cc|c}
        \m_n & (\mathbf{0}_{n-1},-1) & \m_n^{-i} \\
        \m_{n-1} & (\mathbf{0}_{n-2}, 0) & \m_{n-1}
    \end{array}\right)
    &=
    \left|\frac{\prod_{j=1}^{n-1}\left(\ell_{j, n-1}-\ell_{i, n}\right)}{\prod_{j \neq i}\left(\ell_{j, n}-\ell_{i, n}\right)}\right|^{1/2}, \\
    \label{def:reduced_wigner--}
    \left(\begin{array}{cc|c}
        \m_n & (\mathbf{0}_{n-1},-1) & \m_n^{-i} \\
        \m_{n-1} & (\mathbf{0}_{n-2},-1) & \m_{n-1}^{-j}
    \end{array}\right)
    &=
    S(i, j)\left|\frac{\prod_{k \neq j}\left(\ell_{k, n-1}-\ell_{i, n}\right) \prod_{k \neq i}\left(\ell_{k,n}-\ell_{j, n-1}+1\right)}{\prod_{k \neq i}\left(\ell_{k, n}-\ell_{i, n}\right) \prod_{k \neq j}\left(\ell_{k, n-1}-\ell_{j, n-1}+1\right)}\right|^{1/2}.
\end{align}
where $\mathbf{0}_n$ denotes a row vector with $n$ zeros, $\ell_{k,s} \defeq m_{k,s} - k$, $S(i,j) \defeq 1$ if $i \leq j$ and $S(i,j) \defeq -1$ if $i > j$.

More explicitly, the Clebsch--Gordan coefficient $c_{M^{+},M}^{x,+}$ is equal to the product of the reduced Wigner coefficients obtained by cutting the Gelfand--Tsetlin patterns $M$, $x$ and $M^{\pm}$ into pairs of consecutive rows:
\begin{equation}
    c_{M^{+},M}^{x,+}=
    \begin{aligned}
    &\left(\begin{array}{cc|c}
    \m_{x} & (1, \mathbf{0}) & \m_{x}^{+i_x} \\
    \m_{x-1} & (0, \mathbf{0}) & \m_{x-1}
    \end{array}\right)
    \end{aligned}
    \cdot \prod_{n=x+1}^{d}
    \begin{aligned}
    &\left(\begin{array}{cc|c}
    \m_n & (1, \mathbf{0}) & \m_n^{+i_n} \\
    \m_{n-1} & (1, \mathbf{0}) & \m_{n-1}^{+i_{n-1}}
    \end{array}\right)
    \end{aligned}
    .
\end{equation}
On the other hand, for arbitrary $M^{+}$ which is not of the form (\ref{eq:Mprime}), Clebsch--Gordan coefficient $c_{M^{+},M}^{x,+}=0$.
A dual Clebsch--Gordan coefficient $c_{M^{-},M}^{x,-}$ is given by the product of dual reduced Wigner coefficients:
\begin{equation}
    c_{M^{-},M}^{x,-}=
    \begin{aligned}
    &\left(\begin{array}{cc|c}
    \m_{x} & (\mathbf{0},-1) & \m_{x}^{-i_x} \\
    \m_{x-1} & (\mathbf{0},0) & \m_{x-1}
    \end{array}\right)
    \end{aligned}
    \cdot \prod_{n=x+1}^{d}
    \begin{aligned}
    &\left(\begin{array}{cc|c}
    \m_n & (\mathbf{0},-1) & \m_n^{-i_n} \\
    \m_{n-1} & (\mathbf{0},-1) & \m_{n-1}^{-i_{n-1}}
    \end{array}\right)
    \end{aligned}
    .
\end{equation}
On the other hand, for arbitrary $M^{-},M$ which is not of the form (\ref{eq:Mprime}), dual Clebsch--Gordan coefficient $c_{M^{-},M}^{x,-}=0$.

We can summarize the above definitions succinctly as follows. We can define Clebsch--Gordan coefficients $c_{N_{(k)},M_{(k)}}^{x,\pm} = 0$ for arbitrary $k \in [d]$, $\lambda \in \IrrU{k}$, $N_{(k)}, M_{(k)} \in \GT(\lambda,k)$, $x \in [k]$ recursively as 
\begin{equation}\label{app:cg_recursion}
    c_{N_{(k)},M_{(k)}}^{x,\pm} = \rwpm{\m_k}{\n_k}{\m_{k-1}}{\n_{k-1}} \cdot c^{x,\pm}_{N_{(k-1)},M_{(k-1)}},
\end{equation}
where for every $k \in [d]$ and Gelfand--Tsetlin patterns of length $k-1$ we define
\begin{equation}\label{app:cg_recursion_corner_case}
    c^{k,\pm}_{N_{(k-1)},M_{(k-1)}} \defeq \delta_{N_{(k-1)},M_{(k-1)}},
\end{equation}
and the coefficients $\rwpm{\m_k}{\n_k}{\m_{k-1}}{\n_{k-1}}$ are defined for $\m_{k-1} \squb \m_k$ and $\n_{k-1} \squb \n_k$ as
\begin{equation}\label{app:rw+}
    \rwsgn{+}{\m_k}{\n_k}{\m_{k-1}}{\n_{k-1}} \defeq 
    \begin{cases}
       \left(\begin{array}{cc|c}
        \m_k & (1, \mathbf{0})& \m_k^{+i} \\
        \m_{k-1} & (1, \mathbf{0}) & \m_{k-1}^{+j}
        \end{array}\right) & \n_k = \m_k^{+i},\, \n_{k-1} = \m_{k-1}^{+j}\, \text{for some $i \in [k]$, $j \in [k-1]$}  \\
        \left(\begin{array}{cc|c}
        \m_k & (1, \mathbf{0}) & \m_k^{+i}, \\
        \m_{k-1} & (\mathbf{0},0) & \m_{k-1}^{+j}
        \end{array}\right) & \n_k = \m_k^{+i},\, \text{for some $i \in [k]$}, \\
       1 & \m_k = \n_k,\, \m_{k-1} = \n_{k-1}, \\
       0 &\text{otherwise}
    \end{cases}
\end{equation}
and
\begin{equation}\label{app:rw-}
    \rwsgn{-}{\m_k}{\n_k}{\m_{k-1}}{\n_{k-1}} \defeq 
    \begin{cases}
       \left(\begin{array}{cc|c}
        \m_k & (\mathbf{0},-1) & \m_k^{-i} \\
        \m_{k-1} & (\mathbf{0},-1) & \m_{k-1}^{-j}
        \end{array}\right) & \n_k = \m_k^{-i},\, \n_{k-1} = \m_{k-1}^{-j}\, \text{for some $i \in [k]$, $j \in [k-1]$}  \\
        \left(\begin{array}{cc|c}
        \m_k & (\mathbf{0},-1) & \m_k^{-i}, \\
        \m_{k-1} & (\mathbf{0},0) & \m_{k-1}^{-j}
        \end{array}\right) & \n_k = \m_k^{-i},\, \text{for some $i \in [k]$},  \\
       1 & \m_k = \n_k,\, \m_{k-1} = \n_{k-1}, \\
       0 &\text{otherwise}.
    \end{cases}
\end{equation}
If either $\m_{k-1} \squb \m_k$ or $\n_{k-1} \squb \n_k$ is not satisfied then we define $ \rwpm{\m_k}{\n_k}{\m_{k-1}}{\n_{k-1}} \defeq 0$.

\section{Proof of \cref{thm:main}}\label{proof_main_gt_theorem}

%%%%%%%%%%%%%%%%%%%%%%%%%%%%%%%%%%%%%%%%%%%%%%%%%%%%%%%%%%%%%%%%%%%%%%%%%%%%%%%%%%%%%%%%%%%%%%%%%%%%%%%%%%%%%%%%%%%%%%%%%%%%%%%%%%%%%%%%%%%%%%%%%%%%%%%%%%%%%%%

%%%%%%%%%%%%%%%%%%%%%%%%%%%%%%%%%%%%%%%%%%%
%%%%%%%%%%%%%%%%%%%%%%%%%%%%%%%%%%%%%%%%%%%

% Restating the walled Brauer algebra definition
\brauerdef*

%%%%%%%%%%%%%%%%%%%%%%%%%%%%%%%%%%%%%%%%%%%
%%%%%%%%%%%%%%%%%%%%%%%%%%%%%%%%%%%%%%%%%%%

% Restating the main theorem
\mainthm*

%%%%%%%%%%%%%%%%%%%%%%%%%%%%%%%%%%%%%%%%%%%
%%%%%%%%%%%%%%%%%%%%%%%%%%%%%%%%%%%%%%%%%%%

\begin{proof}
    We will prove the theorem in three steps. First, we check that such action defines a representation of $\A_{p,q}^d$ by checking the relations in \cref{rel:transpositions} for transpositions. The relations for transpositions in \cref{gtbasis:transpositions} are defined in the same way as the Young--Yamanouchi basis of the symmetric group \cite{rutherford2013substitutional}. It is essentially folklore knowledge today, however, we still provide the proof for completeness. Next, we check the relations in \cref{rel:contraction1,rel:contraction2,rel:contraction3} for the contraction $\sigma_p$. Finally, we prove that such representation is irreducible.
    For simplicity, we will write $\sigma_i$ instead of $\psi_\lambda(\sigma_i)$.

    (a) To verify $\sigma_i^2=1$, consider the action of $\sigma_i$ on the invariant subspaces $V_T$ spanned by $\set{\ket{T}, \ket{\sigma_i T}}$ for each $T \in \Paths(\lambda)$. It is clear from \cref{gtbasis:transpositions} that the matrix $\sigma_i|_{V_T}$ of this action is
    \begin{equation*}
        \restr{\sigma_i}{V_T} =
        \mx{
          \frac{1}{r_i(T)} & \sqrt{1 - \frac{1}{r_i(T)^2}}  \\
          \sqrt{1 - \frac{1}{r_i(T)^2}} & -\frac{1}{r_i(T)}
        },
    \end{equation*}
    meaning that trivially $(\sigma_i|_{V_T})^2=1$. Since that holds for every $T \in \Paths(\lambda)$, then $\sigma_i^2=1$ holds.

    (b) To verify $(\sigma_i \sigma_{i+1})^3=1$, consider for every $T \in \Paths(\lambda)$ the action of $\sigma_i \sigma_{i+1}$ (according to \cref{gtbasis:transpositions}) on the invariant vector space $V_T \defeq \spn \set{ \ket{T}, \ket{\sigma_i T}, \ket{\sigma_{i+1} \sigma_i T}, \ket{ \sigma_i \sigma_{i+1} \sigma_i T}, \ket{ (\sigma_{i+1} \sigma_i)^2 T}, \ket{ \sigma_i (\sigma_{i+1} \sigma_i)^2 T}}$. Now if we define $a \defeq r_i(T), \, b \defeq r_i\of*{ \sigma_{i+1}\sigma_i T}, \, c \defeq r_i\of*{(\sigma_{i+1}\sigma_i)^2 T}$, then the action of $\sigma_i$ and $\sigma_{i+1}$ on the $V_T$ in the above basis is given by the following matrices:
    \begin{align*}
        \restr{\sigma_i}{V_T} \hspace{-0.2em} = \hspace{-0.2em} \of*{\begin{smallmatrix}
          \frac{1}{a} & \sqrt{1 - \frac{1}{a^2}} & 0 & 0 & 0 & 0 \\
          \sqrt{1 - \frac{1}{a^2}} & -\frac{1}{a} & 0 & 0 & 0 & 0 \\
          0 & 0 & \frac{1}{b} & \sqrt{1 - \frac{1}{b^2}} & 0 & 0 \\
          0 & 0 & \sqrt{1 - \frac{1}{b^2}} & -\frac{1}{b} & 0 & 0 \\
          0 & 0 & 0 & 0 & \frac{1}{c} & \sqrt{1 - \frac{1}{c^2}} \\
          0 & 0 & 0 & 0 & \sqrt{1 - \frac{1}{c^2}} & -\frac{1}{c} \\
        \end{smallmatrix}},
        \restr{\sigma_{i+1}}{V_T} \hspace{-0.2em} = \hspace{-0.2em} \of*{\begin{smallmatrix}
          \frac{1}{c} & 0 & \sqrt{1 - \frac{1}{c^2}} & 0 & 0 & 0 \\
          0 & \frac{1}{b} & 0 & 0 & \sqrt{1 - \frac{1}{b^2}} & 0 \\
          \sqrt{1 - \frac{1}{c^2}} & 0 & -\frac{1}{c} & 0 & 0 & 0 \\
          0 & 0 & 0 & \frac{1}{a} & 0 & \sqrt{1 - \frac{1}{a^2}} \\
          0 & \sqrt{1 - \frac{1}{b^2}} & 0 & 0 & -\frac{1}{b} & 0 \\
          0 & 0 & 0 & \sqrt{1 - \frac{1}{a^2}} & 0 & -\frac{1}{a} \\
        \end{smallmatrix}}
    \end{align*}
    Taking into account the fact $b = a + c$, it is easy to verify that $ (\sigma_i|_{V_T} \sigma_{i+1}|_{V_T} )^3=1$. Since this holds for any $T \in \Paths(\lambda)$, it must be $(\sigma_i \sigma_{i+1})^3=1$.

    (c) Finally, to verify the relation $ \sigma_i \sigma_j = \sigma_j \sigma_i$ for $(|i-j| > 1)$ just note, that $\sigma_i \sigma_j T = \sigma_j \sigma_i T$ and $a \defeq r_i(\sigma_j T) = r_i(T)$, $b \defeq r_j(\sigma_i T) = r_j(T)$. It means that on $W_T \defeq \spn \set{\ket{T}, \ket{\sigma_i T}, \ket{\sigma_j T}, \ket{\sigma_j \sigma_i T}}$ we have a tensor product structure:
   \begin{align*}
    \restr{\sigma_i}{W_T} &=
    \of*{\begin{smallmatrix}
      \frac{1}{a} & \sqrt{1 - \frac{1}{a^2}} & 0 & 0 \\
      \sqrt{1 - \frac{1}{a^2}} & -\frac{1}{a} & 0 & 0 \\
      0 & 0 & \frac{1}{a} & \sqrt{1 - \frac{1}{a^2}} \\
      0 & 0 & \sqrt{1 - \frac{1}{a^2}} & -\frac{1}{a} \\
    \end{smallmatrix}} = I_2 \otimes
    \of*{\begin{smallmatrix}
      \frac{1}{a} & \sqrt{1 - \frac{1}{a^2}}  \\
      \sqrt{1 - \frac{1}{a^2}} & -\frac{1}{a} \\
    \end{smallmatrix}}, \\
    \restr{\sigma_j}{W_T} &=
    \of*{\begin{smallmatrix}
      \frac{1}{b} & 0 & \sqrt{1 - \frac{1}{b^2}} & 0  \\
      0 & \frac{1}{b} & 0  & \sqrt{1 - \frac{1}{b^2}} \\
      \sqrt{1 - \frac{1}{b^2}} & 0 & -\frac{1}{b} & 0  \\
      0 & \sqrt{1 - \frac{1}{b^2}} & 0 & -\frac{1}{b} \\
    \end{smallmatrix}} =
    \of*{\begin{smallmatrix}
      \frac{1}{b} & \sqrt{1 - \frac{1}{b^2}}  \\
      \sqrt{1 - \frac{1}{b^2}} & -\frac{1}{b} \\
    \end{smallmatrix}} \otimes I_2,
    \end{align*}
    and consecutively $ \sigma_i|_{V_T} \sigma_j|_{V_T} = \sigma_j|_{V_T} \sigma_i|_{V_T}$. Therefore, $\sigma_i \sigma_j = \sigma_j \sigma_i$ when $|i-j| > 1$.

    (d) For each $T \in \Paths(\lambda)$ there is an invariant subspace $V_T \defeq \spn \set{\ket{T'} \, | \, T' \in \M{T}}$. If $\M{T} = \0$, then we assume $V_T \defeq \spn \set{\ket{T}}$.
    Note that $\norm{\ket{v_T}}^2_2 = d$, according to \cref{lem:v_Tnorm}. Moreover, it is easy to see from the definition that $\sigma_p|_{V_T} = \ketbra{v_T}{v_T}$. From this it is obvious $\of*{\sigma_p|_{V_T}} ^ 2 = d \cdot \sigma_p|_{V_T}$, and that implies $\sigma_p ^ 2 = d \cdot \sigma_p$.

    (f) To check $\sigma_p \sigma_i = \sigma_i \sigma_p \, (i \neq p \pm 1)$, we define $W_T \defeq \spn \set{\ket{T'} \, | \, T' \in \M{T} \cup \M{\sigma_i T}} \simeq V_T \otimes \spn \set{ \ket{1}, \ket{\sigma_i}}$, where $V_T \defeq \spn \set*{ \ket{T'} \, | \, T' \in \M{T}}$. Here we have a similar tensor product structure as in the case of transpositions making these generators commute. Namely, since $r_i(T) = r_i(T')$ for every $T' \in \M{T}$, we have
    \begin{equation*}
        \restr{\sigma_p}{W_T} = d \of*{ \restr{\ketbra{v_T}{v_T}}{V_T} } \otimes I_2, \,
        \restr{\sigma_i}{W_T} =  I_{\abs{\M{T}}} \otimes
        \of*{\begin{matrix}
            \frac{1}{r_i(T)} & \sqrt{1 - \frac{1}{r_i(T)^2}}  \\
            \sqrt{1 - \frac{1}{r_i(T)^2}} & -\frac{1}{r_i(T)} \\
        \end{matrix}}.
    \end{equation*}
    Therefore $\sigma_p \sigma_i = \sigma_i \sigma_p \, (i \neq p \pm 1)$ holds.

    (e) Now let's first check the relation $\sigma_p \sigma_{p-1} \sigma_p = \sigma_p$. Note that we can conveniently write generators $\sigma_p$ and $\sigma_{p-1}$ in terms of the path algebra matrix units, specifically highlighting only the relevant ones:
    \begin{align}
        \sigma_p &= \sum_{\mu \in B(\lambda)} \sum_{\nu \in C(\lambda,\mu)} \sum_{ \substack{S_1 \in \Paths(\nu) \\ S_2 \in \Paths(\mu,\lambda) \\ m, m' \in \M{\lambda,\mu,\nu}}} c(\mu,m) c(\mu,m') \ketbra{S_1 \circ(\nu,\mu,m,\mu)\circ S_2}{S_1 \circ(\nu,\mu,m',\mu)\circ S_2} \label{thm:e:sigma_p}\\
        \sigma_{p-1} &= \sum_{\mu \in B(\lambda)} \sum_{\nu \in C(\lambda,\mu)} \sum_{ \substack{S_1 \in \Paths(\nu) \\ S_2 \in \Paths(\mu,\lambda) \\ m \in \M{\lambda,\mu,\nu}}} f(\nu,\mu,m) \ketbra{S_1 \circ(\nu,\mu,m,\mu)\circ S_2}{S_1 \circ(\nu,\mu,m,\mu)\circ S_2} + \dotsc, \label{thm:e:sigma_p-1}
    \end{align}
    where we use the following notation:
    \begin{align*}
        B(\lambda) &\defeq \set{\mu \in \Irr{\A_{p-1,0}^{d}} \, | \, \exists \, T \in \Paths(\lambda) : T^{p-1} = T^{p+1} = \mu},
        \\
        C(\lambda,\mu) &\defeq \set{\nu \in \Irr{\A_{p-2,0}^{d}} \, | \, \exists \, T \in \Paths(\lambda) : T^{p-1} = T^{p+1} = \mu, T^{p-2} = \nu} \\
        \M{\lambda,\mu,\nu} &\defeq \set{m \in \Irr{\A_{p,0}^{d}} \, | \, \exists \, T \in \Paths(\lambda) : T^{p-1} = T^{p+1} = \mu, T^{p-2} = \nu, T^p = m},\\
        c(\mu,m) &\defeq c(T) \text{ for arbitrary $T \in \Paths(\lambda) : T^{p-1}=T^{p+1}=\mu, T^p=m$}, \\
        f(\nu,\mu,m) &\defeq \frac{1}{r_{p-1}(T)} \text{ for arbitrary $T \in \Paths(\lambda) : T^{p-2}=\nu, T^{p-1}=\mu, T^p=m$}.
    \end{align*}
    In \cref{thm:e:sigma_p-1} we do not write terms with the matrix units which multiply to zero with the matrix units from the sum of \cref{thm:e:sigma_p}. We also abuse the notation by forgetting that $\nu, \mu, m$ are actually pairs of Young diagrams: we only refer to the left diagrams by dropping the subscript $l$. Using \cref{thm:e:sigma_p,thm:e:sigma_p-1} we can deduce by direct multiplication, that $\sigma_p \sigma_{p-1} \sigma_{p} = \sigma_{p}$ is equivalent to
    \begin{equation}
        \sum_{m \in \M{\lambda,\mu,\nu}} f(\nu,\mu,m) \cdot c(\mu,m)^2 = 1
        \label{thm_e:main_eq}
    \end{equation}
    for every $\mu \in  B(\lambda)$, $\nu \in  C(\lambda,\mu)$, $S_1 \in \Paths(\nu) $, $S_2 \in \Paths(\mu,\lambda)$. Let $c \defeq \mu \backslash \nu$ be the cell containing $p-1$, it is a corner cell of $\mu$, i.e. $c \in RC(\mu)$. Then \cref{thm_e:main_eq} is equivalent to
    \begin{align}
        \sum_{m \in AC(\mu)} \frac{d+\cont(m)}{\cont(m)-\cont(c)} \frac{ \prod_{v \in R C(\mu)} \of{\cont(m)-\cont(v)}}{\prod_{v \in A C(\mu) \setminus m} \of{\cont(m)-\cont(v)}} = 1
    \end{align}
    for every $\mu \in  B(\lambda)$ and $c \in RC(\mu)$. By rewriting the previous formula as
    \begin{align*}
        d \cdot \sum_{m \in AC(\mu)} \frac{ \prod_{w \in R C(\mu) \setminus c} \of{\cont(m)-\cont(w)}}{\prod_{v \in A C(\mu) \setminus m} \of{\cont(m)-\cont(v)} } + \sum_{m \in AC(\mu)} \cont(m) \frac{ \prod_{w \in R C(\mu) \setminus c} \of{\cont(m)-\cont(w)}}{\prod_{v \in A C(\mu) \setminus m} \of{\cont(m)-\cont(v)}} = 1,
    \end{align*}
    and using \cref{cor:ratios} we conclude that \cref{thm_e:main_eq} holds, finishing the proof of $\sigma_p \sigma_{p-1} \sigma_p = \sigma_p$. Similar proof also works for the relation (f) $\sigma_p \sigma_{p+1} \sigma_p = \sigma_p$ which we do not repeat here.

    (g) Finally, checking the relations in \cref{rel:contraction2} is the same in spirit as for (e), but more cumbersome. Let's first write the generators in terms of matrix units, specifically highlighting only the relevant ones:
    \begin{alignat}{2}
        \sigma_p &= \sum_{\substack{\mu \in B(\lambda) \\ (\nu_1,\nu_2) \in C(\lambda,\mu)}}  \sum_{ \substack{S_1 \in \Paths_{p-2}(\nu_1) \\ S_2 \in \Paths_{p+2}(\nu_2,\lambda) \\ m, m' \in \M{\lambda,\nu_2,\mu,\nu_1}}} c(\mu,m) c(\mu,m') \ketbra{S_1 \circ(\nu_1,\mu,m,\mu,\nu_2)\circ S_2}{S_1 \circ(\nu_1,\mu,m',\mu,\nu_2)\circ S_2} \label{thm:g:sigma_p}\\
        \sigma_{p-1} &= \sum_{\substack{\mu \in B(\lambda) \\ (\nu_1,\nu_2) \in C(\lambda,\mu)}}  \sum_{ \substack{S_1 \in \Paths_{p-2}(\nu_1) \\ S_2 \in \Paths_{p+2}(\nu_2,\lambda) \\ m \in \M{\lambda,\nu_2,\mu,\nu_1}}}
        \Big( f_{p-1}(\nu_1,\mu,m) \ketbra{S_1 \circ(\nu_1,\mu,m,\mu,\nu_2)\circ S_2}{S_1 \circ(\nu_1,\mu,m,\mu,\nu_2)\circ S_2} + \nonumber\\
        &\hspace{+2cm} + \sqrt{1-f_{p-1}^2(\nu_1,\mu,m)} \ketbra{S_1 \circ(\nu_1,\sigma_{p-1} \mu,m,\mu,\nu_2)\circ S_2}{S_1 \circ(\nu_1,\mu,m,\mu,\nu_2)\circ S_2}\Big) + \dotsc \label{thm:g:sigma_p-1}
        \\
        \sigma_{p+1} &= \sum_{\substack{\mu \in B(\lambda) \\ (\nu_1,\nu_2) \in C(\lambda,\mu)}}  \sum_{ \substack{S_1 \in \Paths_{p-2}(\nu_1) \\ S_2 \in \Paths_{p+2}(\nu_2,\lambda) \\ m \in \M{\lambda,\nu_2,\mu,\nu_1}}} \Big( f_{p+1}(m,\mu,\nu_2) \ketbra{S_1 \circ(\nu_1,\mu,m,\mu,\nu_2)\circ S_2}{S_1 \circ(\nu_1,\mu,m,\mu,\nu_2)\circ S_2} + \nonumber \\
        &\hspace{+2cm} + \sqrt{1-f_{p+1}^2(m,\mu,\nu_2)} \ketbra{S_1 \circ(\nu_1,\mu,m,\sigma_{p+1}\mu,\nu_2)\circ S_2}{S_1 \circ(\nu_1,\mu,m,\mu,\nu_2)\circ S_2}\Big) + \dotsc \label{thm:g:sigma_p+1}
    \end{alignat}
    where we use the following notation, similar to (e):
    \begin{align*}
        B(\lambda) &\defeq \set{\mu \in \Irr{\A_{p-1,0}^{d}} \, | \, \exists \, T \in \Paths(\lambda) : T^{p-1} = T^{p+1} = \mu},
        \\
        C(\lambda,\mu) &\defeq \set{(\nu_1,\nu_2) \in \Irr{\A_{p-2,0}^{d}} \times \Irr{\A_{p,2}^{d}}  \, | \, \exists \, T \in \Paths(\lambda) :  T^{p-2} = \nu_1, T^{p-1} = T^{p+1} = \mu, T^{p+2} = \nu_2} \\
        \M{\lambda,\nu_2,\mu,\nu_1} &\defeq \set{m \in \Irr{\A_{p,0}^{d}} \, | \, \exists \, T \in \Paths(\lambda) : T^{p-2} = \nu_1, T^{p-1} = T^{p+1} = \mu, T^{p+2} = \nu_2, T^{p} = m},\\
        c(\mu,m) &\defeq c(T) \text{ for arbitrary $T \in \Paths(\lambda)$ s.t. $T^{p-1}=T^{p+1}=\mu$, $T^{p}=m$}, \\
        f_{p-1}(\nu_1,\mu,m) &\defeq \frac{1}{r_{p-1}(T)} \text{ for arbitrary $T \in \Paths(\lambda) : T^{p-2}=\nu_1, T^{p-1}=\mu, T^{p}=m$}, \\
        f_{p+1}(\mu,m,\nu_2) &\defeq \frac{1}{r_{p+1}(T)} \text{ for arbitrary $T \in \Paths(\lambda) :  T^{p}=m, T^{p+1}=\mu, T^{p+2}=\nu_2$}.
    \end{align*}
    With this notation we deduce by direct multiplication that the condition (g) is equivalent to the following statement. Namely, for every $\mu \in B(\lambda), (\nu_1,\nu_2) \in C(\lambda,\mu), S_1 \in \Paths(\nu_1), S_2 \in \Paths(\nu_2), m_1,m_2 \in \M{\lambda,\nu_2,\mu,\nu_1}:$
    \begin{gather}
        c(\mu,m_1) c(\mu,m_2) \of*{f_{p-1}(\nu_1,\mu,m_2)-f_{p+1}(m_2,\mu,\nu_2)}  \of[\Bigg]{\sum_{m \in \M{\lambda,\nu_2,\mu,\nu_1}} c^2(\mu,m)f_{p-1}(\nu_1,\mu,m)f_{p+1}(m,\mu,\nu_2)} = 0,
        \label{thm:g:main1}
    \end{gather}
    \begin{gather}
        \delta_{\mu',\sigma_{p-1}\mu} \delta_{\mu',\sigma_{p+1}\mu}\delta_{\nu_1,\nu_2}c(\mu,m_1) c(\mu,m_2) \of*{f_{p-1}(\nu_1,\mu',m_2)-f_{p+1}(m_2,\mu',\nu_2)} \cdot \nonumber \\ \cdot \of[\Bigg]{\sum_{m \in \M{\lambda,\nu_2,\mu,\nu_1}} c(\mu,m)c(\mu',m) \sqrt{1-f_{p-1}^2(\nu_1,\mu,m)}\sqrt{1-f_{p+1}^2(m,\mu,\nu_2)}} = 0,
        \label{thm:g:main2}
    \end{gather}
    \begin{gather}
        \delta_{\mu',\sigma_{p-1}\mu} \delta_{\mu',\sigma_{p+1}\mu}\delta_{\nu_1,\nu_2}c(\mu,m_1) c(\mu,m_2) \of*{\sqrt{1-f_{p-1}^2(\nu_1,\mu',m_2)}-\sqrt{1-f_{p+1}^2(m_2,\mu',\nu_2)}} \cdot \nonumber \\ \cdot \of[\Bigg]{\sum_{m \in \M{\lambda,\nu_2,\mu,\nu_1}} c(\mu,m)c(\mu',m) \sqrt{1-f_{p-1}^2(\nu_1,\mu,m)}\sqrt{1-f_{p+1}^2(m,\mu,\nu_2)}} = 0.
        \label{thm:g:main3}
    \end{gather}
    \cref{thm:g:main2,thm:g:main3} hold because when $\nu_1=\nu_2$ then $f_{p-1}(\nu_1,\mu',m_2) = f_{p+1}(m_2,\mu',\nu_2)$. Moreover, when $\nu_1=\nu_2$ then $f_{p-1}^2(\nu_1,\mu,m_2) = f_{p+1}^2(m_2,\mu,\nu_2)$ and \cref{thm:g:main1} trivially holds. Finally, in the case $\nu_1 \neq \nu_2$ \cref{thm:g:main1} also holds because:
    \begin{equation}
        \sum_{m \in \M{\lambda,\nu_2,\mu,\nu_1}} c^2(\mu,m)f_{p-1}(\nu_1,\mu,m)f_{p+1}(m,\mu,\nu_2) = 0,
    \end{equation}
    which follows from \cref{lem:magich} by a similar technique which was used to show \cref{thm_e:main_eq} in \cref{cor:ratios}. Similar proof also works for the relation (h) which we do not repeat here.

    Finally, we need to show the irreducibility of our representation. Due to \cref{lem:JMaction} the spectrum of Jucys--Murphy elements in our basis coincides with the canonical definition of the action of Jucys--Murphy elements in the Gelfand--Tsetlin basis \cite{doty2019canonical,grinko2022linear}. Since Jucys--Murphy elements generate a maximal commutative subalgebra of $\A_{p,q}^d$ their action uniquely determines the basis. Therefore our basis coincides with the Gelfand--Tsetlin basis from \cite{doty2019canonical,grinko2022linear}, which is originally defined for irreducible representations.\footnote{Alternatively, one can use \Cref{prop:irreducibility} together with \cref{lem:JMaction} and the results from \cite{doty2019canonical,grinko2022linear}.}
\end{proof}

\begin{lemma}\label{lem:content_into_sym_u_dims}
    For every $\mu \pt p$ and $\lambda \pt p-1$, such that $\lambda \rightarrow \mu$ there holds
    \begin{equation}
        d + \cont(\mu \backslash \lambda) = p \cdot \frac{d_\lambda}{d_\mu} \cdot \frac{m_\mu}{m_\lambda}
    \end{equation}
    where $d_\lambda$ and $m_{\lambda}$ are dimensions of the symmetric and unitary groups irreducible representation $\lambda$.
\end{lemma}

\begin{proof}
    Use the hook length formula \eqref{eq:hook length formula} for $d_\lambda$ and the hook-content formula \eqref{eq:hook content formula} for $m_{\lambda}$:
    \begin{equation}
        p \cdot \frac{d_\lambda}{d_\mu} \cdot \frac{m_{\mu,d}}{m_{\lambda}} = p \cdot \frac{\tfrac{(p-1)!}{\prod_{u \in \lambda} h(u)}}{\tfrac{p!}{\prod_{u \in \mu} h(u)}} \frac{\tfrac{\prod_{u \in \mu} (d + \cont(u))}{\prod_{u \in \mu} h(u)}}{\tfrac{\prod_{u \in \lambda} (d + \cont(u))}{\prod_{u \in \lambda} h(u)}} = d + \cont(\mu \backslash \lambda).
    \end{equation}
\end{proof}

\begin{lemma}[{\cite{kosuda2003new}}]\label{lem:ration_into_u_dims}
    For every $\mu \pt p$ and $\lambda \pt p-1$, such that $\lambda \rightarrow \mu$, we denote the added cell by $x \defeq \mu \backslash \lambda$. Then there holds
    \begin{equation}
        \frac{\prod_{a \in AC(\lambda) \setminus x} \of{\cont(x)-\cont(a)}}{ \prod_{c \in RC(\lambda)} \of{\cont(x)-\cont(c)}} = p \cdot \frac{d_\lambda}{d_\mu}
    \end{equation}
    where $d_\lambda$ is the dimension of the symmetric group irrep $\lambda$.
\end{lemma}

\begin{corollary}
    For every $T$ such that $\M{T} \neq \0$ there holds $\norm{\ket{v_T}}^2_2 = d$.
    \label{lem:v_Tnorm}
\end{corollary}

\begin{proof}
    Denote $\lambda \defeq T_l^{{p-1}}$ and use \cref{lem:content_into_sym_u_dims,lem:ration_into_u_dims}:
    \begin{equation}
        \norm{\ket{v_T}}^2_2 = \sum_{S \in \M{T}} c(S)^2 = \sum_{\substack{\mu \in \Irr{\A_{p,0}^{d}} \\ \mu : \lambda \rightarrow \mu}} \frac{m_\mu}{ m_\lambda} = d,
    \end{equation}
    where the last equality is due to Pieri's rule for the irreducible representations of the unitary group.
\end{proof}

\begin{lemma}[{\cite{waring1779vii,louck1970canonical,chen1996interpolation}}] \label{lem:magich}
    For each integer $k \geq 0$, we have
    \begin{equation}
        \sum_{i=1}^n \frac{x_i^k}{\prod_{i \neq j} (x_i - x_j)} = h_{k-n+1}(x_1,\dotsc,x_n),
    \end{equation}
    where $h_r$ is complete $r$-homogeneous symmetric polynomial, which is defined as $h_r = 0$ for $r<0$ and $h_0 = 1$.
\end{lemma}

\begin{corollary}
    For every Young diagram $\lambda$ and arbitrary $c \in RC(\lambda)$ the following holds:
    \begin{align*}
        \sum_{m \in AC(\lambda)} \frac{ \prod_{w \in RC(\lambda) \setminus c} \of{\cont(m)-\cont(w)}}{\prod_{v \in AC(\lambda) \setminus m} \of{\cont(m)-\cont(v)}} = 0, \quad
        \sum_{m \in AC(\lambda)} \cont(m) \frac{ \prod_{w \in RC(\lambda) \setminus c} \of{\cont(m)-\cont(w)}}{\prod_{v \in AC(\lambda) \setminus m} \of{\cont(m)-\cont(v)}} = 1
    \end{align*}
    \label{cor:ratios}
\end{corollary}

\begin{proof}
    Note that the degree of the numerator as polynomial in $\cont(m)$ is $\abs{AC(\mu)}-2$ in the first case and $\abs{AC(\mu)}-1$ in the second case. Now just apply \cref{lem:magich} with $n = \abs{AC(\mu)}$ and variables being $\set{ \cont(m) \, | \, m \in AC(\mu)}$.
\end{proof}

\begin{proposition}\label{prop:irreducibility}
    If there is a basis $B$ for the representation $\rho$ of the finite-dimensional associative semisimple algebra $\A$ such that:
    \begin{enumerate}
        \item $\forall \, \ket{S},\ket{T} \in B$ $\exists \, b \in \A : \rho(b) \ket{T} = \ket{S} + \ket{v}$, where $\ket{v} \in \spn \set{B \setminus \ket{S}}$,
        \item $\forall \, \ket{S},\ket{T} \in B \,\, \exists \, a_S, a_T \in \A$ such that $ \rho(a_S) = \ketbra{S}{S},\, \rho(a_T) = \ketbra{T}{T} \in \End(V)$,
    \end{enumerate}
    then $\rho(a_S b a_T) = \ketbra{S}{T}$ and, consecutively, this representation is irreducible.
\end{proposition}

\begin{proof}
    The conclusion $\rho(a_S b a_T) = \ketbra{S}{T}$ is trivial, since for every $\ket{T'} \in B$:
    \begin{equation}
        \rho(a_S b  a_T) \ket{T'} =  \rho(a_S) \rho(b) \rho(a_T) \ket{T'} = \ketbra{S}{S} \rho(b) \ket{T} \braket{T}{T'} = \delta_{T,T'} \ketbra{S}{S} \of{\ket{S} + \ket{v}} = \delta_{T,T'} \ket{S}.
    \end{equation}
    Then the representation must be irreducible because there are no invariant subspaces. This result is also known as Burnside’s theorem on matrix algebras.
\end{proof}

\begin{lemma}\label{lem:JMaction}
    The action of Jucys--Murphy elements is diagonal in the Gelfand--Tstelin basis and the spectrum is given by the walled content $\wcont_k(T)$:
    \begin{equation}
        J_k \ket{T} = \wcont_k(T) \ket{T}
    \end{equation}
\end{lemma}

\begin{proof}
    For $k \leq p$ this is just a statement of the same result for the symmetric group, e.g. see \cite{OV1996,VO2005}. The proof for that case can be done by induction by exploiting the relation $J_{k+1} = \sigma_k J_k \sigma_k + \sigma_k$. The base of the induction is trivial. Now using the induction step we can immediately see that
    \begin{align}
        J_{k+1}\ket{T} &= \of*{\sigma_k J_k \sigma_k + \sigma_k} \ket{T} =  \of*{\sigma_k J_k  + 1} \sigma_k \ket{T} = \of*{\sigma_k J_k  + 1} \of*{\tfrac{1}{r_k(T)}\ket{T}+\sqrt{1-\tfrac{1}{r_k(T)^2}}\ket{\sigma_k T}} \nonumber \\
        &= \tfrac{1}{r_k(T)}\ket{T}+\sqrt{1-\tfrac{1}{r_k(T)^2}}\ket{\sigma_k T} + \sigma_k \of*{ \tfrac{\wcont_k(T)}{r_k(T)} \ket{T} + \wcont_k(\sigma_k T)\sqrt{1-\tfrac{1}{r_k(T)^2}} \ket{\sigma_k T} } \nonumber \\
        &= \of*{\tfrac{\wcont_k(T)}{r_k(T)}+1} \of*{\tfrac{\ket{T}}{r_k(T)}+\sqrt{1-\tfrac{1}{r_k(T)^2}}\ket{\sigma_k T}} + \wcont_k(\sigma_k T)\sqrt{1-\tfrac{1}{r_k(T)^2}}\of*{-\tfrac{\ket{\sigma_k T}}{r_k(T)}+\sqrt{1-\tfrac{1}{r_k(T)^2}}\ket{T}} \nonumber\\
        &= \of*{\tfrac{1}{r_k(T)} + \tfrac{\wcont_k(T)-\wcont_{k+1}(T)}{r_k(T)^2} + \wcont_{k+1}} \ket{T} +\sqrt{1-\tfrac{1}{r_k(T)^2}} \of*{1+\tfrac{\wcont_k(T)-\wcont_{k+1}(T)}{r_k(T)}} \ket{\sigma_k T} \nonumber\\
        &= \wcont_k(T) \ket{T},
    \end{align}
    where we used that $\wcont_k(\sigma_k T) = \wcont_{k+1}(T), r_k(\sigma_k T)=-r_k(T)$ and $r_k(T) = \wcont_{k+1}(T)-\wcont_{k}(T)$.

    Similarly, for $k \geq p+1$ we have a similar relation $J_{k+1} = \sigma_k J_k \sigma_k + \sigma_k$ and the same argument holds, assuming that $J_{p+1}\ket{T} = \wcont_{p+1}(T) \ket{T}$. However, for $k = p+1$ we need to prove the claim separately.

    To show the claim for $k = p+1$, recall that $J_{p+1} = d - \rho$, where $\rho \defeq \sum_{i=1}^p (i,p)\sigma_p (i,p)$ and $(i,p)$ is a transposition between sites $i$ and $p$ with the convention $(p,p) \defeq 1$. Without loss of generality assume that $q=1$. Note that since $J_{p+1}$ commutes with $\A_{p,0}^d$, $J_{p+1}$ must be diagonal in the Gelfand--Tsetlin basis and it has the same eigenvalues for every path $T$ which goes through a given vertex $\mu$ at the level $p$ in the Bratteli diagram. Therefore, we now assume $T \in \Paths(\lambda)$ has the property $T^{p} = \mu$. There are two cases:
        \begin{enumerate}
            \item If $\lambda_r = (1)$, then there is no mobile cell in $T$ and the action of $\sigma_p$ is zero, meaning that $J_{p+1} \ket{T} = d \ket{T}$ which is consistent with $J_{p+1} \ket{T} = \wcont_{p+1}(T) \ket{T}$ for that case.
            \item If $\lambda_r = \0$, then for every $T \in \Paths(\lambda)$ with $T^{p} = \mu$ we can write:
            \begin{align}
                \bra{T} \rho \ket{T} &= \frac{1}{d_\mu} \sum_{\substack{T \in \Paths(\lambda) \\ T^{p} = \mu}} \bra{T} \rho \ket{T} = \frac{1}{d_\mu} \sum_{\substack{T \in \Paths(\lambda) \\ T^{p} = \mu}} \sum_{i=1}^p \bra{T} (i,p) \sigma_p (i,p) \ket{T} \nonumber \\
                &= \frac{1}{d_\mu} \sum_{i=1}^p \sum_{\substack{T \in \Paths(\lambda) \\ T^{p} = \mu}} \bra{T} (i,p) \of[\bigg]{\sum_{\substack{S \in \Paths(\lambda) \\ S^{p} = \mu \\ S^{p-1} = \lambda}} \ketbra{v_S}{v_S}} (i,p) \ket{T}
                = \frac{1}{d_\mu} \sum_{i=1}^p \sum_{\substack{T,S \in \Paths(\lambda) \\ T^{p} = S^{p} = \mu \\ S^{p-1} = \lambda}} \abs{\bra{v_S} (i,p) \ket{T}}^2 \nonumber \\
                &= \frac{1}{d_\mu} \sum_{i=1}^p \sum_{\substack{T,S \in \Paths(\lambda) \\ T^{p} = S^{p} = \mu \\ S^{p-1} = \lambda}} c(S)^2 \cdot \abs{\bra{S} (i,p) \ket{T}}^2
                = \frac{c(\lambda, \mu)^2}{d_\mu} \sum_{i=1}^p \sum_{\substack{T,S \in \Paths(\lambda) \\ T^{p} = S^{p} = \mu \\ S^{p-1} = \lambda}} \abs{\bra{S} (i,p) \ket{T}}^2 \nonumber \\
                &= \frac{c(\lambda, \mu)^2}{d_\mu} \sum_{i=1}^p \sum_{\substack{S \in \Paths(\lambda) \\ S^{p} = \mu \\ S^{p-1} = \lambda}} \bra{S} (i,p) \of[\bigg]{\sum_{\substack{T \in \Paths(\lambda) \\ T^{p} = \mu}} \ketbra{T}{T}} (i,p)\ket{S} \nonumber \\
                &= \frac{c(\lambda, \mu)^2}{d_\mu} \sum_{i=1}^p \sum_{\substack{S \in \Paths(\lambda) \\ S^{p} = \mu \\ S^{p-1} = \lambda}} \bra{S} \of[\bigg]{\sum_{\substack{T \in \Paths(\lambda) \\ T^{p} = \mu}} \ketbra{T}{T}} (i,p)^2 \ket{S} = \frac{c(\lambda, \mu)^2}{d_\mu} \sum_{i=1}^p \sum_{\substack{T,S \in \Paths(\lambda) \\ T^{p} = S^{p} = \mu \\ S^{p-1} = \lambda}} \abs{\braket{S}{T}}^2 \nonumber \\
                &= \frac{c(\lambda, \mu)^2}{d_\mu} \cdot p \cdot d_\lambda =  p \cdot \frac{d_\lambda}{d_\mu} \cdot \frac{m_\mu}{m_\lambda} = d + \cont(\mu \backslash \lambda),
            \end{align}
            where we used \cref{lem:content_into_sym_u_dims,lem:ration_into_u_dims}. Therefore
            \begin{equation}
                \bra{T}J_{p+1}\ket{T} = d - \bra{T} \rho \ket{T} = - \cont(\mu \backslash \lambda) = - \cont(T^{p} \backslash T^{p+1}) = \wcont_{p+1}(T).
            \end{equation}
        \end{enumerate}
\end{proof}

\end{appendix}
\end{document}